\documentclass[11pt,usletter]{article}
\usepackage{jheppub}
\pdfoutput=1

\usepackage[section]{placeins}
\usepackage{amsthm}
\usepackage{graphicx}
\usepackage{hyperref}
\usepackage{tikz-cd}
\usepackage{color}

\def\ba{\begin{align}}
\def\ea{\end{align}}
\def\be{\begin{equation}}
\def\ee{\end{equation}}
\def\bea{\begin{eqnarray}}
\def\eea{\end{eqnarray}}
\usepackage{MnSymbol}
\usepackage{titlesec}
\theoremstyle{plain}
\newtheorem{thm}{Theorem}[section]

\theoremstyle{definition}

\theoremstyle{remark}

\begin{document}

\title{Spectral form factors and late time quantum chaos}

\author[a,b]{Junyu Liu}
 \affiliation[a]{Walter Burke
  Institute for Theoretical Physics, California Institute of
  Technology, Pasadena, California 91125, U.S.A}
   \affiliation[b]{Institute for Quantum Information and Matter,
California Institute of Technology, Pasadena, California 91125, U.S.A}
\emailAdd{jliu2@caltech.edu}

\abstract{This is a collection of notes that are about spectral form factors of standard ensembles in the random matrix theory, written for the practical usage of current study of late time quantum chaos. More precisely, we consider Gaussian Unitary Ensemble (GUE), Gaussian Orthogonal Ensemble (GOE), Gaussian Symplectic Ensemble (GSE), Wishart-Laguerre Unitary Ensemble (LUE),  Wishart-Laguerre Orthogonal Ensemble (LOE), and Wishart-Laguerre Symplectic Ensemble (LSE). These results and their physics applications cover a three-fold classification of late time quantum chaos in terms of spectral form factors.}

\hbox{CALT-TH-2018-028}

\maketitle
\flushbottom

\section{Overview}
The theory of quantum chaos, and its connection to random matrix theory, have several new developments recently on understanding novel behaviors of condensed matter system and the quantum nature of black hole physics. The definition of quantum chaos has various versions. Following the pioneer works done by Wigner \cite{Wigner} and Dyson \cite{Dyson}, people regard random matrix theory as a tool to classify a generic random Hamiltonian with discrete symmetries, and their energy spectra have been observed to satisfy universal behaviors \cite{Bohigas:1983er,Zirnbauer:1996zz,Altland:1997zz}. The scientific interests of random matrix theory varies from nonlinear science, mathematics and mathematical physics, to nuclear physics, statistical physics and quantum field theory. (See some books, for instance, \cite{book1,book2,book3,book4,book5}, for reference.) Some recent discoveries of black hole physics lead to interests in understanding of scrambling properties of quantum chaotic systems \cite{Hayden:2007cs, Sekino:2008he, Lashkari:2011yi}, where people start to consider an early chaotic behavior, the Lyaponov exponent appears in the out-of-time-ordered correlators of the large $N$ theory \cite{Shenker:2013pqa,Shenker:2014cwa}, which is bounded by temperature in thermal ensemble \cite{Maldacena:2015waa}. A concrete condensed matter model, the Sachedev-Ye-Kitaev (SYK) model \cite{Sachdev:1992fk,talk,Maldacena:2016hyu}, has been proposed to realize the chaotic properties that saturate the bound. Interestingly, one can also apply random matrix theory classification to the SYK model \cite{You:2016ldz,Garcia-Garcia:2016mno}. Moreover, the spectral form factor, namely, the analytic-continuted partition function correlations in SYK model, could be matched with the prediction of spectral form factor in random matrix theory \cite{Cotler:2016fpe}. Some further investigations show that the spectral form factor is one of the key roles serving in several quantum chaotic systems, and could connect to out-of-time-ordered correlators and some other chaotic diagnostics \cite{Cotler:2017jue,Chen:2017yzn,Gharibyan:2018jrp}.
\\
\\
Those facts motivate us to study the spectral form factor in random matrix theory in detail and study its mathematical properties in detail from a modern chaotic physicist point of view, to build up such a tool box in general. In this paper, we are interested in mostly, the higher point spectral form factors and how to reach them in general from some building blocks. As an explicit example, we will describe the four point spectral form factors, which are mostly closed to the four point out-of-time-ordered correlators. 
\\
\\
From Dyson's classification, for Gaussian ensembles one could classify them by antiunitary symmetries as Gaussian Unitary Ensemble (GUE), Gaussian Orthogonal Ensemble (GOE) and Gaussian Symplectic Ensemble (GSE). For real systems like the SYK model, those ensembles often appear periodically in a list of number of sites. In this paper, we will consider all of them. Moreover, we will also discuss the Wishart-Laguerre ensembles with three symmetry classes. Those ensembles will correspond to supersymmetrized SYK models \cite{Li:2017hdt,Kanazawa:2017dpd,Hunter-Jones:2017crg}. As some examples of physics applications, we will comment on SYK model classifications, out-of-time-ordered-correlators and Page states.
\\
\\
This paper is organized as the following. In Section \ref{GUE} we will discuss the spectral form factor in GUE, the simplest symmetry class. In Section \ref{GOE} we will extend our discussions to GOE and GSE. In Section \ref{L} we will discuss the spectral form factor properties of the Wishart-Laguerre ensembles. In Section \ref{fg} we plot some figures for random matrix theory form factors to show their behaviors. In Section \ref{Ph} we put the collections of physics applications about spectral form factor in the random matrix theory. In Section \ref{conc}, we will arrive at the conclusion and discussion. 

\section{GUE spectral form factor}\label{GUE}
\subsection{Random matrix theory overview}
We consider GUE, the Gaussian Unitary Ensemble in this section. The ensemble is defined by introducing the following distribution function over space of Hermitian matrices $L\times L$,
\begin{align}
P(H) \propto \exp ( - \frac{L}{2}{\rm{Tr(}}{H^2}{\rm{)}})
\end{align}
which means that, for a Hermitian matrices $H$, the off-diagonal elements are independent complex random distributions following Gaussian distribution with mean $0$ and variance $1/L$, while the diagonal elements are independent real random distributions following Gaussian distribution with mean $0$ and variance $1/L$. From this form, one can observe that the GUE ensemble is invariant under a unitary transformation $H\to UHU^\dagger$.
\\
\\
One can also write the result in the eigenvalue basis, where one can show that the distribution over set of matrices could reduce to the distribution of eigenvalues with the following joint distribution
\begin{align}
P({\lambda _1},{\lambda _2} \ldots ,{\lambda _L}) = \exp ( - \frac{L}{2}\sum\limits_i {\lambda _i^2} )\prod\nolimits_{i < j}^{} {{{({\lambda _i} - {\lambda _j})}^2}} 
\end{align}
where $\lambda_i$s are eigenvalues. We could write done the measure of it more formally by defining Vandermonde determinant 
\begin{align}
\Delta (\lambda ) = \prod\nolimits_{i < j} {({\lambda _i} - {\lambda _j})} 
\end{align}
and we could formally write done the measure 
\begin{align}
P(\lambda )d\lambda  = D\lambda  = \exp ( - \frac{L}{2}\sum\limits_i {\lambda _i^2} )\Delta (\lambda )
\end{align}
Thus, based on this we could compute the $n$-point correlation function, where $n<L$ as
\begin{align}
{\rho ^{(n)}}({\lambda _1}, \ldots ,{\lambda _n}) = \int d {\lambda _{n + 1}} \ldots d{\lambda _L}P({\lambda _1}, \ldots ,{\lambda _L})
\end{align}
where we are going to integrate out all eigenvalues from $n+1$ to $L$. One might be interested in what is the result of the correlation function if we take the large $L$ limit. From random matrix theory, people find that the $n$ point function could be determined by a kernel $K$
\begin{align}
{\rho ^{(n)}}({\lambda _1}, \ldots ,{\lambda _n}) = \frac{{(L - n)!}}{{L!}}\det (K({\lambda _i},{\lambda _j}))_{i,j = 1}^n
\end{align}
where the kernel $K$, in the large $L$ limit, behaves as
\begin{align}
K({\lambda _i},{\lambda _j}) \equiv \left\{ \begin{array}{l}
\frac{L}{\pi }\frac{{\sin (L({\lambda _i} - {\lambda _j}))}}{{L({\lambda _i} - {\lambda _j})}}              {\text{         for }}i\not  = j\\
\frac{L}{{2\pi }}\sqrt {4 - \lambda _i^2}             {\text{             for }}i = j
\end{array} \right.
\end{align}
The kernel packages several information about random matrix theory in the large $L$ limit, where at the colliding case $i=j$, this kernel, as a one point function, serves as Wigner's semicircle law. While in the case where $i\ne j$, this kernel is called the sine kernel in random matrix theory, which is even universal in most standard ensembles beyond GUE. 
\\
\\
The main goal of this paper is to try to build up the technology on how to compute the Fourier transform of the $n$-point correlation functions, which is called the spectral form factor,
\begin{align}
{{\cal R}_{2k}}(t) = \sum\limits_{i,j} {\int {D\lambda } {e^{i({\lambda _{{i_1}}} +  \ldots  + {\lambda _{{i_k}}} - {\lambda _{{j_1}}} -  \ldots  - {\lambda _{{j_k}}})}}} 
\end{align}
where $k$ is any positive integer. We will start from our simplest example, the two point form factor 
\begin{align}
{{\cal R}_2}(t) = \sum\limits_{i,j} {\int {D\lambda } {e^{i({\lambda _i} - {\lambda _j})}}} 
\end{align}
and we will discuss how to compute higher points and finite temperature result. 
\subsection{Two point form factor}
\subsubsection{The disconnected piece}
We start to compute the two point form factor $\mathcal{R}_2$,
\begin{align}
&{{\cal R}_2}(t) = \sum\limits_{i,j} {\int {d{\lambda _i}d{\lambda _j}{\rho ^{(2)}}({\lambda _i},{\lambda _j})} {e^{i({\lambda _i} - {\lambda _j})t}}} \nonumber\\
&= L + L(L - 1)\int {d{\lambda _1}d{\lambda _2}{\rho ^{(2)}}({\lambda _1},{\lambda _2}){e^{i({\lambda _1} - {\lambda _2})t}}} 
\end{align}
By directly computing the determinant we have
\begin{align}
{\rho ^{(2)}}({\lambda _1},{\lambda _2}) = \frac{{{L^2}}}{{L(L - 1)}}\rho ({\lambda _1})\rho ({\lambda _2}) - \frac{{{L^2}}}{{L(L - 1)}}\frac{{{{\sin }^2}(L({\lambda _1} - {\lambda _2}))}}{{{{(L\pi ({\lambda _1} - {\lambda _2}))}^2}}}
\end{align}
where we define
\begin{align}
\rho (\lambda ) \equiv {\rho ^{(1)}}(\lambda )
\end{align}
which has been reduced to the Wigner semicircle
\begin{align}
\rho (\lambda ) \equiv {\rho ^{(1)}}(\lambda ) = \frac{1}{{2\pi }}\sqrt {4 - {\lambda ^2}} 
\end{align}
The leading piece we call disconnected, and it is relatively simple to deal with. The Fourier transform along this part is 
\begin{align}
{\cal R}_2^{{\rm{disc}}}(t) = {L^2}\int {d{\lambda _1}d{\lambda _2}\rho ({\lambda _1})\rho ({\lambda _2}){e^{i({\lambda _1} - {\lambda _2})t}}}  = {L^2}r_1^2(t)
\end{align}
where the function $r_1(t)$ is written as 
\begin{align}
{r_1}(t) = \frac{{{J_1}(2t)}}{t}
\end{align}
where $J_\nu(z)$ means the standard notation of the Bessel function. 
\subsubsection{The connected piece: box approximation}
Now let us discuss the connected piece, which is defined as 
\begin{align}
{\cal R}_2^{{\rm{conn}}}(t) = {{\cal R}_2}(t) - {\cal R}_2^{{\rm{disc}}}(t) = L -{L^2}\int {d{\lambda _1}d{\lambda _2}\frac{{{{\sin }^2}(L({\lambda _1} - {\lambda _2}))}}{{{{(L\pi ({\lambda _1} - {\lambda _2}))}^2}}}{e^{i({\lambda _1} - {\lambda _2})t}}} 
\end{align}
However, the integral that appearing here, is divergent. The reason is that the sine kernel written here cannot probe two energy eigenvalues $\lambda_1$ and $\lambda_2$ that are very close to each other, more precisely, around $|\lambda_1-\lambda_2|\sim 1/L$. However, we could invent a technology that is called \emph{box approximation} that could still capture some physics, where we will describe as the following.
\\
\\
Firstly, try to do a coordinate transformation 
\begin{align}
&{u_1} = {\lambda _1} - {\lambda _2}\nonumber\\
&{u_2} = {\lambda _2}
\end{align}
and thus the integral becomes
\begin{align}
{L^2}\int {d{\lambda _1}d{\lambda _2}\frac{{{{\sin }^2}(L({\lambda _1} - {\lambda _2}))}}{{{{(L\pi ({\lambda _1} - {\lambda _2}))}^2}}}{e^{i({\lambda _1} - {\lambda _2})t}}}  = {L^2}\int {d{u_1}d{u_2}\frac{{{{\sin }^2}(L{u_1})}}{{L\pi u_1^2}}{e^{i{u_1}t}}} 
\end{align}
The expression, written in this form, manifests the divergence because we have an uncontrolled integral over the variable $u_2$. Now let us firstly integrate over the variable $u_1$. Performing the integral, we have
\begin{align}
{L^2}\int {d{u_1}\frac{{{{\sin }^2}(L{u_1})}}{{{{(L\pi {u_1})}^2}}}{e^{i{u_1}t}}}  = \frac{L}{\pi }\left\{ {\begin{array}{*{20}{c}}
{1 - \frac{t}{{2L}}}&{{\text{for  }}t < 2L}\\
0&{{\text{for  }}t > 2L}
\end{array}} \right.
\end{align}
So the whole connected piece should be given by this function times the volume of the integration region of $u_2$: $\text{vol}(\mathbb{R})$. However, one could try to cut off the integration range by brute force to get a finite value. Let us assume that this cutoff space is symmetric around the origin, $[-\text{cut},\text{cut}]$, then the result is given by 
\begin{align}
{L^2}\int {d{\lambda _1}d{\lambda _2}\frac{{{{\sin }^2}(L({\lambda _1} - {\lambda _2}))}}{{{{(L\pi ({\lambda _1} - {\lambda _2}))}^2}}}{e^{i({\lambda _1} - {\lambda _2})t}}}  = \frac{{2{\rm{cut}} \times L}}{\pi }\left\{ {\begin{array}{*{20}{c}}
{1 - \frac{t}{{2L}}}&{{\text{for  }}t < 2L}\\
0&{{\text{for  }}t > 2L}
\end{array}} \right.
\end{align}
One can try to solve the $\text{cut}$ by checking the consistency of the result at $t=0$. At $t=0$ we know that the disconnected piece has contributed $L^2$, which is the whole form factor result, so the connected piece should get zero at $t=0$, which means that 
\begin{align}
2\text{cut}=\text{vol}(\mathbb{R})=\pi\to \text{cut}=\frac{\pi}{2}
\end{align}
One can see, which we will discuss later, that this cutoff $\pi/2$ also works for higher point cases. Let us think about an origin of it. Firstly, write done the one point function with Wigner semicircle law 
\begin{align}
\rho (\lambda ) = \frac{1}{{2\pi }}\sqrt {4 - {\lambda ^2}} 
\end{align}
Now let us pretend that $\lambda$ is very small, which is closed to the origin, then we have 
\begin{align}
\rho (\lambda=0) = \frac{1}{{\pi }}
\end{align}
Now, we could approximate, for small enough $\lambda$, that the semi-circle distribution is approximately a line. To compute the length of this line, we could use the normalization condition. The integral over $\rho$ is normalized by $1$, so if we choose our line to be distributed in the range $[-\text{cut},\text{cut}]$, so we get $2\text{cut}/\pi=1$, namely $\text{cut}=\pi/2$. A short explanation of this phenomena is that the box approximation is a brute force choice to make up the difference between the sine kernel and the semicircle when two energies are very close to each other $\lambda_1\to\lambda_2$. 
\\
\\
There is an another interpretation to this result. The connected part of the two point form factor is a linear increase, from $(0,0)$ to $(2L,L)$ in the coordinate $(t,\mathcal{R}_2^\text{conn}(t))$, and the stop growing (we call it as plateau). The origin $(0,0)$ is fixed, and the plateau time $t_p=2L$, is fixed by the property of the Fourier transform of the sine kernel, which will be independent of the cutoff choice. The plateau value $\mathcal{R}_2^\text{conn}(t_p=2L))=L$, is fixed by the long time average interpretation of definition of the form factor (which means that the damping $e^{(i(\lambda_1-\lambda_2)t)}$ for $\lambda_1\ne \lambda_2$ will be cancelled after long time averaging, and the only constant piece with $\lambda_1=\lambda_2$ will give the result $L$ because there are $L$ eigenvalues in total). Thus, drawing a line from $(0,0)$ to $(2L,L)$, assuming linearity, has to obtain the slope $1/2$. Because $(2L,L)$ is already fixed, so we could claim that the result beyond box approximation should be some non-linear physics.
\\
\\
As a summary, we obtain that the connected piece of the two point form factor is given by
\begin{align}
\mathcal{R}_2^{{\rm{conn}}} = L(1 - {r_2}(t))
\end{align}
where $r_2(t)$ is defined as 
\begin{align}
{r_2}(t) = \left\{ {\begin{array}{*{20}{c}}
{1 - \frac{t}{{2L}}}&{{\text{for  }}t < 2L}\\
0&{{\text{for  }}t > 2L}
\end{array}} \right.
\end{align}
\subsubsection{The connected piece: an improvement}
Now we introduce an improvement which is more refined than the box cutoff. In this part, we will try to use the short distance kernel 
\begin{align}
\widetilde{K}({\lambda _i},{\lambda _j}) = L\,\frac{{\sin (\pi L({\lambda _i} - {\lambda _j})\rho (({\lambda _i} + {\lambda _j})/2))}}{{\pi L({\lambda _i} - {\lambda _j})}}
\end{align}
where this kernel is an approximation when $\lambda_i$ and $\lambda_j$ are sufficiently close. The following technology is also mentioned in \cite{Chen:2017yzn}, but the results here, as far as we know, are novel.
\\
\\
Take this kernel in our hand, let us try to compute the connected part of the form factor. It is now captured by an integral
\begin{align}
{L^2}\int {d{\lambda _1}d{\lambda _2}\frac{{{{\sin }^2}(\pi L({\lambda _1} - {\lambda _2})\rho (({\lambda _1} + {\lambda _2})/2))}}{{{{(\pi L({\lambda _i} - {\lambda _j}))}^2}}}{e^{i({\lambda _1} - {\lambda _2})t}}} 
\end{align}
Here, we try applying a different coordinate transform
\begin{align}
&{u_1} = {\lambda _1} - {\lambda _2}\nonumber\\
&{u_2} = \frac{{{\lambda _1} + {\lambda _2}}}{2}
\end{align}
So the integral becomes 
\begin{align}
&{L^2}\int {d{\lambda _1}d{\lambda _2}\frac{{{{\sin }^2}(\pi L({\lambda _1} - {\lambda _2})\rho (({\lambda _1} + {\lambda _2})/2))}}{{{{(\pi L({\lambda _i} - {\lambda _j}))}^2}}}{e^{i({\lambda _1} - {\lambda _2})t}}}  \nonumber\\
&= {L^2}\int {d{u_1}d{u_2}\frac{{{{\sin }^2}(\pi L{u_1}\rho ({u_2}))}}{{{{(\pi L{u_1})}^2}}}{e^{i{u_1}t}}} 
\end{align}
The treatment here we could have is that we could split the space of $u_1$ in $\mathbb{R}$ by infinite number of intervals $\Omega$ at the center $u_2$, with the assumption that the integrand outside the interval will quickly decay. Suppose that we are now at the center $u_2$, and the interval has the range $[-\Omega_0/2,\Omega_0/2]$, then performing the integral, in the large $L$ limit, we have
\begin{align}
&{L^2}\int_{ - {\Omega _0}/2}^{{\Omega _0}/2} {d{u_1}\frac{{{{\sin }^2}(\pi L{u_1}\rho ({u_2}))}}{{{{(\pi L{u_1})}^2}}}{e^{i{u_1}t}}} \nonumber\\
&= \frac{L}{\pi }\rho ({u_2})\int_{ - L\rho ({u_2})\pi {\Omega _0}/2}^{L\rho ({u_2})\pi {\Omega _0}/2} {d{u_1}\frac{{{{\sin }^2}({u_1})}}{{u_1^2}}{e^{i{u_1}t/L\pi \rho ({u_2})}}}  \nonumber\\
&\sim \frac{L}{\pi }\rho ({u_2})\int_{ - \infty }^{ + \infty } {d{u_1}\frac{{{{\sin }^2}({u_1})}}{{u_1^2}}{e^{i{u_1}t/L\pi \rho ({u_2})}}} \nonumber\\
&= L\rho ({u_2})\left\{ {\begin{array}{*{20}{c}}
{1 - \frac{t}{{2\pi L\rho ({u_2})}}}&{{\text{for  }}t < 2\pi L\rho ({u_2})}\\
0&{{\text{for  }}t > 2\pi L\rho ({u_2})}
\end{array}} \right. \nonumber\\
&= \max \left( {L\rho ({u_2}) - \frac{t}{{2\pi }},0} \right)
\end{align}
Here an assumption we are making is that we are extending the range from an $L$ amplified interval to infinity, regardless of the fact that the exponent will be $\mathcal{O}(1)$ even if $u_1$ could scale as $\mathcal{O}(L)$. 
\\
\\
Now, we sum over the all intervals, which means that we are integrating over $u_2$ in the range $[-2,2]$ (the range of the semicircle), we get 
\begin{align}
&{L^2}\int {d{u_1}d{u_2}\frac{{{{\sin }^2}(\pi L{u_1}\rho ({u_2}))}}{{{{(\pi L{u_1})}^2}}}{e^{i{u_1}t}}} \\
&= \int_{ - 2}^2 {d{u_2}\max \left( {L\rho ({u_2}) - \frac{t}{{2\pi }},0} \right)} \\
&= \left\{ {\begin{array}{*{20}{c}}
{\frac{2}{\pi }L{\rm{arccsc}}\left( {\frac{{2L}}{{\sqrt {4{L^2} - {t^2}} }}} \right) - \frac{t}{{2\pi L}}\sqrt {4{L^2} - {t^2}} }&{{\text{for  }}t < 2L}\\
0&{{\text{for  }}t > 2L}
\end{array}} \right.
\end{align}
Thus, the connected form factor is given by
\begin{align}
\mathcal{R}_2^{{\rm{conn}}}(t)=\left\{ {\begin{array}{*{20}{c}}
{L - \frac{2}{\pi }L{\rm{arccsc}}\left( {\frac{{2L}}{{\sqrt {4{L^2} - {t^2}} }}} \right) + \frac{t}{{2\pi L}}\sqrt {4{L^2} - {t^2}} }&{{\text{for  }}t < 2L}\\
L&{{\text{for  }}t > 2L}
\end{array}} \right.
\end{align}
This result will capture more accurate physics. One interesting thing is that in the early time we expand it in small $t$ we get
\begin{align}
{\cal R}_2^{{\rm{conn}}}(t) = \frac{{2t}}{\pi } - \frac{{{t^3}}}{{12\pi {L^2}}} - \frac{{{t^5}}}{{320\pi {L^4}}} + \mathcal{O}\left( {{t^6}} \right)
\end{align}
Thus, this method will give the slope $2/\pi$ in the early time. This fact is verified by numerics in \cite{Cotler:2017jue}, but with the plateau still $(2L,L)$. The reason is that the function in the middle is nonlinear. One can estimate the nonlinear time scale, which is given by $t = \mathcal{O}( L)$, where in this time scale the higher order corrections to the linear function becomes important.
\\
\\
However, as this refined technology cannot be generalized to higher point case simply, we will keep using box cutoff for higher point case, which is believable for physics in spectral form factors. 
\subsection{Higher point form factor: theorem}
Higher point form factor calculations are based on multi-variable Fourier transforms of determinant of sine kernels. We will derive some generic results to establish the framework of computing higher point form factors in general based on the box approximation, and compute a four-point example. 
Our starting point will be the following theorem,
\begin{thm}
[Convolution formula for infinite $L$, in eq.5.2.23, \cite{book2}] We have the following formula to compute the convolution of the sine kernel:
\begin{align}
& \int{\prod\nolimits_{i=1}^{m}{d{{y}_{i}}}\exp (2\pi i\sum\limits_{j=1}^{m}{{{k}_{j}}{{y}_{j}}})s({{y}_{1}}-{{y}_{2}})}s({{y}_{2}}-{{y}_{3}})\ldots s({{y}_{m-1}}-{{y}_{m}})s({{y}_{m}}-{{y}_{1}}) \nonumber\\
& =\delta (\sum\limits_{j=1}^{m}{{{k}_{j}}})\int{dk}g(k)g(k+{{k}_{1}})\ldots g(k+{{k}_{m-1}})
\end{align}
where $s$ is the sine kernel
\begin{align}
s(r):=\frac{\sin (\pi r)}{\pi r}
\end{align}
and the principle valued Fourier transform of the sine kernel is given by
\begin{align}
\int{{{e}^{2\pi ikr}}s(r)}dr=g(k)=\left\{ \begin{matrix}
   1 & \left| k \right|<\frac{1}{2}  \\
   0 & \left| k \right|>\frac{1}{2}  \\
\end{matrix} \right.
\end{align}
\end{thm}
\begin{proof}
Change the variables
\begin{align}
& {{u}_{1}}={{y}_{1}}-{{y}_{2}} \nonumber\\
& {{u}_{2}}={{y}_{2}}-{{y}_{3}} \nonumber\\
& \ldots  \nonumber\\
& {{u}_{m-1}}={{y}_{m-1}}-{{y}_{m}} \nonumber\\
& {{u}_{m}}={{y}_{m}}
\end{align}
the inverse transform is
\begin{align}
  & {{y}_{1}}={{u}_{1}}+{{u}_{2}}+{{u}_{3}}+\ldots +{{u}_{m}}  \nonumber\\
 & {{y}_{2}}={{u}_{2}}+{{u}_{3}}+\ldots +{{u}_{m}}  \nonumber\\
 & \ldots   \nonumber\\
 & {{y}_{m-2}}={{u}_{m-2}}+{{u}_{m-1}}+{{u}_{m}} \nonumber\\
 & {{y}_{m-1}}={{u}_{m-1}}+{{u}_{m}}  \nonumber\\
 & {{y}_{m}}={{u}_{m}}
\end{align}
whose Jacobian is 1. Thus we obtain
\begin{align}
& \int{\prod\nolimits_{i=1}^{m}{d{{y}_{i}}}\exp (2\pi i\sum\limits_{j=1}^{m}{{{k}_{j}}{{y}_{j}}})s({{y}_{1}}-{{y}_{2}})}s({{y}_{2}}-{{y}_{3}})\ldots s({{y}_{m-1}}-{{y}_{m}})s({{y}_{m}}-{{y}_{1}}) \nonumber\\
& =\int{\prod\nolimits_{i=1}^{m}{d{{u}_{i}}}\exp (2\pi i\sum\limits_{l=1}^{m}{{{k}_{l}}\sum\limits_{\alpha =j}^{m}{{{u}_{\alpha }}}})s({{u}_{1}})}s({{u}_{2}})\ldots s({{u}_{m-1}})s(\sum\limits_{j=1}^{m-1}{{{u}_{j}}})  \nonumber\\
& =\int{\prod\nolimits_{i=1}^{m}{d{{u}_{i}}}\exp (2\pi i\sum\limits_{\alpha =1}^{m}{(\sum\limits_{l=1}^{\alpha }{{{k}_{l}}}){{u}_{\alpha }}})s({{u}_{1}})}s({{u}_{2}})\ldots s({{u}_{m-1}})s(\sum\limits_{j=1}^{m-1}{{{u}_{j}}})
\end{align}
From this, firstly we observe that we could firstly read off the integral over $u_m$, which is
\begin{align}
& \int{\prod\nolimits_{i=1}^{m}{d{{u}_{i}}}\exp (2\pi i\sum\limits_{\alpha =1}^{m}{(\sum\limits_{l=1}^{\alpha }{{{k}_{l}}}){{u}_{\alpha }}})}s({{u}_{1}})s({{u}_{2}})\ldots s({{u}_{m-1}})s(\sum\limits_{j=1}^{m-1}{{{u}_{j}}}) \nonumber\\
& =\int{\prod\nolimits_{i=1}^{m-1}{d{{u}_{i}}}\exp (2\pi i\sum\limits_{\alpha =1}^{m-1}{(\sum\limits_{l=1}^{\alpha }{{{k}_{l}}}){{u}_{\alpha }}})}s({{u}_{1}})s({{u}_{2}})\ldots s({{u}_{m-1}})s(\sum\limits_{j=1}^{m-1}{{{u}_{j}}})\nonumber\\
& \times\int{\exp (2\pi i(\sum\limits_{l=1}^{m}{{{k}_{l}}}){{u}_{m}})}d{{u}_{m}} \nonumber\\
& =\delta (\sum\limits_{l=1}^{m}{{{k}_{l}}})\int{\prod\nolimits_{i=1}^{m-1}{d{{u}_{i}}}\exp (2\pi i\sum\limits_{\alpha =1}^{m-1}{(\sum\limits_{l=1}^{\alpha }{{{k}_{l}}}){{u}_{\alpha }}})}s({{u}_{1}})s({{u}_{2}})\ldots s({{u}_{m-1}})s(\sum\limits_{j=1}^{m-1}{{{u}_{j}}})
\end{align}
How to deal with the last sine kernel? Now introduce a new variable $u$, which is
\begin{align}
s(\sum\limits_{j=1}^{m-1}{{{u}_{j}}})=s(-\sum\limits_{j=1}^{m-1}{{{u}_{j}}})=\int{du}s(u)\delta (u+\sum\limits_{j=1}^{m-1}{{{u}_{j}}})
\end{align}
and then, replace the delta function by exponential function
\begin{align}
s(\sum\limits_{j=1}^{m-1}{{{u}_{j}}})=\int{du}dks(u)\exp (2\pi ik(u+\sum\limits_{j=1}^{m-1}{{{u}_{j}}}))
\end{align}
Insert in the integral, we have
\begin{align}
& \int{\prod\nolimits_{i=1}^{m}{d{{y}_{i}}}\exp (2\pi i\sum\limits_{j=1}^{m}{{{k}_{j}}{{y}_{j}}})s({{y}_{1}}-{{y}_{2}})}s({{y}_{2}}-{{y}_{3}})\ldots s({{y}_{n-1}}-{{y}_{n}})s({{y}_{n}}-{{y}_{1}}) \nonumber\\
& =\delta (\sum\limits_{l=1}^{m}{{{k}_{l}}})\int{\prod\nolimits_{i=1}^{m-1}{d{{u}_{i}}}dudk\exp (2\pi i\sum\limits_{\alpha =1}^{m-1}{(\sum\limits_{l=1}^{\alpha }{{{k}_{l}}}){{u}_{\alpha }}})}\exp (2\pi ik(u+\sum\limits_{j=1}^{m-1}{{{u}_{j}}}))\nonumber\\
&\times s({{u}_{1}})s({{u}_{2}})\ldots s({{u}_{m-1}})s(u) \nonumber\\
& =\delta (\sum\limits_{l=1}^{m}{{{k}_{l}}})\int{\prod\nolimits_{i=1}^{m-1}{d{{u}_{i}}}dudk\exp (2\pi i\sum\limits_{\alpha =1}^{m-1}{(\sum\limits_{l=1}^{\alpha }{{{k}_{l}}}+k){{u}_{\alpha }}})}\exp (2\pi iku))\nonumber\\
& \times s({{u}_{1}})s({{u}_{2}})\ldots s({{u}_{m-1}})s(u) \nonumber\\
& =\delta (\sum\limits_{l=1}^{m}{{{k}_{l}}})\int{dk}\left( \prod\nolimits_{i=1}^{m-1}{\int{d{{u}_{i}}\exp (2\pi i(\sum\limits_{l=1}^{i}{{{k}_{l}}}+k){{u}_{i}})s({{u}_{i}})}} \right)\left( \int{du\exp (2\pi iku)s(u)} \right) \nonumber\\
& =\delta (\sum\limits_{l=1}^{m}{{{k}_{l}}})\int{dk}g(k)\prod\nolimits_{i=1}^{m-1}{g(\sum\limits_{l=1}^{i}{{{k}_{l}}}+k)}
\end{align}
as desired. 
\end{proof}
Now it is obvious to generalize this claim to large but finite $L$. We have
\begin{align}
&\int {\prod\limits_{i = 1}^m {d{\lambda _i}K({\lambda _1},{\lambda _2})K({\lambda _2},{\lambda _3}) \ldots K({\lambda _{m - 1}},{\lambda _m})K({\lambda _m},{\lambda _1}){e^{i\sum\limits_{i = 1}^m {{k_i}{\lambda _i}} }}} } \nonumber\\
&= \frac{L}{\pi }\int {d\lambda {e^{i\sum\limits_{i = 1}^m {{k_i}{\lambda}} }}} \int {dkg(k)g(k + \frac{{{k_1}}}{{2L}})g(k + \frac{{{k_2}}}{{2L}}) \ldots g(k + \frac{{{k_{m - 1}}}}{{2L}})} 
\end{align}
where the delta function is replaced by an integral over exponential function. We impose the box approximation again
\begin{align}
\int {d\lambda {e^{i\sum\limits_{i = 1}^m {{k_i}{\lambda}} }}}  \to \int_{ - \pi /2}^{\pi /2} {d\lambda {e^{i\sum\limits_{i = 1}^m {{k_i}{\lambda}} }}} 
\end{align}
which is always fixed by the normalization at initial time,
\begin{align}
\frac{L}{\pi }\int_{ - \pi /2}^{\pi /2} {d\lambda {e^{i\sum\limits_{i = 1}^m {{k_i}{\lambda}} }}} \int {dkg(k)g(k + \frac{{{k_1}}}{{2L}})g(k + \frac{{{k_2}}}{{2L}}) \ldots g(k + \frac{{{k_{m - 1}}}}{{2L}})} {|_{{k_1} = {k_2} =  \ldots {k_{m - 1}} = 0}} = L
\end{align}
and we find the number $\pi/2$ is universal for all $m$. So we finally get the useful formula
\begin{thm}[Convolution formula for finite large $L$]
\begin{align}
&\int {\prod\limits_{i = 1}^m {d{\lambda _i}K({\lambda _1},{\lambda _2})K({\lambda _2},{\lambda _3}) \ldots K({\lambda _{m - 1}},{\lambda _m})K({\lambda _m},{\lambda _1}){e^{i\sum\limits_{i = 1}^m {{k_i}{\lambda _i}} }}} } \nonumber\\
&= L{r_3}(\sum\limits_{i = 1}^m {{k_i}} )\int {dkg(k)g(k + \frac{{{k_1}}}{{2L}})g(k + \frac{{{k_2}}}{{2L}}) \ldots g(k + \frac{{{k_{m - 1}}}}{{2L}})} 
\end{align}
where we define the function
\begin{align}
{r_3}(t) = \frac{{\sin (\pi t/2)}}{{\pi t/2}}
\end{align}
\end{thm}
This convolution formula allows us to compute any higher point spectral form factors. We will show an example about how the four point form factor has been computed.  
\subsection{Four point form factor}
Now let us consider the four point form factor as an example
\begin{align}
{\mathcal{R}_{4}}=\sum\limits_{a,b,c,d=1}^{L}{\int{D\lambda}{{e}^{i({{\lambda }_{a}}+{{\lambda }_{b}}-{{\lambda }_{c}}-{{\lambda }_{d}})t}}}
\end{align}
Before our computation, we will define the following building block functions
\begin{align}
&{r_1}(t) = \frac{{{J_1}(2t)}}{t}\nonumber\\
&{r_2}(t) = \left\{ {\begin{array}{*{20}{c}}
{1 - \frac{t}{{2L}}}&{{\text{for  }}t < 2L}\\
0&{{\text{for  }}t > 2L}
\end{array}} \right. \nonumber\\
&{r_3}(t) = \frac{{\sin (\pi t/2)}}{{\pi t/2}}
\end{align}
Take a look at the classifications of combinations in ${\mathcal{R}_{4}}$, which is
\begin{itemize}
\item $a=b=c=d=e=f$: Contribute $L$.
\item $a=b$: Contribute $L(L-1)(L-2)\int{D\lambda}{{e}^{i(2{{\lambda }_{1}}-{{\lambda }_{2}}-{{\lambda }_{3}})t}}$.
\item $c=d$: Contribute $L(L-1)(L-2)\int{D\lambda}{{e}^{i({{\lambda }_{1}}+{{\lambda }_{2}}-2{{\lambda }_{3}})t}}$.
\item $a=c$ or $a=d$ or $b=c$ or $b=d$: Contribute $4L(L-1)(L-2)\int{D\lambda}{{e}^{i({{\lambda }_{1}}-{{\lambda }_{2}})t}}$.
\item $b=c=d$ or $a=c=d$ or $a=b=d$ or $a=b=c$: Contribute $4L(L-1)\int{D\lambda}{{e}^{i({{\lambda }_{1}}-{{\lambda }_{2}})t}}$.
\item $a=b$ and $c=d$:  Contribute $L(L-1)\int{D\lambda}{{e}^{i(2{{\lambda }_{1}}-2{{\lambda }_{2}})t}}$.
\item $a=c$ and $b=d$, or $a=d$ and $b=c$: Contribute $2L(L-1)$.
\item All inequal indexes: $L(L-1)(L-2)(L-3)\int{D\lambda}{{e}^{i({{\lambda }_{1}}+{{\lambda }_{2}}-{{\lambda }_{3}}-{{\lambda }_{4}})t}}$
\end{itemize}
Add the total prefactors will give $L^4$. Add them together we get
\begin{align}
  & {\mathcal{R}_{4}}=L(L-1)(L-2)(L-3)\int{D\lambda}{{e}^{i({{\lambda }_{1}}+{{\lambda }_{2}}-{{\lambda }_{3}}-{{\lambda }_{4}})t}} \nonumber\\
 & +2L(L-1)(L-2)\operatorname{Re}\int{D\lambda}{{e}^{i(2{{\lambda }_{1}}-{{\lambda }_{2}}-{{\lambda }_{3}})t}}  \nonumber\\
 & +L(L-1)\int{D\lambda}{{e}^{i(2{{\lambda }_{1}}-2{{\lambda }_{2}})t}}  \nonumber\\
 & +4L{{(L-1)}^{2}}\int{D\lambda}{{e}^{i({{\lambda }_{1}}-{{\lambda }_{2}})t}}  \nonumber\\
 & +2{{L}^{2}}-L  \nonumber\\
 & =L(L-1)(L-2)(L-3)\int{D\lambda}{{e}^{i({{\lambda }_{1}}+{{\lambda }_{2}}-{{\lambda }_{3}}-{{\lambda }_{4}})t}}  \nonumber\\
 & +2L(L-1)(L-2)\operatorname{Re}\int{D\lambda}{{e}^{i(2{{\lambda }_{1}}-{{\lambda }_{2}}-{{\lambda }_{3}})t}}  \nonumber\\
 & +{L}^{2}|{r}_{1}(2t)|^{2}-Lr_2(2t) \nonumber\\
 & +4(L-1)({{L}^{2}}|r_{1}(t)|^{2}-L{r_{2}}(t))  \nonumber\\
 & +2{{L}^{2}}-L
\end{align}
We already obtained what the last three terms are. Now we only need to consider the first two terms.
\subsubsection{The first term}
The first term is an actual four point function.
\begin{align}
L(L-1)(L-2)(L-3)\int{D\lambda}{{e}^{i({{\lambda }_{1}}+{{\lambda }_{2}}-{{\lambda }_{3}}-{{\lambda }_{4}})t}}
\end{align}
When expanding the determinant, the terms could be summarized as the following,
\begin{itemize}
\item 4-type: In this case we have
\begin{align}
  & -2\int{d\lambda_1d\lambda_2d\lambda_3d\lambda_4 {{K}}({{\lambda }_{1}},{{\lambda }_{3}}){{K}}({{\lambda }_{3}},{{\lambda }_{2}}){{K}}({{\lambda }_{2}},{{\lambda }_{4}}){{K}}({{\lambda }_{4}},{{\lambda }_{1}})}{{e}^{i({{\lambda }_{1}}+{{\lambda }_{2}}-{{\lambda }_{3}}-{{\lambda }_{4}})t}} \nonumber\\
 & -2\int{d\lambda_1d\lambda_2d\lambda_3d\lambda_4 {{K}}({{\lambda }_{1}},{{\lambda }_{2}}){{K}}({{\lambda }_{2}},{{\lambda }_{3}}){{K}}({{\lambda }_{3}},{{\lambda }_{4}}){{K}}({{\lambda }_{4}},{{\lambda }_{1}})}{{e}^{i({{\lambda }_{1}}+{{\lambda }_{2}}-{{\lambda }_{3}}-{{\lambda }_{4}})t}} \nonumber\\
 & -2\int{d\lambda_1d\lambda_2d\lambda_3d\lambda_4 {{K}}({{\lambda }_{1}},{{\lambda }_{2}}){{K}}({{\lambda }_{2}},{{\lambda }_{4}}){{K}}({{\lambda }_{4}},{{\lambda }_{3}}){{K}}({{\lambda }_{3}},{{\lambda }_{1}})}{{e}^{i({{\lambda }_{1}}+{{\lambda }_{2}}-{{\lambda }_{3}}-{{\lambda }_{4}})t}}
\end{align}
thus the result is
\begin{align}
-6Lr_{2}(2t)
\end{align}
\item 1-1-1-1-type: In this case we have
\begin{align}
\int{d{{\lambda }_{1}}}{{K}}({{\lambda }_{1}},{{\lambda }_{1}}){{e}^{i{{\lambda }_{1}}t}}\int{d{{\lambda }_{2}}{{K}}({{\lambda }_{2}},{{\lambda }_{2}}){{e}^{i{{\lambda }_{2}}t}}}\int{d{{\lambda }_{3}}{{K}}({{\lambda }_{3}},{{\lambda }_{3}}){{e}^{-i{{\lambda }_{3}}t}}}\int{d{{\lambda }_{4}}{{K}}({{\lambda }_{4}},{{\lambda }_{4}}){{e}^{-i{{\lambda }_{4}}t}}}
\end{align}
This term contributes
\begin{align}
{{L}^{4}}|r_{1}(t)|^{4}
\end{align}
\item 1-1-2-type: In this case we have
\begin{align}
  & -\int{d{{\lambda }_{1}}d{{\lambda }_{2}}{K^2}({{\lambda }_{1}},{{\lambda }_{2}})}{{e}^{i({{\lambda }_{1}}+{{\lambda }_{2}})t}}\int{d{{\lambda }_{3}}}{{K}}({{\lambda }_{3}},{{\lambda }_{3}}){{e}^{-i{{\lambda }_{3}}t}}\int{d{{\lambda }_{4}}{{K}}({{\lambda }_{4}},{{\lambda }_{4}}){{e}^{-i{{\lambda }_{4}}t}}}  \nonumber\\
 & -\int{d{{\lambda }_{3}}d{{\lambda }_{4}}{K^2}({{\lambda }_{3}},{{\lambda }_{3}})}{{e}^{-i({{\lambda }_{3}}+{{\lambda }_{4}})t}}\int{d{{\lambda }_{1}}}{{K}}({{\lambda }_{1}},{{\lambda }_{1}}){{e}^{i{{\lambda }_{1}}t}}\int{d{{\lambda }_{2}}{{K}}({{\lambda }_{2}},{{\lambda }_{2}}){{e}^{i{{\lambda }_{2}}t}}}  \nonumber\\
 & -\int{d{{\lambda }_{1}}d{{\lambda }_{4}}{K^2}({{\lambda }_{1}},{{\lambda }_{4}})}{{e}^{i({{\lambda }_{1}}-{{\lambda }_{4}})t}}\int{d{{\lambda }_{2}}}{{K}}({{\lambda }_{2}},{{\lambda }_{2}}){{e}^{i{{\lambda }_{2}}t}}\int{d{{\lambda }_{3}}{{K}}({{\lambda }_{3}},{{\lambda }_{3}}){{e}^{-i{{\lambda }_{3}}t}}}  \nonumber\\
 & -\int{d{{\lambda }_{1}}d{{\lambda }_{3}}{K^2}({{\lambda }_{1}},{{\lambda }_{3}})}{{e}^{i({{\lambda }_{1}}-{{\lambda }_{3}})t}}\int{d{{\lambda }_{2}}}{{K}}({{\lambda }_{2}},{{\lambda }_{2}}){{e}^{i{{\lambda }_{2}}t}}\int{d{{\lambda }_{4}}{{K}}({{\lambda }_{4}},{{\lambda }_{4}}){{e}^{-i{{\lambda }_{4}}t}}}  \nonumber\\
 & -\int{d{{\lambda }_{2}}d{{\lambda }_{4}}{K^2}({{\lambda }_{2}},{{\lambda }_{4}})}{{e}^{i({{\lambda }_{2}}-{{\lambda }_{4}})t}}\int{d{{\lambda }_{1}}}{{K}}({{\lambda }_{1}},{{\lambda }_{1}}){{e}^{i{{\lambda }_{1}}t}}\int{d{{\lambda }_{3}}{{K}}({{\lambda }_{3}},{{\lambda }_{3}}){{e}^{-i{{\lambda }_{3}}t}}}  \nonumber\\
 & -\int{d{{\lambda }_{2}}d{{\lambda }_{3}}{K^2}({{\lambda }_{2}},{{\lambda }_{3}})}{{e}^{i({{\lambda }_{2}}-{{\lambda }_{3}})t}}\int{d{{\lambda }_{1}}}{{K}}({{\lambda }_{1}},{{\lambda }_{1}}){{e}^{i{{\lambda }_{1}}t}}\int{d{{\lambda }_{4}}{{K}}({{\lambda }_{4}},{{\lambda }_{4}}){{e}^{-i{{\lambda }_{4}}t}}}
\end{align}
This term contributes
\begin{align}
-2{{L}^{3}}\text{Re}(r_{1}^{2}(t)){{r}_{2}}(t){{r}_{3}}(2t)-4{{L}^{3}}|r_{1}(t)|^2{{r}_{2}}(t)
\end{align}
\item 2-2-type: In this case we have
\begin{align}
  & +\int{d{{\lambda }_{1}}d{{\lambda }_{4}}{K^2}({{\lambda }_{1}},{{\lambda }_{4}}){{e}^{i({{\lambda }_{1}}-{{\lambda }_{4}})t}}}\int{d{{\lambda }_{2}}d{{\lambda }_{3}}{K^2}({{\lambda }_{2}},{{\lambda }_{3}}){{e}^{i({{\lambda }_{2}}-{{\lambda }_{3}})t}}} \nonumber\\
 & +\int{d{{\lambda }_{1}}d{{\lambda }_{3}}{K^2}({{\lambda }_{1}},{{\lambda }_{3}}){{e}^{i({{\lambda }_{1}}-{{\lambda }_{3}})t}}}\int{d{{\lambda }_{2}}d{{\lambda }_{4}}{K^2}({{\lambda }_{2}},{{\lambda }_{4}}){{e}^{i({{\lambda }_{2}}-{{\lambda }_{4}})t}}} \nonumber\\
 & +\int{d{{\lambda }_{1}}d{{\lambda }_{2}}{K^2}({{\lambda }_{1}},{{\lambda }_{2}}){{e}^{i({{\lambda }_{1}}+{{\lambda }_{2}})t}}}\int{d{{\lambda }_{3}}d{{\lambda }_{4}}{K^2}({{\lambda }_{3}},{{\lambda }_{4}}){{e}^{-i({{\lambda }_{3}}+{{\lambda }_{4}})t}}}
\end{align}
This term contributes
\begin{align}
2{{L}^{2}}r_{2}^{2}(t)+{{L}^{2}}r_{2}^{2}(t)r_{3}^{2}(2t)
\end{align}
\item 3-1-type: In this case we have
\begin{align}
8{{L}^{2}}\text{Re}(r_{1}(t))r_{2}(t){r_{3}}(t)
\end{align}
Finally as a summary we have
\begin{align}
& L(L-1)(L-2)(L-3)\int{D\lambda}{{e}^{i({{\lambda }_{1}}+{{\lambda }_{2}}-{{\lambda }_{3}}-{{\lambda }_{4}})t}} \nonumber\\
  &={{L}^{4}}|r_{1}(t)|^4 \nonumber\\
 & -2{{L}^{3}}\text{Re}(r_{1}^{2}(t)){{r}_{2}}(t){{r}_{3}}(2t)-4{{L}^{3}}|r_{1}(t)|^2{{r}_{2}}(t) \nonumber\\
 & +2{{L}^{2}}r_{2}^{2}(t)+{{L}^{2}}r_{2}^{2}(t)r_{3}^{2}(2t)+8{{L}^{2}}\text{Re}({{r}_{1}}(t)){{r}_{2}}(t){{r}_{3}}(t) \nonumber\\
 & -6L{r}_{2}(2t)
\end{align}
\end{itemize}
\subsubsection{The second term}
In this part we will evaluate the second term
\begin{align}
2L(L-1)(L-2)\operatorname{Re}\int{D\lambda}{{e}^{i(2{{\lambda }_{1}}-{{\lambda }_{2}}-{{\lambda }_{3}})t}}
\end{align}
Let us firstly consider it without a factor of 2
\begin{align}
L(L-1)(L-2)\operatorname{Re}\int{D\lambda}{{e}^{i(2{{\lambda }_{1}}-{{\lambda }_{2}}-{{\lambda }_{3}})t}}
\end{align}
Then we obtain
\begin{itemize}
\item 3-type: In this case we have
\begin{align}
2\text{Re}\int{d{{\lambda }_{1}}d{{\lambda }_{2}}d{{\lambda }_{3}}{{K}}({{\lambda }_{1}},{{\lambda }_{2}}){{K}}({{\lambda }_{2}},{{\lambda }_{3}}){{K}}({{\lambda }_{3}},{{\lambda }_{1}})}{{e}^{i(2{{\lambda }_{1}}-{{\lambda }_{2}}-{{\lambda }_{3}})t}}
\end{align}
This term contributes
\begin{align}
2Lr_3(t)
\end{align}
\item 1-1-1-type: In this case we have
\begin{align}
\text{Re}\int{d{{\lambda }_{1}}}{{K}}({{\lambda }_{1}},{{\lambda }_{1}}){{e}^{2i{{\lambda }_{1}}t}}\int{d{{\lambda }_{2}}}{{K}}({{\lambda }_{2}},{{\lambda }_{2}}){{e}^{-i{{\lambda }_{2}}t}}\int{d{{\lambda }_{3}}}{{K}}({{\lambda }_{3}},{{\lambda }_{3}}){{e}^{-i{{\lambda }_{3}}t}}
\end{align}
This term contributes
\begin{align}
{{L}^{3}}\text{Re}({{r}_{1}}(2t)r_{1}^{*2}(t))
\end{align}
\item 2-1-type: In this case we have
\begin{align}
& -\int{d{{\lambda }_{1}}}{{K}}({{\lambda }_{1}},{{\lambda }_{1}}){{e}^{i2{{\lambda }_{1}}t}}\int{d{{\lambda }_{2}}d{{\lambda }_{3}}{K^2}({{\lambda }_{2}},{{\lambda }_{3}})}{{e}^{-i({{\lambda }_{2}}+{{\lambda }_{3}})t}} \nonumber\\
& -\int{d{{\lambda }_{2}}}{{K}}({{\lambda }_{2}},{{\lambda }_{2}}){{e}^{-i{{\lambda }_{2}}t}}\int{d{{\lambda }_{1}}d{{\lambda }_{3}}{K^2}({{\lambda }_{1}},{{\lambda }_{3}})}{{e}^{i(2{{\lambda }_{1}}-{{\lambda }_{3}})t}} \nonumber\\
& -\int{d{{\lambda }_{3}}}{{K}}({{\lambda }_{3}},{{\lambda }_{3}}){{e}^{-i{{\lambda }_{3}}t}}\int{d{{\lambda }_{1}}d{{\lambda }_{2}}{K^2}({{\lambda }_{1}},{{\lambda }_{2}})}{{e}^{i(2{{\lambda }_{1}}-{{\lambda }_{2}})t}}
\end{align}
This term contributes
\begin{align}
-{{L}^{2}}\text{Re}({{r}_{1}}(2t)){{r}_{3}}(2t){{r }_{2}}(t)-2{{L}^{2}}\text{Re}({{r}_{1}^*}(t)){{r}_{3}}(t){{r }_{2}}(2t)
\end{align}
\end{itemize}
So finally we make a summary that
\begin{align}
  & 2L(L-1)(L-2)\operatorname{Re}\int{D\lambda }{{e}^{i(2{{\lambda }_{1}}-{{\lambda }_{2}}-{{\lambda }_{3}})t}} \nonumber\\
 & =2{{L}^{3}}\text{Re}({{r}_{1}}(2t)r_{1}^{*2}(t)) \nonumber\\
 & -2{{L}^{2}}\text{Re}({{r}_{1}}(2t)){{r}_{3}}(2t){{r }_{2}}(t)-4{{L}^{2}}\text{Re}({{r}_{1}^*}(t)){{r}_{3}}(t){{r }_{2}}(2t)\nonumber\\
 & +4L{{r}_{2}}(3t)
\end{align}
\subsubsection{The final result}
Here we make a summary. The final result for $\mathcal{R}_4$ is
\begin{align}
&{{\mathcal{R}}_{4}}={{L}^{4}}|{{r}_{1}}(t){{|}^{4}} \nonumber\\
&-2{{L}^{3}}\text{Re}(r_{1}^{2}(t)){{r}_{2}}(t){{r}_{3}}(2t)-4{{L}^{3}}|{{r}_{1}}(t){{|}^{2}}{{r}_{2}}(t)+2{{L}^{3}}\text{Re}({{r}_{1}}(2t)r_{1}^{*2}(t))+4{{L}^{3}}|{{r}_{1}}(t){{|}^{2}} \nonumber\\
&+2{{L}^{2}}r_{2}^{2}(t)+{{L}^{2}}r_{2}^{2}(t)r_{3}^{2}(2t)+8{{L}^{2}}\text{Re}({{r}_{1}}(t)){{r}_{2}}(t){{r}_{3}}(t)-2{{L}^{2}}\text{Re}({{r}_{1}}(2t)){{r}_{3}}(2t){{r}_{2}}(t)\nonumber\\
&-4{{L}^{2}}\text{Re}(r_{1}^{*}(t)){{r}_{3}}(t){{r}_{2}}(2t)+{{L}^{2}}|{{r}_{1}}(2t){{|}^{2}}-4{{L}^{2}}|{{r}_{1}}(t){{|}^{2}}-4{{L}^{2}}{{r}_{2}}(t)+2{{L}^{2}}\nonumber\\
&-7L{{r}_{2}}(2t)+4L{{r}_{2}}(3t)+4L{{r}_{2}}(t)-L
\end{align}
In the large $L$ limit, one can find some terms are suppressed in all time scale. We could also have an approximate formula 
\begin{align}
&{{\mathcal{R}}_{4}}={{L}^{4}}|{{r}_{1}}(t){{|}^{4}} \nonumber\\
&+2{{L}^{2}}r_{2}^{2}(t)-4{{L}^{2}}{{r}_{2}}(t)+2{{L}^{2}}\nonumber\\
&-7L{{r}_{2}}(2t)+4L{{r}_{2}}(3t)+4L{{r}_{2}}(t)-L
\end{align}
which capture main physics of the four point spectral form factor. 
\subsection{Finite temperature result}
Finally, we will take a look at the finite temperature result, where this result will also rely on the refined kernel and the interval splitting technology as mentioned before, so here we only precisely compute the two point case. The definition of the finite temperature two point form factor is
\begin{align}
{\mathcal{R}_2} = \sum\limits_{i,j} {\int {D\lambda } {e^{i({\lambda _i} - {\lambda _j})t}}{e^{ - \beta ({\lambda _i} - {\lambda _j})}}} 
\end{align}
Following from a simple analysis we have
\begin{align}
&{\mathcal{R}_2} = \sum\limits_{i,j} {\int {D\lambda } {e^{i({\lambda _i} - {\lambda _j})t}}{e^{ - \beta ({\lambda _i} + {\lambda _j})}}} \nonumber\\
&= L\int {d\lambda } \rho (\lambda ){e^{ - 2\beta \lambda }} + L(L - 1)\int {d{\lambda _1}d{\lambda _2}} {\rho ^{(2)}}(\lambda ){e^{i({\lambda _1} - {\lambda _2})t}}{e^{ - \beta ({\lambda _1} + {\lambda _2})}} \nonumber\\
&= L{r_1}(2i\beta ) + {L^2}{r_1}(t + i\beta ){r_1}(t - i\beta ) - \int {d{\lambda _1}d{\lambda _2}} {K^2}({\lambda _1},{\lambda _2}){e^{i({\lambda _1} - {\lambda _2})t}}{e^{ - \beta ({\lambda _1} + {\lambda _2})}}
\end{align}
So we have the seperation
\begin{align}
&{\cal R}_2^{{\rm{disc}}}(t,\beta ) = {L^2}{r_1}(t + i\beta ){r_1}(t - i\beta )\nonumber\\
&{\cal R}_2^{{\rm{conn}}}(t,\beta ) = L{r_1}(2i\beta ) - \int {d{\lambda _1}d{\lambda _2}} {K^2}({\lambda _1},{\lambda _2}){e^{i({\lambda _1} - {\lambda _2})t}}{e^{ - \beta ({\lambda _1} + {\lambda _2})}}
\end{align}
Thus for the connected part, we could compute the integral
\begin{align}
{L^2}\int {d{\lambda _1}d{\lambda _2}\frac{{{{\sin }^2}(\pi L({\lambda _1} - {\lambda _2})\rho (({\lambda _1} + {\lambda _2})/2))}}{{{{(\pi L({\lambda _i} - {\lambda _j}))}^2}}}{e^{i({\lambda _1} - {\lambda _2})t}}{e^{ - \beta ({\lambda _1} + {\lambda _2})}}} 
\end{align}
Transform the coordinate again we get
\begin{align}
&{L^2}\int {d{\lambda _1}d{\lambda _2}\frac{{{{\sin }^2}(\pi L({\lambda _1} - {\lambda _2})\rho (({\lambda _1} + {\lambda _2})/2))}}{{{{(\pi L({\lambda _i} - {\lambda _j}))}^2}}}{e^{i({\lambda _1} - {\lambda _2})t}}{e^{ - \beta ({\lambda _1} + {\lambda _2})}}} \nonumber\\
&= {L^2}\int {d{u_1}d{u_2}\frac{{{{\sin }^2}(\pi L{u_1}\rho ({u_2}))}}{{{{(\pi L{u_1})}^2}}}{e^{i{u_1}t}}{e^{ - 2\beta {u_2}}}} 
\end{align}
The small interval contribution will again give
\begin{align}
{L^2}\int_{ - {\Omega _0}/2}^{{\Omega _0}/2} {d{u_1}\frac{{{{\sin }^2}(\pi L{u_1}\rho ({u_2}))}}{{{{(\pi L{u_1})}^2}}}{e^{i{u_1}t}}{e^{ - 2\beta {u_2}}}}  \sim {e^{ - 2\beta {u_2}}}\max \left( {L\rho ({u_2}) - \frac{t}{{2\pi }},0} \right)
\end{align}
We cannot find an analytic formula for general $\beta$ if we wish to compute this integral over $u_2$. However, one can expand it over small $\beta$. We have
\begin{align}
&{L^2}\int {d{u_1}d{u_2}\frac{{{{\sin }^2}(\pi L{u_1}\rho ({u_2}))}}{{{{(\pi L{u_1})}^2}}}{e^{i{u_1}t}}{e^{ - 2\beta {u_2}}}} \nonumber\\
&= \left\{ {\begin{array}{*{20}{c}}
\begin{array}{l}
\frac{2}{\pi }L{\rm{arccsc}}\left( {\frac{{2L}}{{\sqrt {4{L^2} - {t^2}} }}} \right) - \frac{t}{{2\pi L}}\sqrt {4{L^2} - {t^2}} \\
 + \frac{{{\beta ^2}\left( { - 10{L^2}t\sqrt {4{L^2} - {t^2}}  + {t^3}\sqrt {4{L^2} - {t^2}}  + 24{L^4}{\rm{arccsc}}\left( {\frac{{2L}}{{\sqrt {4{L^2} - {t^2}} }}} \right)} \right)}}{{6\pi {L^3}}} + \mathcal{O}({\beta ^4})
\end{array}&{{\text{for  }}t < 2L}\\
0&{{\text{for  }}t > 2L}
\end{array}} \right.
\end{align}
\section{GOE/GSE spectral form factor}\label{GOE}
\subsection{Random matrix theory review}
GOE and GSE describe physical systems with discrete antiunitary symmetries. Here we will briefly review the mathematical construction. We define the joint distribution of eigenvalues for GOE and GSE as
\begin{align}
P(\lambda_1,\ldots,\lambda_L) \sim e^{-\tilde{\beta}\frac{L}{4} \sum_i \lambda_i^2} \prod_{i<j} (\lambda_i-\lambda_j)^{\tilde{\beta}}
\end{align}
where $\tilde{\beta}=1,4$ for GOE or GSE. Here our convention is $L\times L$ for GOE and $2L\times 2L$ for GSE (where the later only has $L$ independent eigenvalues because it has natural degeneracy two by construction). Again, we define the $n$-point correlation function as
\begin{align}
\rho^{(n)}(\lambda_1,\ldots,\lambda_n) = \int d\lambda_{n+1}\ldots d\lambda_L P(\lambda_1,\ldots,\lambda_L)
\end{align}
To go further, we need some quaternion matrix theories.  A quaternion is generated by four units $e_{1,2,3}$,
\begin{align}
q=q^{(0)}+q^{(1)}e_1+q^{(2)}e_2+q^{(3)}e_3
\end{align}
The units are defined to satisfy the following multiplication laws
\begin{align}
&1\times e_j=e_j\times 1=e_j\nonumber\\
&e_1^2=e_2^2=e_3^2=e_1e_2e_3=-1
\end{align}
These units have matrix representations:
\begin{align}
& 1\to \left( \begin{matrix}
   1 & 0  \\
   0 & 1  \\
\end{matrix} \right) ~~~~~~ {{e}_{1}}\to \left( \begin{matrix}
   i & 0  \\
   0 & -i  \\
\end{matrix} \right) \nonumber\\
 & {{e}_{2}}\to \left( \begin{matrix}
   0 & 1  \\
   -1 & 0  \\
\end{matrix} \right) ~~~~~~ {{e}_{3}}\to \left( \begin{matrix}
   0 & i  \\
   i & 0  \\
\end{matrix} \right)
\end{align}
The determinant of a $n \times n$ quaternion matrix $Q=(Q_{jk})$ is defined as
\begin{align}
\det Q=\sum\limits_{\sigma }{{{(-1)}^{n-\ell (\sigma ,Q)}}\prod\nolimits_{t=1}^{\ell (\sigma ,Q)}{\text{cy}{{\text{c}}_{t}}(\sigma ,Q)}}
\end{align}
where $\sigma$ is any possible permutations from $1$ to $n$, for corresponding permutations, we could find all $\ell$ closed cycles for that permutations. For instance, for a cycle $t$ like:
\begin{align}
t: ~~~a\to b \to c \to \cdots \to d \to a
\end{align}
the corresponding contribution in the product is
\begin{align}
{{\text{cy}{{\text{c}}_{t}}(\sigma ,Q)}}=(q_{ab}q_{bc}\cdots q_{cd}q_{da})^{(0)}
\end{align}
where the upper index ${(0)}$ means the scalar part, or equivalently
\begin{align}
(q_{ab}q_{bc}\cdots q_{cd}q_{da})^{(0)}=\frac{1}{2}\text{Tr}(q_{ab}q_{bc}\cdots q_{cd}q_{da})
\end{align}
There are some useful theorems to compute the quaternion determinant. For instance, if we use the matrix representation of quaternion by replacing $e_j$ to $2\times 2$ matrices, we can define a $2n\times 2n$ complex matrix $C(Q)$ for a $n\times n$ quaternion matrix $Q$. Now define $Z=C(e_2 I_n)$, where $I_n$ is the $n\times n$ unit matrix, then if $Q$ is a self-dual matrix, namely each cycle from the product in the definition of determinant are reversible (dual to each other), then one can show that
\begin{align}
\det Q=\text{Pf}(ZC(Q))=\text{det}^{1/2}(C(Q))
\end{align}
With these definitions, we could define the quaternion kernels for GOE and GSE. In GOE and GSE, the sine kernel $K$ is replaced by quaternion, which could be represented as a $2\times 2$ matrix. In fact, define the following function:
\begin{align}
&\hat{s}(r)=\frac{\sin (r)}{r}~~~~~~\epsilon(r)=\frac{1}{2}\text{sign}(r)\nonumber\\
&\mathbf{D}\hat{s}(r)=\partial_r \hat{s}(r)~~~~~~\mathbf{I}\hat{s}(r)=\int_{0}^{r}\hat{s}(t)dt
\end{align}
Thus the quaternion kernel for GOE is
\begin{align}
K({\lambda _i},{\lambda _j}) \equiv \left\{ {\begin{array}{*{20}{l}}
{\frac{L}{\pi }\left( {\begin{array}{*{20}{c}}
{\hat s(L({\lambda _i} - {\lambda _j}))}&{{\bf{D}}\hat s(L({\lambda _i} - {\lambda _j}))}\\
{{\bf{I}}\hat s(L({\lambda _i} - {\lambda _j})) - \epsilon (L({\lambda _i} - {\lambda _j}))}&{\hat s(L({\lambda _i} - {\lambda _j}))}
\end{array}} \right)}&{{\rm{for}}\;\;i\not  = j}\\
{}&{}\\
{\frac{L}{{2\pi }}\sqrt {4 - \lambda _i^2} \left( {\begin{array}{*{20}{c}}
1&0\\
0&1
\end{array}} \right)\;}&{{\rm{for}}\;\;i = j}
\end{array}} \right.
\end{align}
while for GSE it is
\begin{align}
K({\lambda _i},{\lambda _j}) \equiv \left\{ {\begin{array}{*{20}{l}}
{\frac{L}{\pi }\left( {\begin{array}{*{20}{c}}
{\hat s(2L({\lambda _i} - {\lambda _j}))}&{{\bf{D}}\hat s(2L({\lambda _i} - {\lambda _j}))}\\
{{\bf{I}}\hat s(2L({\lambda _i} - {\lambda _j}))}&{\hat s(2L({\lambda _i} - {\lambda _j}))}
\end{array}} \right)}&{{\rm{for}}\;\;i\not  = j}\\
{}&{}\\
{\frac{L}{{2\pi }}\sqrt {4 - \lambda _i^2} \left( {\begin{array}{*{20}{c}}
1&0\\
0&1
\end{array}} \right)\;}&{{\rm{for}}\;\;i = j}
\end{array}} \right.
\end{align}
The structure of GOE and GSE is not called simply determined by the ordinary meaning of a determinant of some two point functions. It is not called the \emph{determinantal point process} in random matrix theory literature, but it is called the \emph{Pfaffian point process}. For our practical motivation, we may define the joint eigenvalue distribution as some linear combination of \emph{cluster function} $T$,
\begin{align}
{{\rho }^{(n)}}({{\lambda }_{1}},\ldots ,{{\lambda }_{n}})=\frac{(L-n)!}{L!}\sum\limits_{m\ge 1}^{\bigcupdot\nolimits_{i=1}^{m}{{{S}_{i}}}=\left\{ 1,\ldots ,n \right\}}{{{(-1)}^{n-m}}{{T}_{{{S}_{1}}}}\ldots {{T}_{{{S}_{m}}}}}
\end{align}
where $T_S=T_l(x_{i_1},x_{i_2},\cdots,x_{i_l})$ and for $S=\left\{i_1,i_2,\cdots,i_l\right\}$, and the sum is over all possible decompositions of $\left\{ 1,\ldots ,n \right\}$ ($\bigcupdot$ means disjoint union). For instance
\begin{align}
{{\rho }^{(2)}}({{\lambda }_{1}},{{\lambda }_{2}})=\frac{1}{L(L-1)}\left( {{T}_{1}}({{\lambda }_{1}}){{T}_{1}}({{\lambda }_{2}})-{{T}_{2}}({{\lambda }_{1}},{{\lambda }_{2}}) \right)
\end{align}
The cluster function could be computed by the quaternion kernel as
\begin{align}
{{T}_{n}}({{\lambda }_{1}},\ldots ,{{\lambda }_{n}})=\frac{1}{2n}\sum\limits_{\sigma }{\text{Tr}\left( K({{\lambda }_{{{\sigma }_{1}}}},{{\lambda }_{{{\sigma }_{2}}}})K({{\lambda }_{{{\sigma }_{2}}}},{{\lambda }_{{{\sigma }_{3}}}})\ldots K({{\lambda }_{{{\sigma }_{n}}}},{{\lambda }_{{{\sigma }_{1}}}}) \right)}
\end{align}
where the sum is taken over all permutations $\sigma$ of $\left\{ 1,\ldots ,n \right\}$. Thus, from these computations, we could be in principle reduce the correlation functions into cluster functions, and then the products of trace of kernels which are essentially computable. There are some simplest examples for those formulas, for instance,
\begin{align}
& {{\rho }^{(1)}}({{\lambda }_{1}})=\frac{1}{L}\times \frac{1}{2}\text{Tr}\left( K({{\lambda }_{1}},{{\lambda }_{1}}) \right) \nonumber\\
& {{\rho }^{(2)}}({{\lambda }_{1}},{{\lambda }_{2}})=\frac{1}{L(L-1)}\times \left( \frac{1}{4}\text{Tr}\left( K({{\lambda }_{1}},{{\lambda }_{1}}) \right)\text{Tr}\left( K({{\lambda }_{2}},{{\lambda }_{2}}) \right)-\frac{1}{2}\text{Tr}\left( K({{\lambda }_{1}},{{\lambda }_{2}})K({{\lambda }_{2}},{{\lambda }_{1}}) \right) \right)
\end{align}
With these knowledge, we could start to compute spectral form factors. 
\subsection{Form factor computation with box approximation}
\subsubsection{Theorems}
It is straightforward to generalize our previous formula of convolution kernels to the quaternion matrix theory. We have  
\begin{thm}[Convolution formula for GOE]
We have 
\begin{align}
&\int{\prod\limits_{i=1}^{m}{d{{\lambda }_{i}}}}{K}({{\lambda }_{1}},{{\lambda }_{2}}){K}({{\lambda }_{2}},{{\lambda }_{3}})\ldots {K}({{\lambda }_{m-1}},{{\lambda }_{m}}){K}({{\lambda }_{m}},{{\lambda }_{1}}) \,e^{i \sum_{i=1}^{m} k_i \lambda_i} \nonumber\\
& = Lr_3(\sum_{i=1}^m k_i)\int{dk}\,G(k)G\Big(k+\frac{{{k}_{1}}}{2L}\Big) G\Big(k+\frac{{{k}_{2}}}{2L}\Big) \ldots G\Big(k+\frac{{{k}_{m-1}}}{2L}\Big)
\end{align}
where
\begin{align}
G(k)=\left( \begin{matrix}
   g(k) & (-2\pi i k)g(k)  \\
   \frac{g(k)-1}{-2\pi i k} & g(k)  \\
\end{matrix} \right)
\end{align}
\end{thm}
and 
\begin{thm}[Convolution formula for GSE]
We have 
\begin{align}
&\int{\prod\limits_{i=1}^{m}{d{{\lambda }_{i}}}}{K}({{\lambda }_{1}},{{\lambda }_{2}}){K}({{\lambda }_{2}},{{\lambda }_{3}})\ldots {K}({{\lambda }_{m-1}},{{\lambda }_{m}}){K}({{\lambda }_{m}},{{\lambda }_{1}}) \,e^{i \sum_{i=1}^{m} k_i \lambda_i} \nonumber\\
& = Lr_3(\sum_{i=1}^m k_i)\int{dk}\,H(k)H\Big(k+\frac{{{k}_{1}}}{2L}\Big) H\Big(k+\frac{{{k}_{2}}}{2L}\Big) \ldots H\Big(k+\frac{{{k}_{m-1}}}{2L}\Big)
\end{align}
where 
\begin{align}
H(k)=\frac{1}{2}g(\frac{k}{2})\left( \begin{matrix}
   1 & -\pi ik  \\
   \frac{1}{-\pi ik} & 1  \\
\end{matrix} \right)
\end{align}
\end{thm}
The original infinite $L$ versions of these formulas are hidden in eq.(6.4.21) and eq.(7.2.10) in \cite{book2}. 
\subsubsection{Two point form factor}
Based on our GUE knowledge, we will briefly describe how to compute form factors. 
\\
\\
We start by computing $\mathcal{R}_2$ at infinite temperature for GOE.
\begin{align}
{{\mathcal{R}}_{2}}(t)=L+\int{d}{{\lambda }_{1}}d{{\lambda }_{2}}\left( \frac{1}{4}\text{Tr}\left( K({{\lambda }_{1}},{{\lambda }_{1}}) \right)\text{Tr}\left( K({{\lambda }_{2}},{{\lambda }_{2}}) \right)-\frac{1}{2}\text{Tr}\left( K({{\lambda }_{1}},{{\lambda }_{2}})K({{\lambda }_{2}},{{\lambda }_{1}}) \right) \right){{e}^{i({{\lambda }_{1}}-{{\lambda }_{2}})t}}
\end{align}
Evaluating the first term in the integral, we find
\begin{align}
\int{d}{{\lambda }_{1}}d{{\lambda }_{2}}\left( \frac{1}{4}\text{Tr}\left( K({{\lambda }_{1}},{{\lambda }_{1}}) \right)\text{Tr}\left( K({{\lambda }_{2}},{{\lambda }_{2}}) \right) \right){{e}^{i({{\lambda }_{1}}-{{\lambda }_{2}})t}}={{L}^{2}} {r}_{1}^{2}(t)
\end{align}
The second term can be evaluated as
\begin{equation}
\int{d}{{\lambda }_{1}}d{{\lambda }_{2}}\left( \frac{1}{2}\text{Tr}\left( K({{\lambda }_{1}},{{\lambda }_{2}})K({{\lambda }_{2}},{{\lambda }_{1}}) \right) \right){{e}^{i({{\lambda }_{1}}-{{\lambda }_{2}})t}} = L r_2(t)
\end{equation}
where
\begin{align}
r_2(t)=\left\{ \begin{matrix}
   1-\frac{t}{L}+\frac{t}{2L}\log \left( 1+\frac{t}{L} \right) & t<2L  \\
   -1+\frac{t}{2L}\log \left( \frac{t+L}{t-L} \right) & t>2L  \\
\end{matrix} \right.
\end{align}
The final result is
\begin{align}
{{\mathcal {R}}_{2}}(t)=L+{{L}^{2}} {r}_{1}^{2}(t)-Lr_2(t)
\end{align}
In GSE, the only difference between GOE and GSE is
\begin{align}
\int{d}{{\lambda }_{1}}d{{\lambda }_{2}}\left( \frac{1}{2}\text{Tr}\left( K({{\lambda }_{1}},{{\lambda }_{2}})K({{\lambda }_{2}},{{\lambda }_{1}}) \right) \right){{e}^{i({{\lambda }_{1}}-{{\lambda }_{2}})t}} = L\left\{ \begin{matrix}
   1-\frac{1}{2}\cdot \frac{t}{2L}+\frac{1}{4}\cdot \frac{t}{2L}\cdot \log \left| 1-\frac{t}{2L} \right| & t<4L  \\
   0 & t>4L  \\
\end{matrix} \right.
\end{align}
This integration is, in fact, divergent between $0<t\le 2L$. It is because of a pole $1/k$ in the expression of $H(k)$. However, that is an artifact of the Fourier transformation of the integral of the sine kernel $\text{sinc}(x)$. Besides the methods of explicitly computing the Fourier transform, we could also understand the time before $2L$ as a continuation. As a result, there is a pole at $t=2L$.
\\
\\
So as a conclusion, in GSE we have to replace the result of $r_2$ by
\begin{align}
{{r }_{2}}=\left\{ \begin{matrix}
   1-\frac{t}{4L}+\frac{t}{8L}\log \left| 1-\frac{t}{2L} \right| & t<4L  \\
   0 & t>4L  \\
\end{matrix} \right.
\end{align}
\subsubsection{Four point form factor}
In this part we need to compute $\mathcal{R}_4$ in GOE, which is
\begin{align}
{\mathcal{R}_{4}}=\sum\limits_{a,b,c,d=1}^{L}{\int{D\lambda}{{e}^{i({{\lambda }_{a}}+{{\lambda }_{b}}-{{\lambda }_{c}}-{{\lambda }_{d}})t}}}
\end{align}
Take a look at the classifications of combinations in ${\mathcal{R}_{4}}$, which is
\begin{itemize}
\item $a=b=c=d=e=f$: Contribute $L$.
\item $a=b$: Contribute $L(L-1)(L-2)\int{D\lambda}{{e}^{i(2{{\lambda }_{1}}-{{\lambda }_{2}}-{{\lambda }_{3}})t}}$.
\item $c=d$: Contribute $L(L-1)(L-2)\int{D\lambda}{{e}^{i({{\lambda }_{1}}+{{\lambda }_{2}}-2{{\lambda }_{3}})t}}$.
\item $a=c$ or $a=d$ or $b=c$ or $b=d$: Contribute $4L(L-1)(L-2)\int{D\lambda}{{e}^{i({{\lambda }_{1}}-{{\lambda }_{2}})t}}$.
\item $b=c=d$ or $a=c=d$ or $a=b=d$ or $a=b=c$: Contribute $4L(L-1)\int{D\lambda}{{e}^{i({{\lambda }_{1}}-{{\lambda }_{2}})t}}$.
\item $a=b$ and $c=d$:  Contribute $L(L-1)\int{D\lambda}{{e}^{i(2{{\lambda }_{1}}-2{{\lambda }_{2}})t}}$.
\item $a=c$ and $b=d$, or $a=d$ and $b=c$: Contribute $2L(L-1)$.
\item All inequal indexes: $L(L-1)(L-2)(L-3)\int{D\lambda}{{e}^{i({{\lambda }_{1}}+{{\lambda }_{2}}-{{\lambda }_{3}}-{{\lambda }_{4}})t}}$
\end{itemize}
Add the total prefactors will give $L^4$. Add them together we get
\begin{align}
  & {\mathcal{R}_{4}}=L(L-1)(L-2)(L-3)\int{D\lambda}{{e}^{i({{\lambda }_{1}}+{{\lambda }_{2}}-{{\lambda }_{3}}-{{\lambda }_{4}})t}} \nonumber\\
 & +2L(L-1)(L-2)\operatorname{Re}\int{D\lambda}{{e}^{i(2{{\lambda }_{1}}-{{\lambda }_{2}}-{{\lambda }_{3}})t}}  \nonumber\\
 & +L(L-1)\int{D\lambda}{{e}^{i(2{{\lambda }_{1}}-2{{\lambda }_{2}})t}}  \nonumber\\
 & +4L{{(L-1)}^{2}}\int{D\lambda}{{e}^{i({{\lambda }_{1}}-{{\lambda }_{2}})t}}  \nonumber\\
 & +2{{L}^{2}}-L  \nonumber\\
 & =L(L-1)(L-2)(L-3)\int{D\lambda}{{e}^{i({{\lambda }_{1}}+{{\lambda }_{2}}-{{\lambda }_{3}}-{{\lambda }_{4}})t}}  \nonumber\\
 & +2L(L-1)(L-2)\operatorname{Re}\int{D\lambda}{{e}^{i(2{{\lambda }_{1}}-{{\lambda }_{2}}-{{\lambda }_{3}})t}}  \nonumber\\
 & +{{L}^{2}}{r}_{1}^{2}(2t)-Lr_2(2t) \nonumber\\
 & +4(L-1)({{L}^{2}}r_{1}^{2}(t)-L{r_{2}}(t))  \nonumber\\
 & +2{{L}^{2}}-L
\end{align}
We already obtained what the last three terms are. Now we only need to consider the first two terms.
\paragraph{The first term}
The first term is an actual four point function.
\begin{align}
L(L-1)(L-2)(L-3)\int{D\lambda}{{e}^{i({{\lambda }_{1}}+{{\lambda }_{2}}-{{\lambda }_{3}}-{{\lambda }_{4}})t}}
\end{align}
In order to compute it we will use the following decomposition from the correlation function to cluster function
\begin{align}
  & L(L-1)(L-2)(L-3){{\rho }^{(4)}}({{\lambda }_{1}},{{\lambda }_{2}},{{\lambda }_{3}},{{\lambda }_{4}}) \nonumber\\
 & =-{{T}_{4}}({{\lambda }_{1}},{{\lambda }_{2}},{{\lambda }_{3}},{{\lambda }_{4}})  \nonumber\\
 & +{{T}_{3}}({{\lambda }_{2}},{{\lambda }_{3}},{{\lambda }_{4}}){{T}_{1}}({{\lambda }_{1}})+{{T}_{3}}({{\lambda }_{1}},{{\lambda }_{3}},{{\lambda }_{4}}){{T}_{1}}({{\lambda }_{2}})+{{T}_{3}}({{\lambda }_{1}},{{\lambda }_{2}},{{\lambda }_{4}}){{T}_{1}}({{\lambda }_{3}})+{{T}_{3}}({{\lambda }_{1}},{{\lambda }_{2}},{{\lambda }_{3}}){{T}_{1}}({{\lambda }_{4}})  \nonumber\\
 & +{{T}_{2}}({{\lambda }_{1}},{{\lambda }_{2}}){{T}_{2}}({{\lambda }_{3}},{{\lambda }_{4}})+{{T}_{2}}({{\lambda }_{1}},{{\lambda }_{3}}){{T}_{2}}({{\lambda }_{2}},{{\lambda }_{4}})+{{T}_{2}}({{\lambda }_{1}},{{\lambda }_{4}}){{T}_{2}}({{\lambda }_{2}},{{\lambda }_{3}})  \nonumber\\
 & -{{T}_{1}}({{\lambda }_{1}}){{T}_{1}}({{\lambda }_{2}}){{T}_{2}}({{\lambda }_{3}},{{\lambda }_{4}})-{{T}_{1}}({{\lambda }_{1}}){{T}_{1}}({{\lambda }_{3}}){{T}_{2}}({{\lambda }_{2}},{{\lambda }_{4}})-{{T}_{1}}({{\lambda }_{1}}){{T}_{1}}({{\lambda }_{4}}){{T}_{2}}({{\lambda }_{2}},{{\lambda }_{3}})  \nonumber\\
 & -{{T}_{1}}({{\lambda }_{2}}){{T}_{1}}({{\lambda }_{3}}){{T}_{2}}({{\lambda }_{1}},{{\lambda }_{4}})-{{T}_{1}}({{\lambda }_{2}}){{T}_{1}}({{\lambda }_{4}}){{T}_{2}}({{\lambda }_{1}},{{\lambda }_{3}})-{{T}_{1}}({{\lambda }_{3}}){{T}_{1}}({{\lambda }_{4}}){{T}_{2}}({{\lambda }_{1}},{{\lambda }_{2}})  \nonumber\\
 & +{{T}_{1}}({{\lambda }_{1}}){{T}_{1}}({{\lambda }_{2}}){{T}_{1}}({{\lambda }_{3}}){{T}_{1}}({{\lambda }_{4}})
\end{align}
From the previous discussions, we have
\begin{align}
  & {{T}_{1}}({{\lambda }_{a}})=\frac{1}{2}\text{Tr}\left( K({{\lambda }_{a}},{{\lambda }_{a}}) \right)  \nonumber\\
 & {{T}_{2}}({{\lambda }_{a}},{{\lambda }_{b}})=\frac{1}{2}\text{Tr}\left( K({{\lambda }_{a}},{{\lambda }_{b}})K({{\lambda }_{b}},{{\lambda }_{a}}) \right)  \nonumber\\
 & {{T}_{3}}({{\lambda }_{a}},{{\lambda }_{b}},{{\lambda }_{c}})=\frac{1}{2}\text{Tr}\left( K({{\lambda }_{a}},{{\lambda }_{b}})K({{\lambda }_{b}},{{\lambda }_{c}})K({{\lambda }_{c}},{{\lambda }_{a}}) \right)+\frac{1}{2}\text{Tr}\left( K({{\lambda }_{a}},{{\lambda }_{c}})K({{\lambda }_{c}},{{\lambda }_{b}})K({{\lambda }_{b}},{{\lambda }_{a}}) \right)  \nonumber\\
 & {{T}_{4}}({{\lambda }_{a}},{{\lambda }_{b}},{{\lambda }_{c}},{{\lambda }_{d}})=\frac{1}{2}\text{Tr}\left( K({{\lambda }_{a}},{{\lambda }_{b}})K({{\lambda }_{b}},{{\lambda }_{c}})K({{\lambda }_{c}},{{\lambda }_{d}})K({{\lambda }_{d}},{{\lambda }_{a}}) \right)  \nonumber\\
 & +\frac{1}{2}\text{Tr}\left( K({{\lambda }_{a}},{{\lambda }_{b}})K({{\lambda }_{b}},{{\lambda }_{d}})K({{\lambda }_{d}},{{\lambda }_{c}})K({{\lambda }_{c}},{{\lambda }_{a}}) \right)  \nonumber\\
 & +\frac{1}{2}\text{Tr}\left( K({{\lambda }_{a}},{{\lambda }_{c}})K({{\lambda }_{c}},{{\lambda }_{b}})K({{\lambda }_{b}},{{\lambda }_{d}})K({{\lambda }_{d}},{{\lambda }_{a}}) \right)  \nonumber\\
 & +\frac{1}{2}\text{Tr}\left( K({{\lambda }_{a}},{{\lambda }_{c}})K({{\lambda }_{c}},{{\lambda }_{d}})K({{\lambda }_{d}},{{\lambda }_{b}})K({{\lambda }_{b}},{{\lambda }_{a}}) \right)  \nonumber\\
 & +\frac{1}{2}\text{Tr}\left( K({{\lambda }_{a}},{{\lambda }_{d}})K({{\lambda }_{d}},{{\lambda }_{b}})K({{\lambda }_{b}},{{\lambda }_{c}})K({{\lambda }_{c}},{{\lambda }_{a}}) \right)  \nonumber\\
 & +\frac{1}{2}\text{Tr}\left( K({{\lambda }_{a}},{{\lambda }_{d}})K({{\lambda }_{d}},{{\lambda }_{c}})K({{\lambda }_{c}},{{\lambda }_{b}})K({{\lambda }_{b}},{{\lambda }_{a}}) \right)
\end{align}
where we have already used the property of cyclic invariance for the trace operator. We can separately discuss these terms as the following,
\begin{itemize}
\item 4-type: In this case we only have the $T_4$, In this case, all six terms in the expansion give the same answer, which is
\begin{align}
-\int{d}{{\lambda }_{1}}d{{\lambda }_{2}}d{{\lambda }_{3}}d{{\lambda }_{4}}{{T}_{4}}({{\lambda }_{1}},{{\lambda }_{2}},{{\lambda }_{3}},{{\lambda }_{4}}){{e}^{i({{\lambda }_{1}}+{{\lambda }_{2}}-{{\lambda }_{3}}-{{\lambda }_{4}})t}}=-6L{{r }_{4}}(t)
\end{align}
where
\begin{align}
{{r }_{4}}(t)=\left\{ \begin{matrix}
   1-\frac{7t}{4L}+\frac{5t}{4L}\log \left( 1+\frac{t}{L} \right) & 0<t<L  \\
   -\frac{3}{2}+\frac{3t}{4L}+\frac{5t}{4L}\log \left( \frac{1}{2}\left( 1+\frac{3L}{2t-L} \right) \right) & L<t<2L  \\
   0 & t>2L  \\
\end{matrix} \right.
\end{align}
\item 1-1-1-1-type: In this case we only have $T_1$. Thus we have
\begin{align}
\int{d}{{\lambda }_{1}}d{{\lambda }_{2}}d{{\lambda }_{3}}d{{\lambda }_{4}}{{T}_{1}}({{\lambda }_{1}}){{T}_{1}}({{\lambda }_{2}}){{T}_{{1}}}({{\lambda }_{3}}){{T}_{1}}({{\lambda }_{4}}){{e}^{i({{\lambda }_{1}}+{{\lambda }_{2}}-{{\lambda }_{3}}-{{\lambda }_{4}})t}}={{L}^{4}}r_{1}^{4}(t)
\end{align}
\item 1-1-2-type: In this case we have both $T_1$ and $T_2$. Thus we have
\begin{align}
  & -\int{d}{{\lambda }_{1}}d{{\lambda }_{2}}d{{\lambda }_{3}}d{{\lambda }_{4}}{{T}_{1}}({{\lambda }_{1}}){{T}_{1}}({{\lambda }_{2}}){{T}_{2}}({{\lambda }_{3}},{{\lambda }_{4}}){{e}^{i({{\lambda }_{1}}+{{\lambda }_{2}}-{{\lambda }_{3}}-{{\lambda }_{4}})t}} \nonumber\\
 & -\int{d}{{\lambda }_{1}}d{{\lambda }_{2}}d{{\lambda }_{3}}d{{\lambda }_{4}}{{T}_{1}}({{\lambda }_{1}}){{T}_{1}}({{\lambda }_{3}}){{T}_{2}}({{\lambda }_{2}},{{\lambda }_{4}}){{e}^{i({{\lambda }_{1}}+{{\lambda }_{2}}-{{\lambda }_{3}}-{{\lambda }_{4}})t}}  \nonumber\\
 & -\int{d}{{\lambda }_{1}}d{{\lambda }_{2}}d{{\lambda }_{3}}d{{\lambda }_{4}}{{T}_{1}}({{\lambda }_{1}}){{T}_{1}}({{\lambda }_{4}}){{T}_{2}}({{\lambda }_{2}},{{\lambda }_{3}}){{e}^{i({{\lambda }_{1}}+{{\lambda }_{2}}-{{\lambda }_{3}}-{{\lambda }_{4}})t}}  \nonumber\\
 & -\int{d}{{\lambda }_{1}}d{{\lambda }_{2}}d{{\lambda }_{3}}d{{\lambda }_{4}}{{T}_{1}}({{\lambda }_{2}}){{T}_{1}}({{\lambda }_{3}}){{T}_{2}}({{\lambda }_{1}},{{\lambda }_{4}}){{e}^{i({{\lambda }_{1}}+{{\lambda }_{2}}-{{\lambda }_{3}}-{{\lambda }_{4}})t}}  \nonumber\\
 & -\int{d}{{\lambda }_{1}}d{{\lambda }_{2}}d{{\lambda }_{3}}d{{\lambda }_{4}}{{T}_{1}}({{\lambda }_{2}}){{T}_{1}}({{\lambda }_{4}}){{T}_{2}}({{\lambda }_{1}},{{\lambda }_{3}}){{e}^{i({{\lambda }_{1}}+{{\lambda }_{2}}-{{\lambda }_{3}}-{{\lambda }_{4}})t}}  \nonumber\\
 & -\int{d}{{\lambda }_{1}}d{{\lambda }_{2}}d{{\lambda }_{3}}d{{\lambda }_{4}}{{T}_{1}}({{\lambda }_{3}}){{T}_{1}}({{\lambda }_{4}}){{T}_{2}}({{\lambda }_{1}},{{\lambda }_{2}}){{e}^{i({{\lambda }_{1}}+{{\lambda }_{2}}-{{\lambda }_{3}}-{{\lambda }_{4}})t}} \nonumber\\
 & =-2{{L}^{3}}r_{1}^{2}(t){{r}_{3}}(2t){{r }_{2}}(t)-4{{L}^{3}}r_{1}^{2}(t){{r }_{2}}(t)
\end{align}
\item 2-2-type: In this case we only have $T_2$. Thus we have
\begin{align}
  & +\int{d}{{\lambda }_{1}}d{{\lambda }_{2}}d{{\lambda }_{3}}d{{\lambda }_{4}}{{T}_{2}}({{\lambda }_{1}},{{\lambda }_{2}}){{T}_{2}}({{\lambda }_{3}},{{\lambda }_{4}}){{e}^{i({{\lambda }_{1}}+{{\lambda }_{2}}-{{\lambda }_{3}}-{{\lambda }_{4}})t}} \nonumber\\
 & +\int{d}{{\lambda }_{1}}d{{\lambda }_{2}}d{{\lambda }_{3}}d{{\lambda }_{4}}{{T}_{2}}({{\lambda }_{1}},{{\lambda }_{3}}){{T}_{2}}({{\lambda }_{2}},{{\lambda }_{4}}){{e}^{i({{\lambda }_{1}}+{{\lambda }_{2}}-{{\lambda }_{3}}-{{\lambda }_{4}})t}} \nonumber\\
 & +\int{d}{{\lambda }_{1}}d{{\lambda }_{2}}d{{\lambda }_{3}}d{{\lambda }_{4}}{{T}_{2}}({{\lambda }_{1}},{{\lambda }_{4}}){{T}_{2}}({{\lambda }_{2}},{{\lambda }_{3}}){{e}^{i({{\lambda }_{1}}+{{\lambda }_{2}}-{{\lambda }_{3}}-{{\lambda }_{4}})t}} \nonumber\\
 & =2{{L}^{2}}r _{2}^{2}(t)+{{L}^{2}}r _{2}^{2}(t)r_{3}^{2}(2t)
\end{align}
\item 3-1-type: In this case we both have $T_3$ and $T_1$. Thus we have
\begin{align}
  & +\int{d}{{\lambda }_{1}}d{{\lambda }_{2}}d{{\lambda }_{3}}d{{\lambda }_{4}}{{T}_{3}}({{\lambda }_{2}},{{\lambda }_{3}},{{\lambda }_{4}}){{T}_{1}}({{\lambda }_{1}}){{e}^{i({{\lambda }_{1}}+{{\lambda }_{2}}-{{\lambda }_{3}}-{{\lambda }_{4}})t}} \nonumber\\
 & +\int{d}{{\lambda }_{1}}d{{\lambda }_{2}}d{{\lambda }_{3}}d{{\lambda }_{4}}{{T}_{3}}({{\lambda }_{1}},{{\lambda }_{3}},{{\lambda }_{4}}){{T}_{1}}({{\lambda }_{2}}){{e}^{i({{\lambda }_{1}}+{{\lambda }_{2}}-{{\lambda }_{3}}-{{\lambda }_{4}})t}}  \nonumber\\
 & +\int{d}{{\lambda }_{1}}d{{\lambda }_{2}}d{{\lambda }_{3}}d{{\lambda }_{4}}{{T}_{3}}({{\lambda }_{1}},{{\lambda }_{2}},{{\lambda }_{4}}){{T}_{1}}({{\lambda }_{3}}){{e}^{i({{\lambda }_{1}}+{{\lambda }_{2}}-{{\lambda }_{3}}-{{\lambda }_{4}})t}}  \nonumber\\
 & +\int{d}{{\lambda }_{1}}d{{\lambda }_{2}}d{{\lambda }_{3}}d{{\lambda }_{4}}{{T}_{3}}({{\lambda }_{1}},{{\lambda }_{2}},{{\lambda }_{3}}){{T}_{1}}({{\lambda }_{4}}){{e}^{i({{\lambda }_{1}}+{{\lambda }_{2}}-{{\lambda }_{3}}-{{\lambda }_{4}})t}}  \nonumber\\
 & =6{{L}^{2}}{{r}_{1}}(t){{r }_{3,1}}(t){{r}_{3}}(t)+2{{L}^{2}}{{r}_{1}}(t){{r }_{3,2}}(t){{r}_{3}}(t)
\end{align}
where
\begin{align}
{{r }_{3,1}}(t)=\left\{ \begin{matrix}
   1-\frac{2t}{L}+\frac{3t}{2L}\log (1+\frac{t}{L}) & 0<t<L  \\
   -2+\frac{t}{L}+\frac{3t}{2L}\log \left( \frac{1}{2}\left( 1+\frac{3L}{2t-L} \right) \right) & L<t<2L  \\
   0 & t>2L  \\
\end{matrix} \right.
\end{align}
and
\begin{align}
r_{3,2}(t)=r_2(t)
\end{align}
\end{itemize}
So we have the total expression as
\begin{align}
& L(L-1)(L-2)(L-3)\int{D\lambda }{{e}^{i({{\lambda }_{1}}+{{\lambda }_{2}}-{{\lambda }_{3}}-{{\lambda }_{4}})t}} \nonumber\\
& ={{L}^{4}}r_{1}^{4}(t) \nonumber\\
 & -2{{L}^{3}}r_{1}^{2}(t){{r}_{3}}(2t){{r }_{2}}(t)-4{{L}^{3}}r_{1}^{2}(t){{r }_{2}}(t) \nonumber\\
 & +2{{L}^{2}}r _{2}^{2}(t)+{{L}^{2}}r _{2}^{2}(t)r_{3}^{2}(2t)+6{{L}^{2}}{{r}_{1}}(t){{r }_{3,1}}(t){{r}_{3}}(t)+2{{L}^{2}}{{r}_{1}}(t){{r }_{3,2}}(t){{r}_{3}}(t) \nonumber\\
 & -6L{{r }_{4}}(t)
\end{align}
\paragraph{The second term}
In this part we will evaluate the second term
\begin{align}
2L(L-1)(L-2)\operatorname{Re}\int{D\lambda}{{e}^{i(2{{\lambda }_{1}}-{{\lambda }_{2}}-{{\lambda }_{3}})t}}
\end{align}
Let us firstly consider it without a factor of 2
\begin{align}
L(L-1)(L-2)\operatorname{Re}\int{D\lambda}{{e}^{i(2{{\lambda }_{1}}-{{\lambda }_{2}}-{{\lambda }_{3}})t}}
\end{align}
Do the same cluster decomposition
\begin{align}
  & L(L-1)(L-2){{\rho }^{(3)}}({{\lambda }_{1}},{{\lambda }_{2}},{{\lambda }_{3}}) \nonumber\\
 & ={{T}_{3}}({{\lambda }_{1}},{{\lambda }_{2}},{{\lambda }_{3}})-{{T}_{2}}({{\lambda }_{2}},{{\lambda }_{3}}){{T}_{1}}({{\lambda }_{1}})-{{T}_{2}}({{\lambda }_{1}},{{\lambda }_{3}}){{T}_{1}}({{\lambda }_{2}})-{{T}_{2}}({{\lambda }_{1}},{{\lambda }_{2}}){{T}_{1}}({{\lambda }_{3}}) \nonumber\\
 & +{{T}_{1}}({{\lambda }_{1}}){{T}_{1}}({{\lambda }_{2}}){{T}_{1}}({{\lambda }_{3}})
\end{align}
Then we obtain
\begin{itemize}
\item 3-type: In this case we have
\begin{align}
\operatorname{Re}\int{d}{{\lambda }_{1}}d{{\lambda }_{2}}d{{\lambda }_{3}}{{T}_{3}}({{\lambda }_{1}},{{\lambda }_{2}},{{\lambda }_{3}}){{e}^{i(2{{\lambda }_{1}}-{{\lambda }_{2}}-{{\lambda }_{3}})t}}=2L{{r }_{3,3}}(t)
\end{align}
where
\begin{align}
{{r }_{3,3}}(t)=\left\{ \begin{matrix}
   1-\frac{3t}{L}+\frac{t}{L}\log (1+\frac{t}{L})+\frac{5t}{4L}\log (1+\frac{2t}{L}) & 0<t<\frac{2}{3}L  \\
   -2+\frac{3t}{2L}+\frac{t}{L}\log \left( \frac{L+t}{3t-L} \right)+\frac{5t}{4L}\log \left( \frac{L+2t}{3t-L} \right) & \frac{2}{3}L<t<L  \\
   -1+\frac{t}{2L}+\frac{5t}{4L}\log \left( \frac{L+2t}{3t-L} \right) & L<t<2L  \\
   0 & t>2L  \\
\end{matrix} \right.
\end{align}
\item 1-1-1-type: In this case we have
\begin{align}
\operatorname{Re}\int{d}{{\lambda }_{1}}d{{\lambda }_{2}}d{{\lambda }_{3}}{{T}_{1}}({{\lambda }_{1}}){{T}_{1}}({{\lambda }_{2}}){{T}_{1}}({{\lambda }_{3}}){{e}^{i(2{{\lambda }_{1}}-{{\lambda }_{2}}-{{\lambda }_{3}})t}}={{L}^{3}}{{r}_{1}}(2t)r_{1}^{2}(t)
\end{align}
\item 2-1-type: In this case we have
\begin{align}
  & -\operatorname{Re}\int{d}{{\lambda }_{1}}d{{\lambda }_{2}}d{{\lambda }_{3}}{{T}_{2}}({{\lambda }_{2}},{{\lambda }_{3}}){{T}_{1}}({{\lambda }_{1}}){{e}^{i(2{{\lambda }_{1}}-{{\lambda }_{2}}-{{\lambda }_{3}})t}} \nonumber\\
 & -\operatorname{Re}\int{d}{{\lambda }_{1}}d{{\lambda }_{2}}d{{\lambda }_{3}}{{T}_{2}}({{\lambda }_{1}},{{\lambda }_{3}}){{T}_{1}}({{\lambda }_{2}}){{e}^{i(2{{\lambda }_{1}}-{{\lambda }_{2}}-{{\lambda }_{3}})t}} \nonumber\\
 & -\operatorname{Re}\int{d}{{\lambda }_{1}}d{{\lambda }_{2}}d{{\lambda }_{3}}{{T}_{2}}({{\lambda }_{1}},{{\lambda }_{2}}){{T}_{1}}({{\lambda }_{3}}){{e}^{i(2{{\lambda }_{1}}-{{\lambda }_{2}}-{{\lambda }_{3}})t}} \nonumber\\
 & =-{{L}^{2}}{{r}_{1}}(2t){{r}_{3}}(2t){{r }_{2}}(t)-2{{L}^{2}}{{r}_{1}}(t){{r}_{3}}(t){{r }_{2}}(2t)
\end{align}
\end{itemize}
Finally, we can collect all those terms together, and then we get
\begin{align}
  & 2L(L-1)(L-2)\operatorname{Re}\int{D\lambda }{{e}^{i(2{{\lambda }_{1}}-{{\lambda }_{2}}-{{\lambda }_{3}})t}} \nonumber\\
 & =2{{L}^{3}}{{r}_{1}}(2t)r_{1}^{2}(t)-2{{L}^{2}}{{r}_{1}}(2t){{r}_{3}}(2t){{r }_{2}}(t)-4{{L}^{2}}{{r}_{1}}(t){{r}_{3}}(t){{r }_{2}}(2t)+4L{{r }_{3,3}}(t)
\end{align}
\paragraph{Final expression and summary}
From those calculations we could obtain the final expression for $\mathcal{R}_4$, which is
\begin{align}
  & {{\mathcal{R}}_{4}}=+{{L}^{4}}r_{1}^{4}(t) \nonumber\\
 & -2{{L}^{3}}r_{1}^{2}(t){{r }_{2}}(t){{r}_{3}}(2t)-4{{L}^{3}}r_{1}^{2}(t){{r }_{2}}(t)+2{{L}^{3}}{{r}_{1}}(2t)r_{1}^{2}(t)+4{{L}^{3}}r_{1}^{2}(t) \nonumber\\
 & +2{{L}^{2}}r _{2}^{2}(t)+{{L}^{2}}r _{2}^{2}(t)r_{3}^{2}(2t)+6{{L}^{2}}{{r}_{1}}(t){{r }_{3,1}}(t){{r}_{3}}(t)+2{{L}^{2}}{{r}_{1}}(t){{r }_{3,2}}(t){{r}_{3}}(t) \nonumber\\
 & -2{{L}^{2}}{{r}_{1}}(2t){{r}_{3}}(2t){{r }_{2}}(t)-4{{L}^{2}}{{r}_{1}}(t){{r}_{3}}(t){{r }_{2}}(2t)+{{L}^{2}}r_{1}^{2}(2t)-4{{L}^{2}}r_{1}^{2}(t)-4{{L}^{2}}{{r }_{2}}(t)+2{{L}^{2}} \nonumber\\
 & -6L{{r }_{4}}(t)-L{{r }_{2}}(2t)+4L{{r }_{3,3}}(t)+4L{{r }_{2}}(t)-L
\end{align}
where
\begin{align}
r_2(t)=\left\{ \begin{matrix}
   1-\frac{t}{L}+\frac{t}{2L}\log \left( 1+\frac{t}{L} \right) & t<2L  \\
   -1+\frac{t}{2L}\log \left( \frac{t+L}{t-L} \right) & t>2L  \\
\end{matrix} \right.
\end{align}
\begin{align}
{{r }_{3,1}}(t)=\left\{ \begin{matrix}
   1-\frac{2t}{L}+\frac{3t}{2L}\log (1+\frac{t}{L}) & 0<t<L  \\
   -2+\frac{t}{L}+\frac{3t}{2L}\log \left( \frac{1}{2}\left( 1+\frac{3L}{2t-L} \right) \right) & L<t<2L  \\
   0 & t>2L  \\
\end{matrix} \right.
\end{align}
\begin{align}
r_{3,2}(t)=r_2(t)
\end{align}
\begin{align}
{{r }_{3,3}}(t)=\left\{ \begin{matrix}
   1-\frac{3t}{L}+\frac{t}{L}\log (1+\frac{t}{L})+\frac{5t}{4L}\log (1+\frac{2t}{L}) & 0<t<\frac{2}{3}L  \\
   -2+\frac{3t}{2L}+\frac{t}{L}\log \left( \frac{L+t}{3t-L} \right)+\frac{5t}{4L}\log \left( \frac{L+2t}{3t-L} \right) & \frac{2}{3}L<t<L  \\
   -1+\frac{t}{2L}+\frac{5t}{4L}\log \left( \frac{L+2t}{3t-L} \right) & L<t<2L  \\
   0 & t>2L  \\
\end{matrix} \right.
\end{align}
\begin{align}
{{r }_{4}}(t)=\left\{ \begin{matrix}
   1-\frac{7t}{4L}+\frac{5t}{4L}\log \left( 1+\frac{t}{L} \right) & 0<t<L  \\
   -\frac{3}{2}+\frac{3t}{4L}+\frac{5t}{4L}\log \left( \frac{1}{2}\left( 1+\frac{3L}{2t-L} \right) \right) & L<t<2L  \\
   0 & t>2L  \\
\end{matrix} \right.
\end{align}
After dropping out the less-dominated terms, we could obtain
\begin{align}
  & {{\mathcal{R}}_{4}}\sim +{{L}^{4}}r_{1}^{4}(t) \nonumber\\
 & +2{{L}^{2}}r _{2}^{2}(t)-4{{L}^{2}}{{r }_{2}}(t)+2{{L}^{2}} \nonumber\\
 & -6L{{r }_{4}}(t)-L{{r }_{2}}(2t)+4L{{r }_{3,3}}(t)+4L{{r }_{2}}(t)-L
\end{align}
\paragraph{GSE}
In GSE the computations are very similar, and we have to replace these block functions by
\begin{align}
{{r }_{4}}(t)=\left\{ \begin{matrix}
   1-\frac{1}{2}\frac{t}{L}+\frac{3}{16}\frac{t}{L}\log \left| \frac{t}{L}-1 \right| & t<2L  \\
   0 & t>2L  \\
\end{matrix} \right.
\end{align}
\begin{align}
r_{3,1}=r_4
\end{align}
\begin{align}
r_{3,2}=r_2
\end{align}
\begin{align}
{{r }_{3,3}}=\left\{ \begin{matrix}
   1-\frac{3t}{4L}+\frac{t}{32L}\log \left| \frac{2t-2L}{2L-t} \right|+\frac{9t}{32L}\log |\frac{3}{2}\frac{t}{L}-1| & t<\frac{4}{3}L  \\
   0 & t>\frac{4}{3}L  \\
\end{matrix} \right.
\end{align}
\subsection{Refined two point form factor}
Now we discuss the trick that is similar with our previous improvement. Let us start from GOE. We will use the short distance refined kernel,
\begin{align}
&\tilde{ K}({\lambda _i},{\lambda _j}) \equiv L\rho((\lambda_i+\lambda_j)/2)\times \nonumber\\
&\left( {\begin{array}{*{20}{c}}
{\hat s(\pi L\rho (({\lambda _i} + {\lambda _j})/2)({\lambda _i} - {\lambda _j}))}&{{\bf{D}}\hat s(\pi L\rho (({\lambda _i} + {\lambda _j})/2)({\lambda _i} - {\lambda _j}))}\\
{{\bf{I}}\hat s(\pi L\rho (({\lambda _i} + {\lambda _j})/2)({\lambda _i} - {\lambda _j})) - \epsilon (\pi L\rho (({\lambda _i} + {\lambda _j})/2)({\lambda _i} - {\lambda _j}))}&{\hat s(\pi L\rho (({\lambda _i} + {\lambda _j})/2)({\lambda _i} - {\lambda _j}))}
\end{array}} \right)
\end{align}
We will try to use this formula to evaluate the two point form factor, at a generic finite temperature, $\beta$ (while the refined infinite temperature form factor could be obtained by sending $\beta\to 0$). We have
\begin{align}
&{{\cal R}_2}(t,\beta )\nonumber\\
&=L{r_1}(2i\beta ) + {L^2}{r_1}(t + i\beta ){r_1}(t - i\beta )\nonumber\\
&- \frac{1}{2}\int d {\lambda _1}d{\lambda _2}\left( {{\rm{Tr}}\left( {K({\lambda _1},{\lambda _2})K({\lambda _2},{\lambda _1})} \right)} \right){e^{i({\lambda _1} - {\lambda _2})t}}{e^{ - \beta ({\lambda _1} + {\lambda _2})}}
\end{align}
while for the later integral we expand it as
\begin{align}
&- \frac{1}{2}\int d {\lambda _1}d{\lambda _2}\left( {{\rm{Tr}}\left( {K({\lambda _1},{\lambda _2})K({\lambda _2},{\lambda _1})} \right)} \right){e^{i({\lambda _1} - {\lambda _2})t}}{e^{ - \beta ({\lambda _1} + {\lambda _2})}}\nonumber\\
&=- {L^2}\int d {\lambda _1}d{\lambda _2}\frac{{{{\sin }^2}(\pi L\rho (({\lambda _1} + {\lambda _2})/2)({\lambda _1} - {\lambda _2}))}}{{{{(\pi L({\lambda _1} - {\lambda _2}))}^2}}}{e^{i({\lambda _1} - {\lambda _2})t}}{e^{ - \beta ({\lambda _1} + {\lambda _2})}}\nonumber\\
&+ {L^2}\int d {\lambda _1}d{\lambda _2}{\rho ^2}(\frac{{{\lambda _1} + {\lambda _2}}}{2})\left( {{\bf{D}}\hat s(\pi L\rho (\frac{{{\lambda _1} + {\lambda _2}}}{2})({\lambda _1} - {\lambda _2})){\bf{I}}\hat s(\pi L\rho (\frac{{{\lambda _1} + {\lambda _2}}}{2})({\lambda _i} - {\lambda _j}))} \right){e^{i({\lambda _1} - {\lambda _2})t}}{e^{ - \beta ({\lambda _1} + {\lambda _2})}}\nonumber\\
&- {L^2}\int d {\lambda _1}d{\lambda _2}{\rho ^2}(\frac{{{\lambda _1} + {\lambda _2}}}{2})\left( {{\bf{D}}\hat s(\pi L\rho (\frac{{{\lambda _1} + {\lambda _2}}}{2})({\lambda _1} - {\lambda _2}))\epsilon (\pi L\rho (\frac{{{\lambda _1} + {\lambda _2}}}{2})({\lambda _i} - {\lambda _j}))} \right){e^{i({\lambda _1} - {\lambda _2})t}}{e^{ - \beta ({\lambda _1} + {\lambda _2})}}
\end{align}
Again, changing the variable 
\begin{align}
&{u_1} = {\lambda _1} - {\lambda _2}\nonumber\\
&{u_2} = \frac{{{\lambda _1} + {\lambda _2}}}{2}
\end{align}
We simplify it as
\begin{align}
&- \frac{1}{2}\int d {\lambda _1}d{\lambda _2}\left( {{\rm{Tr}}\left( {K({\lambda _1},{\lambda _2})K({\lambda _2},{\lambda _1})} \right)} \right){e^{i({\lambda _1} - {\lambda _2})t}}{e^{ - \beta ({\lambda _1} + {\lambda _2})}}\nonumber\\
&=- {L^2}\int d {u_1}d{u_2}\frac{{{{\sin }^2}(\pi L\rho ({u_2}){u_1})}}{{{{(\pi L{u_1})}^2}}}{e^{i{u_1}t}}{e^{ - 2\beta {u_2}}}\nonumber\\
&+ {L^2}\int d {u_1}d{u_2}{\rho ^2}({u_2})\left( {{\bf{D}}\hat s(\pi L\rho ({u_2}){u_1}){\bf{I}}\hat s(\pi L\rho ({u_2}){u_1})} \right){e^{i{u_1}t}}{e^{ - 2\beta {u_2}}}\nonumber\\
&- {L^2}\int d {u_1}d{u_2}{\rho ^2}({u_2})\left( {{\bf{D}}\hat s(\pi L\rho ({u_2}){u_1})\epsilon (\pi L\rho ({u_2}){u_1})} \right){e^{i{u_1}t}}{e^{ - 2\beta {u_2}}}
\end{align}
We could firstly perform the integral over $u_1$, and the result is
\begin{align}
&Le^{-2\beta u_2}\rho ({u_2})\left\{ {\begin{array}{*{20}{c}}
{1 - \frac{t}{{\pi \rho ({u_2})L}} + \frac{t}{{2\pi \rho ({u_2})L}}\log \left( {1 + \frac{t}{{\pi \rho ({u_2})L}}} \right)}&{t < 2\pi \rho ({u_2})L}\\
{ - 1 + \frac{t}{{2\pi \rho ({u_2})L}}\log \left( {\frac{{1 + \frac{t}{{\pi \rho ({u_2})L}}}}{{\frac{t}{{\pi \rho ({u_2})L}} - 1}}} \right)}&{t > 2\pi \rho ({u_2})L}
\end{array}} \right.\nonumber\\
&=e^{-2\beta u_2} \max \left( {L\rho ({u_2}) - \frac{t}{\pi } + \frac{t}{{2\pi }}\log \left( {1 + \frac{t}{{\pi \rho ({u_2})L}}} \right), - L\rho ({u_2}) + \frac{t}{{2\pi }}\log \left( {\frac{{1 + \frac{t}{{\pi \rho ({u_2})L}}}}{{\frac{t}{{\pi \rho ({u_2})L}} - 1}}} \right)} \right)
\end{align}
In the GSE case we have the refined kernel,
\begin{align}
&\tilde K({\lambda _i},{\lambda _j}) \equiv L\rho (({\lambda _i} + {\lambda _j})/2) \times \nonumber\\
&\left( {\begin{array}{*{20}{c}}
{\hat s(2\pi L\rho (({\lambda _i} + {\lambda _j})/2)({\lambda _i} - {\lambda _j}))}&{{\bf{D}}\hat s(2\pi L\rho (({\lambda _i} + {\lambda _j})/2)({\lambda _i} - {\lambda _j}))}\\
{{\bf{I}}\hat s(2\pi L\rho (({\lambda _i} + {\lambda _j})/2)({\lambda _i} - {\lambda _j}))}&{\hat s(2\pi L\rho (({\lambda _i} + {\lambda _j})/2)({\lambda _i} - {\lambda _j}))}
\end{array}} \right)
\end{align}
The same technology gives
\begin{align}
&{{\cal R}_2}(t,\beta )\nonumber\\
&=L{r_1}(2i\beta ) + {L^2}{r_1}(t + i\beta ){r_1}(t - i\beta )\nonumber\\
&- \frac{1}{2}\int d {\lambda _1}d{\lambda _2}\left( {{\rm{Tr}}\left( {K({\lambda _1},{\lambda _2})K({\lambda _2},{\lambda _1})} \right)} \right){e^{i({\lambda _1} - {\lambda _2})t}}{e^{ - \beta ({\lambda _1} + {\lambda _2})}}
\end{align}
where
\begin{align}
&- \frac{1}{2}\int d {\lambda _1}d{\lambda _2}\left( {{\rm{Tr}}\left( {K({\lambda _1},{\lambda _2})K({\lambda _2},{\lambda _1})} \right)} \right){e^{i({\lambda _1} - {\lambda _2})t}}{e^{ - \beta ({\lambda _1} + {\lambda _2})}}\nonumber\\
&=-\int {d{u_2}} {e^{ - 2\beta {u_2}}}\max \left( {L\rho ({u_2}) - \frac{t}{{4\pi }} + \frac{t}{{8\pi }}\log \left| {1 - \frac{t}{{2\pi \rho ({u_2})L}}} \right|,0} \right)
\end{align}
Although we didn't find an analytic result, one can compute those integrals numerically.
\section{Wishart-Laguerre spectral form factor}\label{L}
\subsection{Random matrix theory review}
In this part, we will consider the Wishart-Laguerre random matrices from product of square Gaussian ensembles. These Wishart-Laguerre random matrices are square of standard Gaussian random matrices, which we call Wishart-Laguerre Unitary Ensemble (LUE),  Wishart-Laguerre Orthogonal Ensemble (LOE), and Wishart-Laguerre Symplectic Ensemble (LSE), for square of GUE, GOE and GSE distribution (For LSE we also mean $2L\times 2L$ matrix, while LUE and LOE we mean $L\times L$ matrices). The joint eigenvalue distribution is given by
\begin{align}
P(\lambda )\sim {\left| {\Delta (\lambda )} \right|^{\tilde \beta }}\prod\nolimits_{k = 1}^L {{e^{ - \frac{{\tilde \beta L}}{4}{\lambda _k}}}} 
\end{align}
where $\tilde{\beta}=1,2,4$ corresponds to LOE, LUE and LSE ensembles. We are also interested in the $n$ point correlation functions
\begin{align}
\rho^{(n)}(\lambda_1,\ldots,\lambda_n) = \int d\lambda_{n+1}\ldots d\lambda_L P(\lambda_1,\ldots,\lambda_L)
\end{align}
The one point function is the square of the semicircle law in the large $L$ limit, which we could call it as Pastur-Marchenko distribution
\begin{align}
{\rho ^{(1)}}(\lambda ) = \rho (\lambda ) = \frac{1}{{2\pi \lambda }}\sqrt {\lambda (4 - \lambda )} 
\end{align}
where now the value of $\lambda$ is ranging from 0 to 4. 
\\
\\
We will use the kernels in the large $L$ limit to compute correlation functions and the form factors in terms of box approximation. Similarly, for LUE it is a determinant point process, so we could determine the correlation functions as 
\begin{align}
{\rho ^{(n)}}({\lambda _1}, \ldots ,{\lambda _n}) = \frac{{(L - n)!}}{{L!}}\det (K({\lambda _i},{\lambda _j}))_{i,j = 1}^n
\end{align}
where
\begin{align}
K({\lambda _i},{\lambda _j}) \equiv \left\{ {\begin{array}{*{20}{l}}
{\frac{{\sin (L\rho (u)\pi ({\lambda _i} - {\lambda _j}))}}{{\pi ({\lambda _i} - {\lambda _j})}}{\text{         for }}i\not  = j}\\
{\frac{L}{{2\pi {\lambda _i}}}\sqrt {{\lambda _i}(4 - {\lambda _i})} {\text{             for }}i = j}
\end{array}} \right.
\end{align}
with an undetermined constant $u$ from $[0,4]$. The origination of this constant is from the approximation method, finding an average to put a box in the whole interval $[0,4]$. In the GUE case, we naturally choose $u$ to be 0 because the interval is symmetric in the range $[-2,2]$. However, here in the range with a positive definite eigenvalue we cannot use such a prescription.  
\\
\\
Similarly, for the LOE case we have
\begin{align}
K({\lambda _i},{\lambda _j}) \equiv \left\{ {\begin{array}{*{20}{l}}
{L\rho (u)\left( {\begin{array}{*{20}{c}}
{\hat s(L\rho (u)\pi ({\lambda _i} - {\lambda _j}))}&{{\bf{D}}\hat s(L\rho (u)\pi ({\lambda _i} - {\lambda _j}))}\\
{{\bf{I}}\hat s(L\rho (u)\pi ({\lambda _i} - {\lambda _j})) - \epsilon (L\rho (u)\pi ({\lambda _i} - {\lambda _j}))}&{\hat s(L\rho (u)\pi ({\lambda _i} - {\lambda _j}))}
\end{array}} \right)}&{{\rm{for}}\;\;i\not  = j}\\
{}&{}\\
{\frac{L}{{2\pi \lambda }}\sqrt {\lambda (4 - \lambda )} \left( {\begin{array}{*{20}{c}}
1&0\\
0&1
\end{array}} \right)\;}&{{\rm{for}}\;\;i = j}
\end{array}} \right.
\end{align}
and for the LSE we have 
\begin{align}
K({\lambda _i},{\lambda _j}) \equiv \left\{ {\begin{array}{*{20}{l}}
{L\rho (u)\left( {\begin{array}{*{20}{c}}
{\hat s(2\pi L\rho (u)({\lambda _i} - {\lambda _j}))}&{{\bf{D}}\hat s(2\pi L\rho (u)({\lambda _i} - {\lambda _j}))}\\
{{\bf{I}}\hat s(2\pi L\rho (u)({\lambda _i} - {\lambda _j}))}&{\hat s(2\pi L\rho (u)({\lambda _i} - {\lambda _j}))}
\end{array}} \right)}&{{\rm{for}}\;\;i\not  = j}\\
{}&{}\\
{\frac{L}{{2\pi \lambda }}\sqrt {\lambda (4 - \lambda )} \left( {\begin{array}{*{20}{c}}
1&0\\
0&1
\end{array}} \right)\;}&{{\rm{for}}\;\;i = j}
\end{array}} \right.
\end{align}
and the Pfaffian point process will determine the structure of the correlation functions in terms of form factors as 
\begin{align}
{{\rho }^{(n)}}({{\lambda }_{1}},\ldots ,{{\lambda }_{n}})=\frac{(L-n)!}{L!}\sum\limits_{m\ge 1}^{\bigcupdot\nolimits_{i=1}^{m}{{{S}_{i}}}=\left\{ 1,\ldots ,n \right\}}{{{(-1)}^{n-m}}{{T}_{{{S}_{1}}}}\ldots {{T}_{{{S}_{m}}}}}
\end{align}
where
\begin{align}
{{T}_{n}}({{\lambda }_{1}},\ldots ,{{\lambda }_{n}})=\frac{1}{2n}\sum\limits_{\sigma }{\text{Tr}\left( K({{\lambda }_{{{\sigma }_{1}}}},{{\lambda }_{{{\sigma }_{2}}}})K({{\lambda }_{{{\sigma }_{2}}}},{{\lambda }_{{{\sigma }_{3}}}})\ldots K({{\lambda }_{{{\sigma }_{n}}}},{{\lambda }_{{{\sigma }_{1}}}}) \right)}
\end{align}
About convolution theorems, in this case we could obtain
\begin{thm}[Convolution formula for LUE]
We have 
\begin{align}
&\int {\prod\limits_{i = 1}^m {d{\lambda _i}} } K({\lambda _1},{\lambda _2})K({\lambda _2},{\lambda _3}) \ldots K({\lambda _{m - 1}},{\lambda _m})K({\lambda _m},{\lambda _1}){e^{\sum\nolimits_{i = 1}^m {i{k_i}{\lambda _i}} }}\nonumber\\
&= L{r_3}(\sum\nolimits_{i = 1}^m {{k_i}} )\int {dk} g(k)g(k + \frac{{{k_1}}}{{2\pi {\alpha _L}}})g(k + \frac{{{k_2}}}{{2\pi {\alpha _L}}}) \ldots g(k + \frac{{{k_{m - 1}}}}{{2\pi {\alpha _L}}})
\end{align}
where
\begin{align}
{\alpha _L} = L\rho (u)
\end{align}
and
\begin{align}
{r_3}(t) = \frac{{\sin (t/2\rho (u))}}{{t/2\rho (u)}}
\end{align}
\end{thm}
\begin{thm}[Convolution formula for LOE]
We have 
\begin{align}
&\int{\prod\limits_{i=1}^{m}{d{{\lambda }_{i}}}}{K}({{\lambda }_{1}},{{\lambda }_{2}}){K}({{\lambda }_{2}},{{\lambda }_{3}})\ldots {K}({{\lambda }_{m-1}},{{\lambda }_{m}}){K}({{\lambda }_{m}},{{\lambda }_{1}}) \,e^{i \sum_{i=1}^{m} k_i \lambda_i} \nonumber\\
& = L{r_3}(\sum\nolimits_{i = 1}^m {{k_i}} )\int {dk} \,G(k)G(k + \frac{{{k_1}}}{{2\pi {\alpha _L}}})G(k + \frac{{{k_2}}}{{2\pi {\alpha _L}}}) \ldots G(k + \frac{{{k_{m - 1}}}}{{2\pi {\alpha _L}}})
\end{align}
\end{thm}
\begin{thm}[Convolution formula for LSE]
We have 
\begin{align}
&\int{\prod\limits_{i=1}^{m}{d{{\lambda }_{i}}}}{K}({{\lambda }_{1}},{{\lambda }_{2}}){K}({{\lambda }_{2}},{{\lambda }_{3}})\ldots {K}({{\lambda }_{m-1}},{{\lambda }_{m}}){K}({{\lambda }_{m}},{{\lambda }_{1}}) \,e^{i \sum_{i=1}^{m} k_i \lambda_i} \nonumber\\
& = L{r_3}(\sum\nolimits_{i = 1}^m {{k_i}} )\int {dk} \,H(k)H(k + \frac{{{k_1}}}{{2\pi {\alpha _L}}})H(k + \frac{{{k_2}}}{{2\pi {\alpha _L}}}) \ldots H(k + \frac{{{k_{m - 1}}}}{{2\pi {\alpha _L}}})
\end{align}
\end{thm}
Notice that in order to mimic the delta function, we have to integrate over $\mathbb{R}$ for all variables instead of bounded range (Similar with the treatment in the Gaussian random matrices). Based on these knowledge, we could start to summarize the results for form factors in the case of Wishart-Laguerre matrices.
\subsection{Result summary}
\subsubsection{Two point form factor}
The two point form factor has the universal form
\begin{align}
{\mathcal{R}_2} = L + {L^2}{r_1}(t)r_1^*(t) - L{r_2}(t)
\end{align}
where we always have 
\begin{align}
{{r}_{1}}(t)={{e}^{2it}}({{J}_{0}}(2t)-i{{J}_{1}}(2t)) 
\end{align}
and for LUE we have
\begin{align}
 {{r}_{2}}(t)=\left\{ \begin{array}{*{35}{l}}
   1-\frac{t}{2\pi L\rho (u)} & \text{for}~~0<t<2\pi L\rho (u)  \\
   {} & {}  \\
   0 & \text{for}~~t>2\pi L\rho (u)  \\
\end{array} \right. 
\end{align}
for LOE we have
\begin{align}
{r _2}(t) = \left\{ {\begin{array}{*{20}{c}}
{1 - \frac{t}{{\pi \rho (u)L}} + \frac{t}{{2\pi \rho (u)L}}\log \left( {1 + \frac{t}{{\pi \rho (u)L}}} \right)}&{t < 2L\pi \rho (u)}\\
{ - 1 + \frac{t}{{2\pi \rho (u)L}}\log \left( {\frac{{1 + t/\pi \rho (u)L}}{{t/L\pi \rho (u) - 1}}} \right)}&{t > 2L\pi \rho (u)}
\end{array}} \right.
\end{align}
for LSE we have
\begin{align}
{r_2}(t) = \left\{ {\begin{array}{*{20}{c}}
{1 - \frac{t}{{4\pi \rho (u)L}} + \frac{t}{{8\pi \rho (u)L}}\log \left| {1 - \frac{t}{{2\pi \rho (u)L}}} \right|}&{t < 4L\pi \rho (u)}\\
0&{t > 4L\pi \rho (u)}
\end{array}} \right.
\end{align}
\subsubsection{Four point form factor}
The four point form factor has the universal form 
\begin{align}
&{{\cal R}_4} =  {L^4}|{r_1}(t){|^4}\nonumber\\
&- 2{L^3}{\mathop{\rm Re}\nolimits} (r_1^2(t)){r_2}(t){r_3}(2t) - 4{L^3}|{r_1}(t){|^2}{r_2}(t) + 2{L^3}{\rm{Re}}({r_1}(2t)r_1^{*2}(t)) + 4{L^3}|{r_1}(t){|^2}\nonumber\\
&+ 2{L^2}r_2^2(t) + {L^2}r_2^2(t)r_3^2(2t) + 6{L^2}{\rm{Re}}({r_1}(t)){r _{3,1}}(t){r_3}(t) + 2{L^2}{\rm{Re}}({r_1}(t)){r _{3,2}}(t){r_3}(t)\nonumber\\
&- 2{L^2}{\rm{Re}}({r_1}(2t){r_3}(2t){r _2}(t) - 4{L^2}{\rm{Re}}(r_1^*(t)){r_3}(t){r _2}(2t) + {L^2}|{r_1}(2t){|^2} - 4{L^2}|{r_1}(t){|^2}\nonumber\\
&- 4{L^2}{r_2}(t) + 2{L^2}\nonumber\\
&- 6L{r_4}(t) - L{r_2}(2t) + 4L{r_{3,3}}(t) + 4L{r_2}(t) - L
\end{align}
where the dominated term is 
\begin{align}
&{{\cal R}_4}\sim{L^4}|{r_1}(t){|^4}\nonumber\\
&+ 2{L^2}r_2^2(t) - 4{L^2}{r_2}(t) + 2{L^2}\nonumber\\
&- 6L{r_4}(t) - L{r_2}(2t) + 4L{r_{3,3}}(t) + 4L{r_2}(t) - L
\end{align}
Now we specifying these block functions for different ensembles. For all three ensembles we still have 
\begin{align}
&{{r}_{1}}(t)={{e}^{2it}}({{J}_{0}}(2t)-i{{J}_{1}}(2t))\nonumber\\ 
&{r_3}(t) =\frac{{\sin (t/2\rho (u))}}{{t/2\rho (u)}}
\end{align}
For LUE we have
\begin{align}
&{r_{3,1}}(t) = {r_{3,2}}(t) =r_{3,3}(t/3)= {r_4}(t/2) = {r_2}(t) =\left\{ \begin{array}{*{35}{l}}
   1-\frac{t}{2\pi L\rho (u)} & \text{for}~~0<t<2\pi L\rho (u)  \\
   {} & {}  \\
   0 & \text{for}~~t>2\pi L\rho (u)  \\
\end{array} \right. 
\end{align}
For LOE we have 
\begin{align}
&{r_{3,2}}(t) = {r_2}(t) = \left\{ {\begin{array}{*{20}{c}}
{1 - \frac{t}{{\pi \rho (u)L}} + \frac{t}{{2\pi \rho (u)L}}\log \left( {1 + \frac{t}{{\pi \rho (u)L}}} \right)}&{t < 2L\pi \rho (u)}\\
{ - 1 + \frac{t}{{2\pi \rho (u)L}}\log \left( {\frac{{t/\pi \rho (u)L + 1}}{{t/\pi \rho (u) - 1}}} \right)}&{t > 2L\pi \rho (u)}
\end{array}} \right.\nonumber\\
&{r_{3,1}}(t) = \left\{ {\begin{array}{*{20}{c}}
{1 - \frac{{2t}}{{\pi \rho (u)L}} + \frac{{3t}}{{2\pi \rho (u)L}}\log (1 + \frac{t}{{\pi \rho (u)L}})}&{0 < t < \pi \rho (u)L}\\
{ - 2 + \frac{t}{{\pi \rho (u)L}} + \frac{{3t}}{{2\pi \rho (u)L}}\log \left( {\frac{1}{2}\left( {1 + \frac{3}{{2t/\pi \rho (u) - 1}}} \right)} \right)}&{\pi \rho (u)L < t < 2\pi \rho (u)L}\\
0&{t > 2\pi \rho (u)L}
\end{array}} \right.\nonumber\\
&{r_{3,3}}(t) = \left\{ {\begin{array}{*{20}{c}}
{1 - \frac{{3t}}{{\pi \rho (u)L}} + \frac{t}{{\pi \rho (u)L}}\log (1 + \frac{t}{{\pi \rho (u)L}}) + \frac{{5t}}{{4\pi \rho (u)L}}\log (1 + \frac{{2t}}{{\pi \rho (u)L}})}&{0 < t < \frac{2}{3}\pi \rho (u)L}\\
{ - 2 + \frac{{3t}}{{2\pi \rho (u)L}} + \frac{t}{{\pi \rho (u)L}}\log \left( {\frac{{1 + t/\pi \rho (u)L}}{{3t/\pi \rho (u)L - 1}}} \right) + \frac{{5t}}{{4\pi \rho (u)L}}\log \left( {\frac{{1 + 2t/\pi \rho (u)L}}{{3t/\pi \rho (u)L - 1}}} \right)}&{\frac{2}{3}\pi \rho (u)L < t < \pi \rho (u)L}\\
{ - 1 + \frac{t}{{2\pi \rho (u)L}} + \frac{{5t}}{{4\pi \rho (u)L}}\log \left( {\frac{{1 + 2t/\pi \rho (u)L}}{{3t/\pi \rho (u)L - 1}}} \right)}&{\pi \rho (u)L < t < 2\pi \rho (u)L}\\
0&{t > 2\pi \rho (u)L}
\end{array}} \right.\nonumber\\
&{r_4}(t) = \left\{ {\begin{array}{*{20}{c}}
{1 - \frac{{7t}}{{4\pi \rho (u)L}} + \frac{{5t}}{{4\pi \rho (u)L}}\log \left( {1 + \frac{t}{{\pi \rho (u)L}}} \right)}&{0 < t < \pi \rho (u)L}\\
{ - \frac{3}{2} + \frac{{3t}}{{4\pi \rho (u)L}} + \frac{{5t}}{{4\pi \rho (u)L}}\log \left( {\frac{1}{2}\left( {1 + \frac{3}{{2t/\pi \rho (u)L - 1}}} \right)} \right)}&{L < t < 2\pi \rho (u)L}\\
0&{t > 2\pi \rho (u)L}
\end{array}} \right.
\end{align}
For LSE we have
\begin{align}
&{r_{3,1}} = {r_4}(t) = \left\{ {\begin{array}{*{20}{c}}
{1 - \frac{1}{2}\frac{t}{{\pi \rho (u)L}} + \frac{3}{{16}}\frac{t}{{\pi \rho (u)L}}\log \left| {\frac{t}{{\pi \rho (u)L}} - 1} \right|}&{t < 2\pi \rho (u)L}\\
0&{t > 2\pi \rho (u)L}
\end{array}} \right.\nonumber\\
&{r_{3,2}}(t) = {r_2}(t)\nonumber\\
&{r_{3,3}} = \left\{ {\begin{array}{*{20}{c}}
{1 - \frac{{3t}}{{4\pi \rho (u)L}} + \frac{t}{{32\pi \rho (u)L}}\log \left| {\frac{{2t/\pi \rho (u)L - 2}}{{2 - t/\pi \rho (u)L}}} \right| + \frac{{9t}}{{32\pi \rho (u)L}}\log |\frac{3}{2}\frac{t}{{\pi \rho (u)L}} - 1|}&{t < \frac{4}{3}\pi \rho (u)L}\\
0&{t > \frac{4}{3}\pi \rho (u)L}
\end{array}} \right.
\end{align}
\subsubsection{Refined two point form factor}
We will discuss the improvement of two point form factor with finite temperature, by interval method in this section. For LUE, we have
\begin{align}
&{{\cal R}_2}(t,\beta ) = L{r_1}(2i\beta ) + {L^2}{r_1}(t + i\beta ){r_1}(t - i\beta ) - \int {d{u_2}{e^{ - 2\beta {u_2}}}\max \left( {L\rho ({u_2}) - \frac{t}{{2\pi }},0} \right)} 
\end{align}
When $\beta=0$ the integral is 
\begin{align}
\int {d{u_2}{e^{ - 2\beta {u_2}}}\max \left( {L\rho ({u_2}) - \frac{t}{{2\pi }},0} \right)}  = \frac{{L\left( {2\arctan \left( {\frac{{L - t}}{{2t}}} \right) + \pi } \right)}}{{2\pi }}
\end{align}
So we get
\begin{align}
{{\cal R}_2}(t) ={L^2}{\left| {{r_1}(t)} \right|^2} + \frac{{L\left( {\pi  - 2\arctan \left( {\frac{L}{{2t}} - \frac{t}{{2L}}} \right)} \right)}}{{2\pi }}
\end{align}
The early time expansion of the connected piece gives
\begin{align}
R_2^{\text{conn}}(t) \approx \frac{{2t}}{\pi } + \mathcal{O}\left( {{t^3}} \right)
\end{align}
which means that in the box approximation, we could approximately set 
\begin{align}
\rho (u) = \frac{1}{4} \Rightarrow u = \frac{{16}}{{4 + {\pi ^2}}}
\end{align}
then it could approximately capture the form factor dynamics.
\\
\\
For LOE we have 
\begin{align}
&{\mathcal{R}_2}(t,\beta ) = L{r_1}(2i\beta ) + {L^2}{r_1}(t + i\beta ){r_1}(t - i\beta )\nonumber\\
&- \int {d{u_2}} {e^{ - 2\beta {u_2}}}\max \left( {L\rho ({u_2}) - \frac{t}{\pi } + \frac{t}{{2\pi }}\log \left( {1 + \frac{t}{{\pi \rho ({u_2})L}}} \right), - L\rho ({u_2}) + \frac{t}{{2\pi }}\log \left( {\frac{{1 + \frac{t}{{\pi \rho ({u_2})L}}}}{{\frac{t}{{\pi \rho ({u_2})L}} - 1}}} \right)} \right)
\end{align}
This time in $\beta=0$ case, the expansion gives
\begin{align}
\mathcal{R}_2^{\text{conn}}(t) \approx \frac{{4t}}{\pi } + \mathcal{O}\left( {{t^3}} \right)
\end{align}
But we still have 
\begin{align}
\rho (u) = \frac{1}{4} \Rightarrow u = \frac{{16}}{{4 + {\pi ^2}}}
\end{align}
For LSE case, we have
\begin{align}
&{{\cal R}_2}(t,\beta ) = L{r_1}(2i\beta ) + {L^2}{r_1}(t + i\beta ){r_1}(t - i\beta )\nonumber\\
&- \int {d{u_2}} {e^{ - 2\beta {u_2}}}\max \left( {L\rho ({u_2}) - \frac{t}{{4\pi }} + \frac{t}{{8\pi }}\log \left| {1 - \frac{t}{{2\pi \rho ({u_2})L}}} \right|,0} \right)
\end{align}
where in $\beta=0$ case, the expansion gives
\begin{align}
\mathcal{R}_2^{\text{conn}}(t) \approx \frac{{t}}{\pi } + \mathcal{O}\left( {{t^3}} \right)
\end{align}
And the solution of $u$ is still the same
\begin{align}
\rho (u) = \frac{1}{4} \Rightarrow u = \frac{{16}}{{4 + {\pi ^2}}}
\end{align}
\section{Figures}\label{fg}
We obtain numerous analytic results in the previous sections. In this section, we will try to plot some of those results. 
\\
\\
Figures in \ref{fig1} are describing the two point spectral form factors in Gaussian ensembles. One could observe that the main difference among three ensembles is the behavior around the plateau time. We have a smooth corner for GOE, a kink for GUE and a sharp peak for GSE. These features are universal also for because of different sine kernels.
\begin{figure}[htbp]
  \centering
  \includegraphics[width=0.8\textwidth]{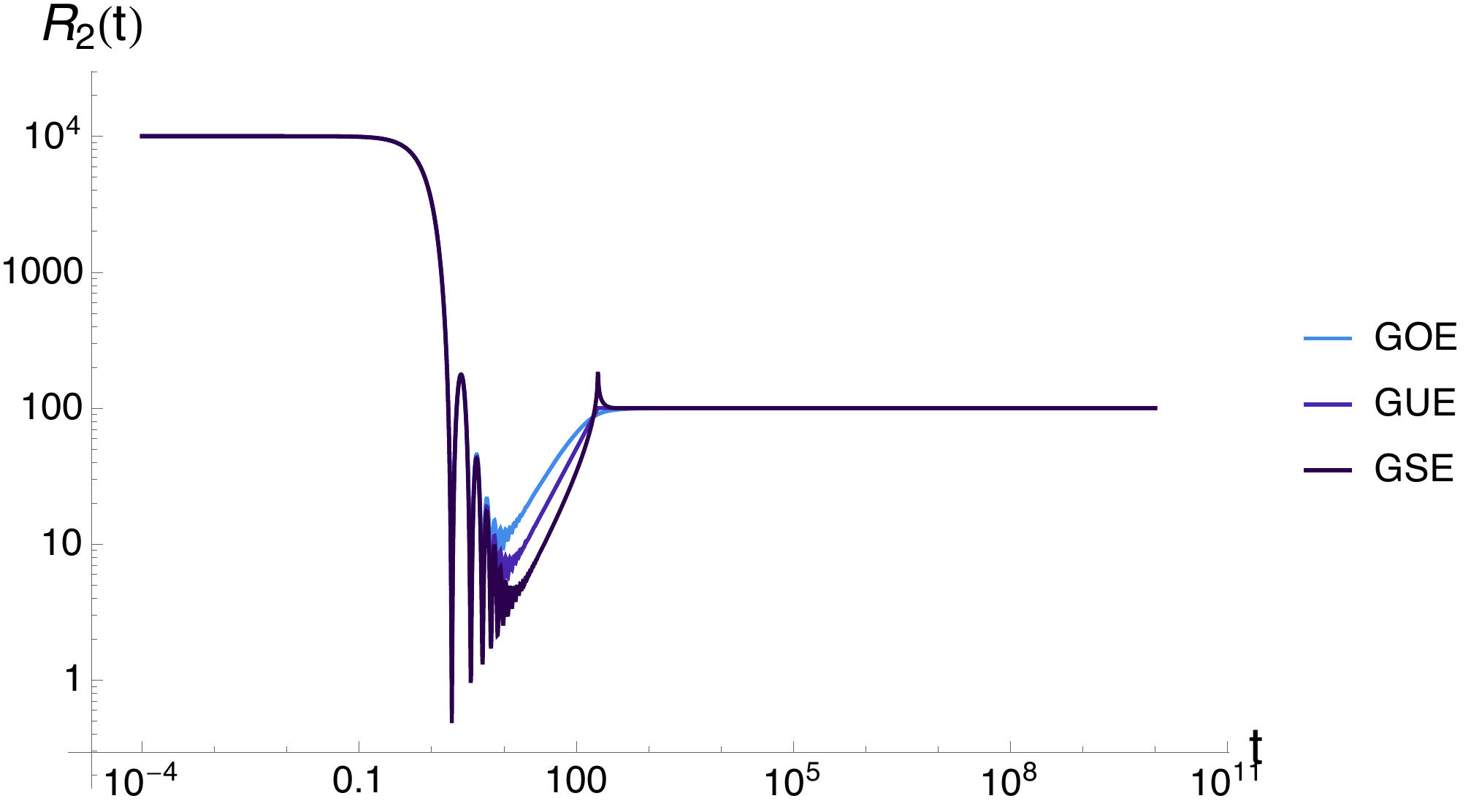}
  \includegraphics[width=0.8\textwidth]{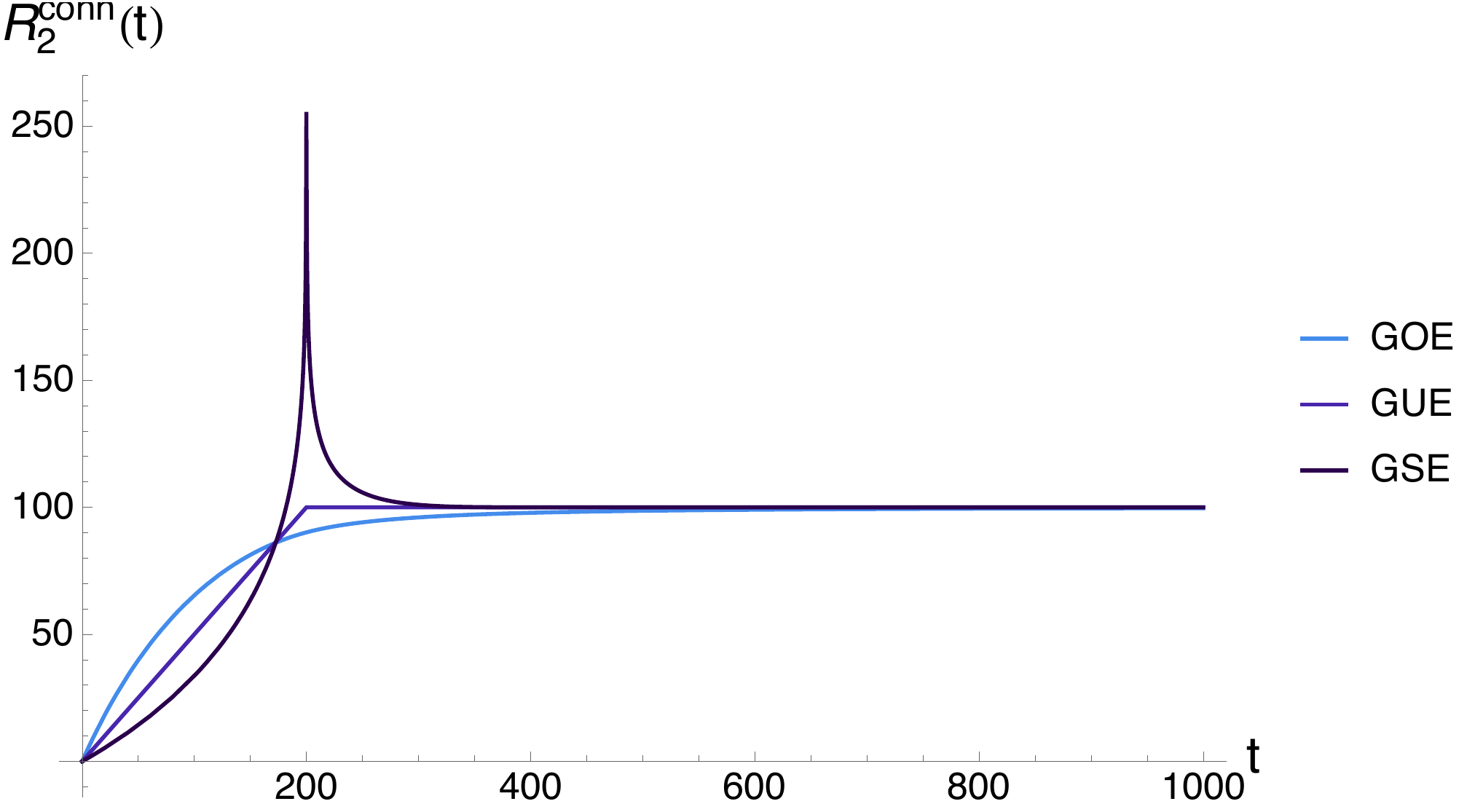}
  \caption{\label{fig1} GOE, GUE, GSE two point form factors $\mathcal{R}_2(t)$ with box cutoff and infinite temperature. We choose $L=100$. Up: full form factor; Down: connected form factor.}
\end{figure}
\\
\\
Figure \ref{fig2} shows a similar behavior four four point form factor $\mathcal{R}_4$.
\begin{figure}[htbp]
\centering
\includegraphics[width=12cm]{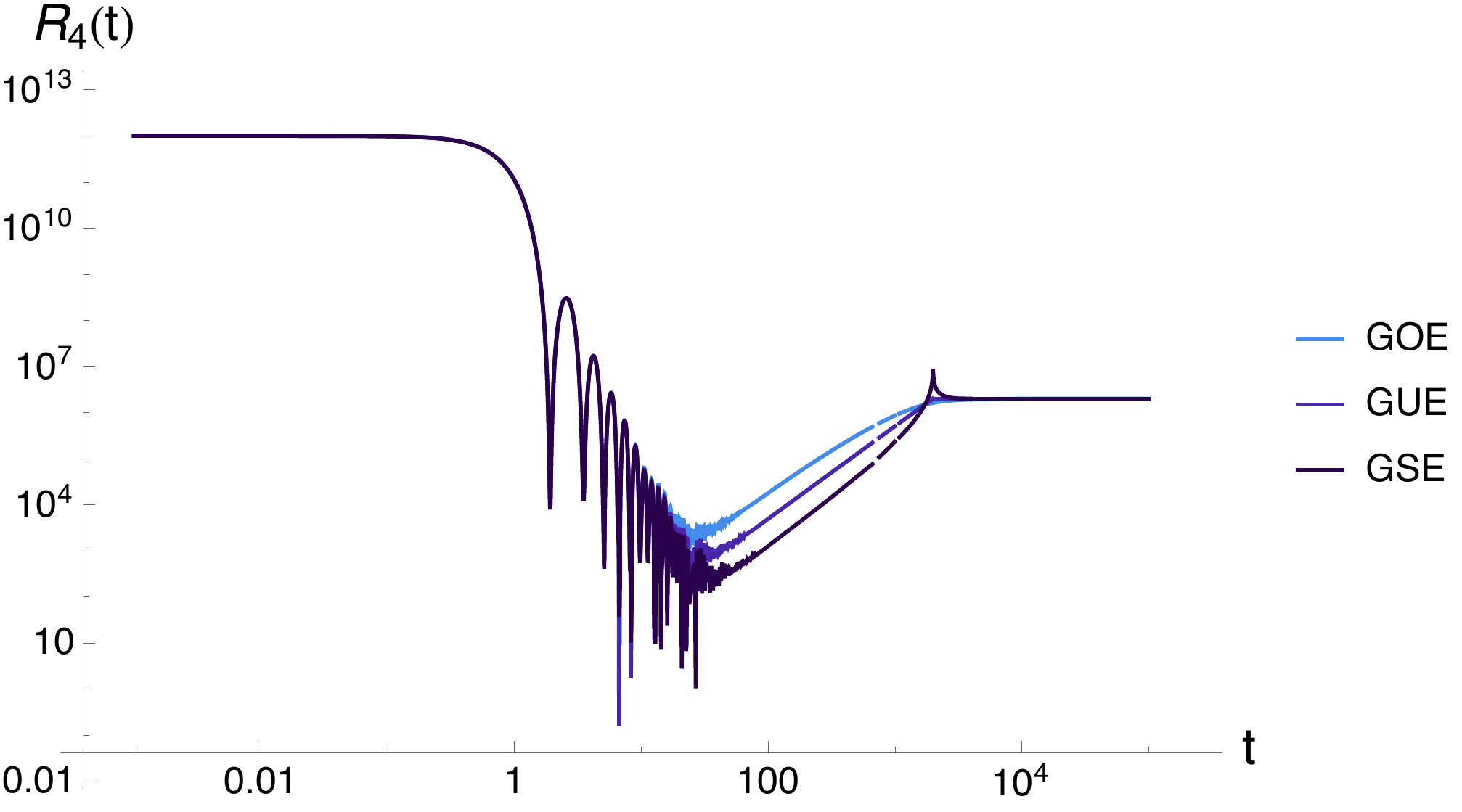}
\caption{GOE, GUE, GSE four point form factors $\mathcal{R}_4(t)$ with box cutoff and infinite temperature. We choose $L=1000$.}\label{fig2}
\end{figure}
\\
\\
We plot similar figures for Wishart-Laguerre ensembles in Figure \ref{fig3} and \ref{fig4}. The main difference is the decaying rate in the relatively early time from $r_1$. Expanding $r_1(t)$ we get $r^{-3/2}$ for Gaussian ensembles and $r^{-1/2}$ for Wishart-Laguerre ensembles. A direct comparison is displayed in Figure \ref{fig5}.
\begin{figure}[htbp]
  \centering
  \includegraphics[width=0.8\textwidth]{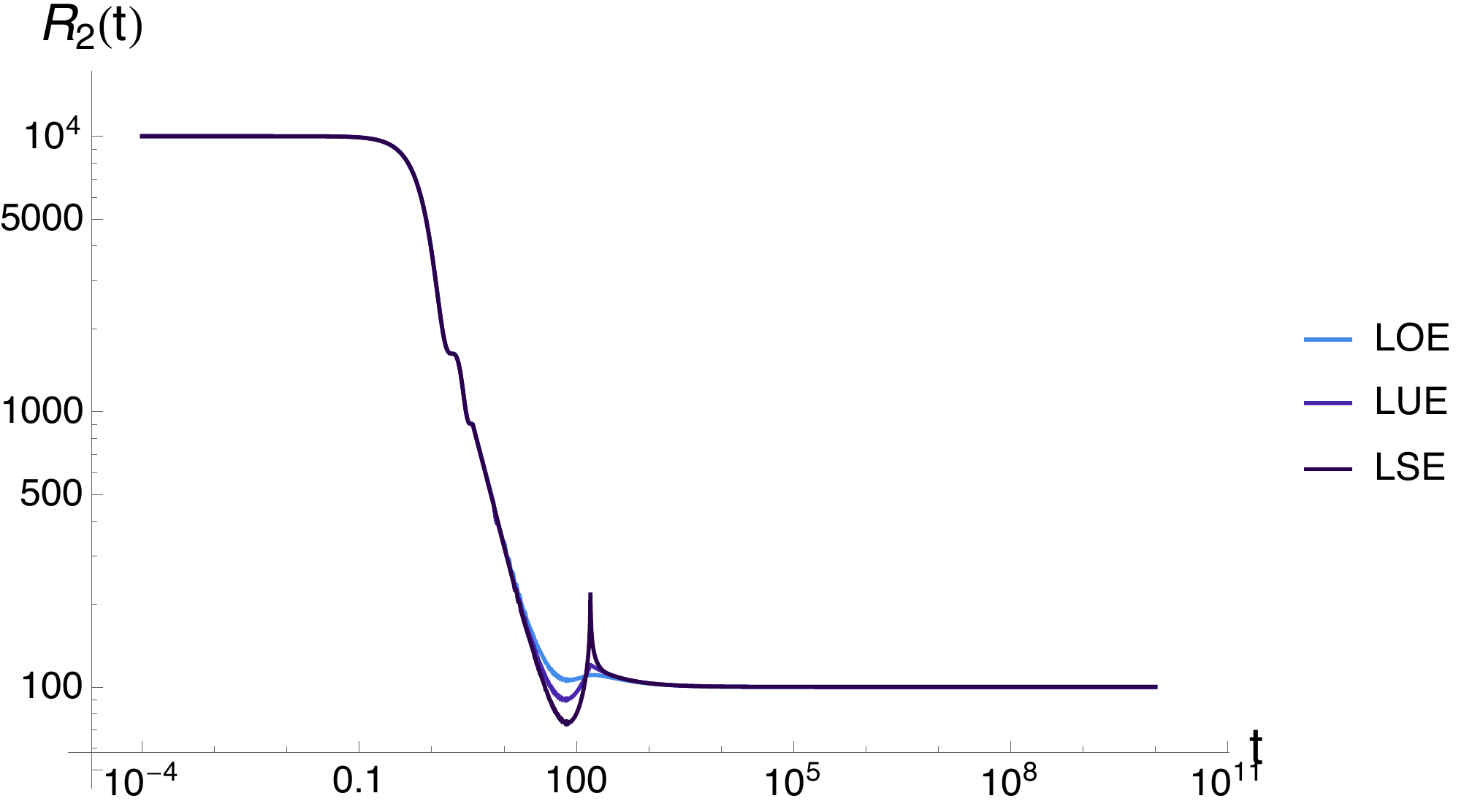}
  \includegraphics[width=0.8\textwidth]{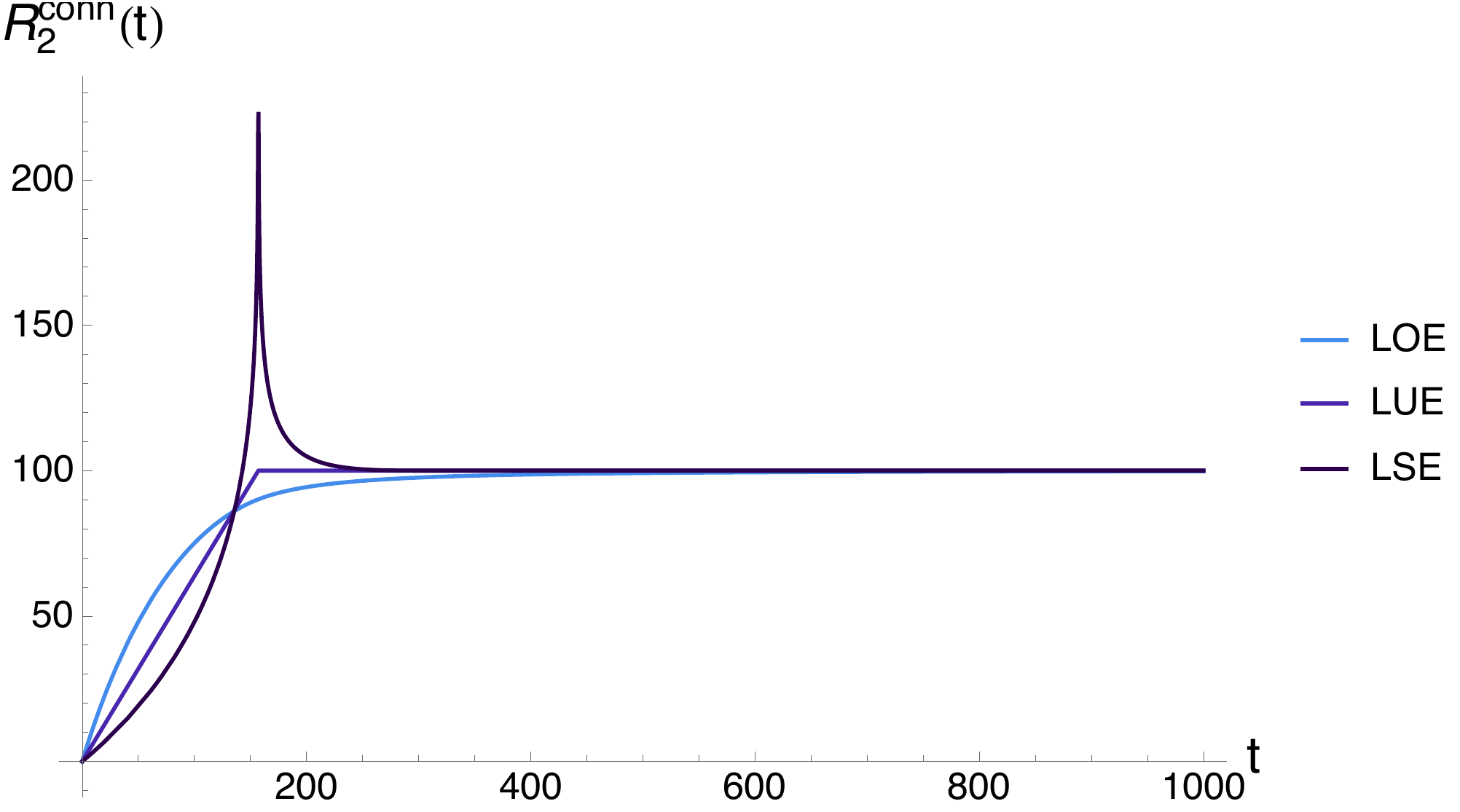}
  \caption{\label{fig3} LOE, LUE, LSE two point form factors $\mathcal{R}_2(t)$ with box cutoff and infinite temperature. We choose $L=100$. Up: full form factor; Down: connected form factor.}
\end{figure}
\begin{figure}[htbp]
\centering
\includegraphics[width=12cm]{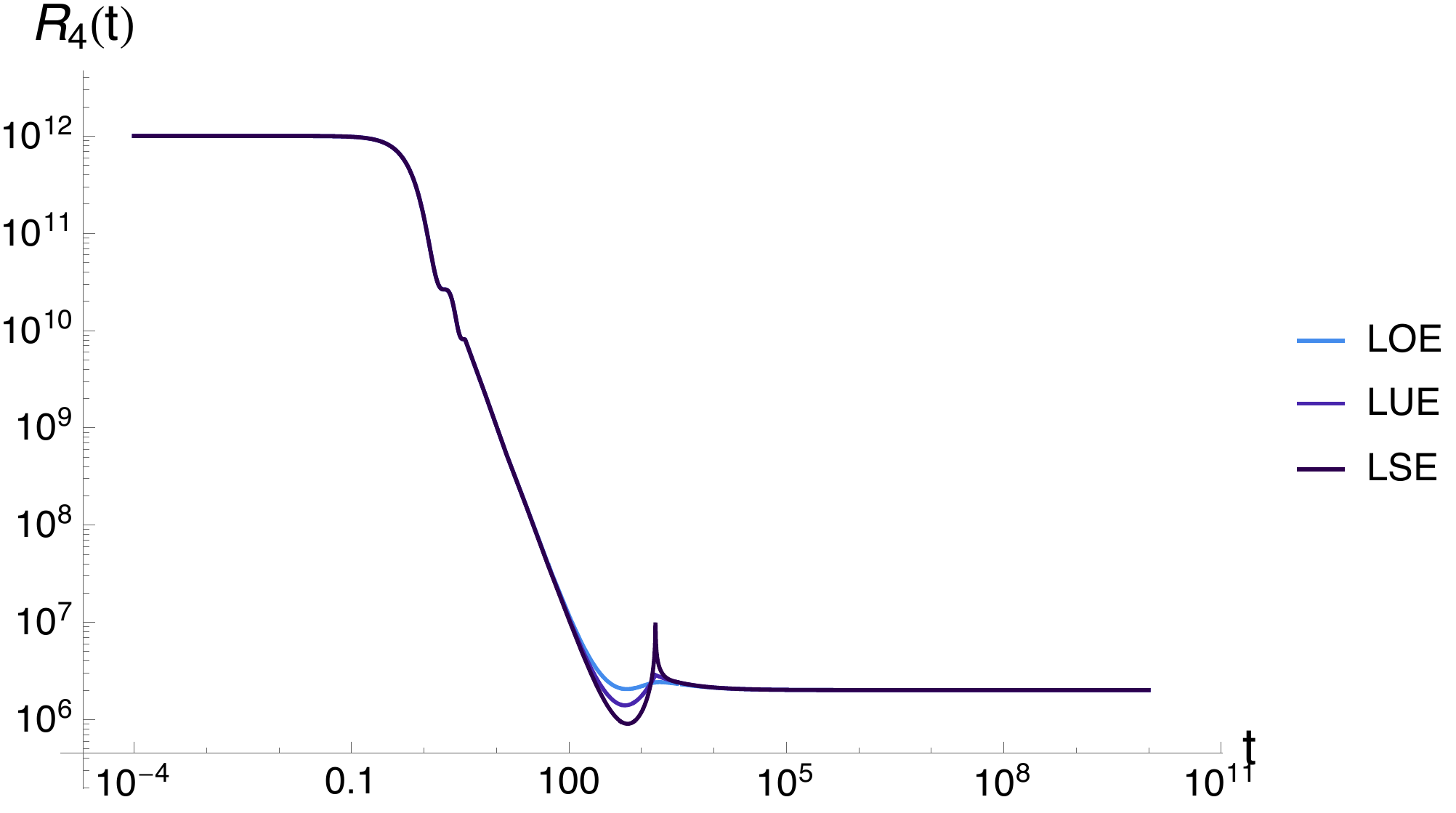}
\caption{LOE, LUE, LSE four point form factors $\mathcal{R}_4(t)$ with box cutoff and infinite temperature. We choose $L=1000$.}\label{fig4}
\end{figure}
\begin{figure}[htbp]
  \centering
  \includegraphics[width=0.8\textwidth]{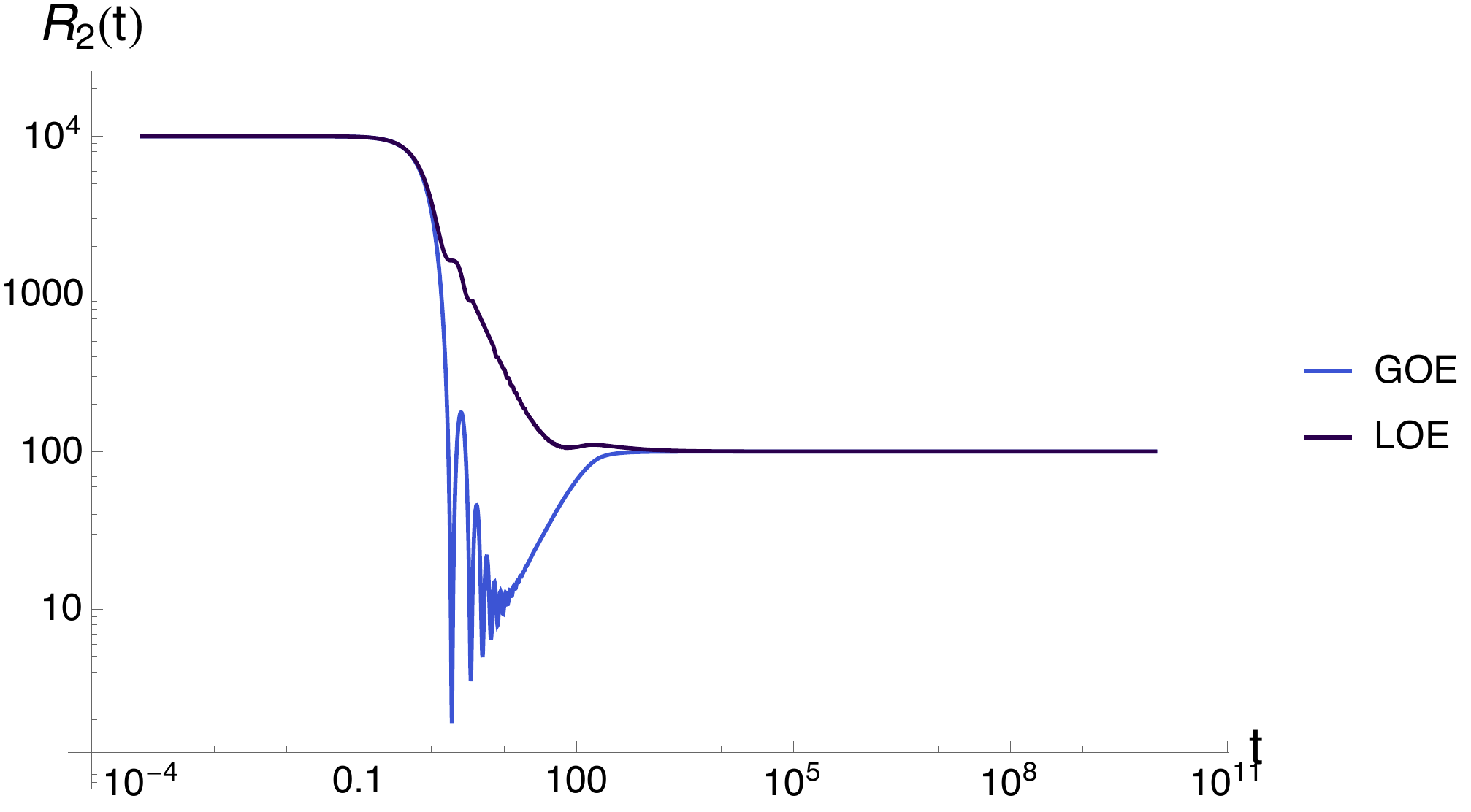}
  \includegraphics[width=0.8\textwidth]{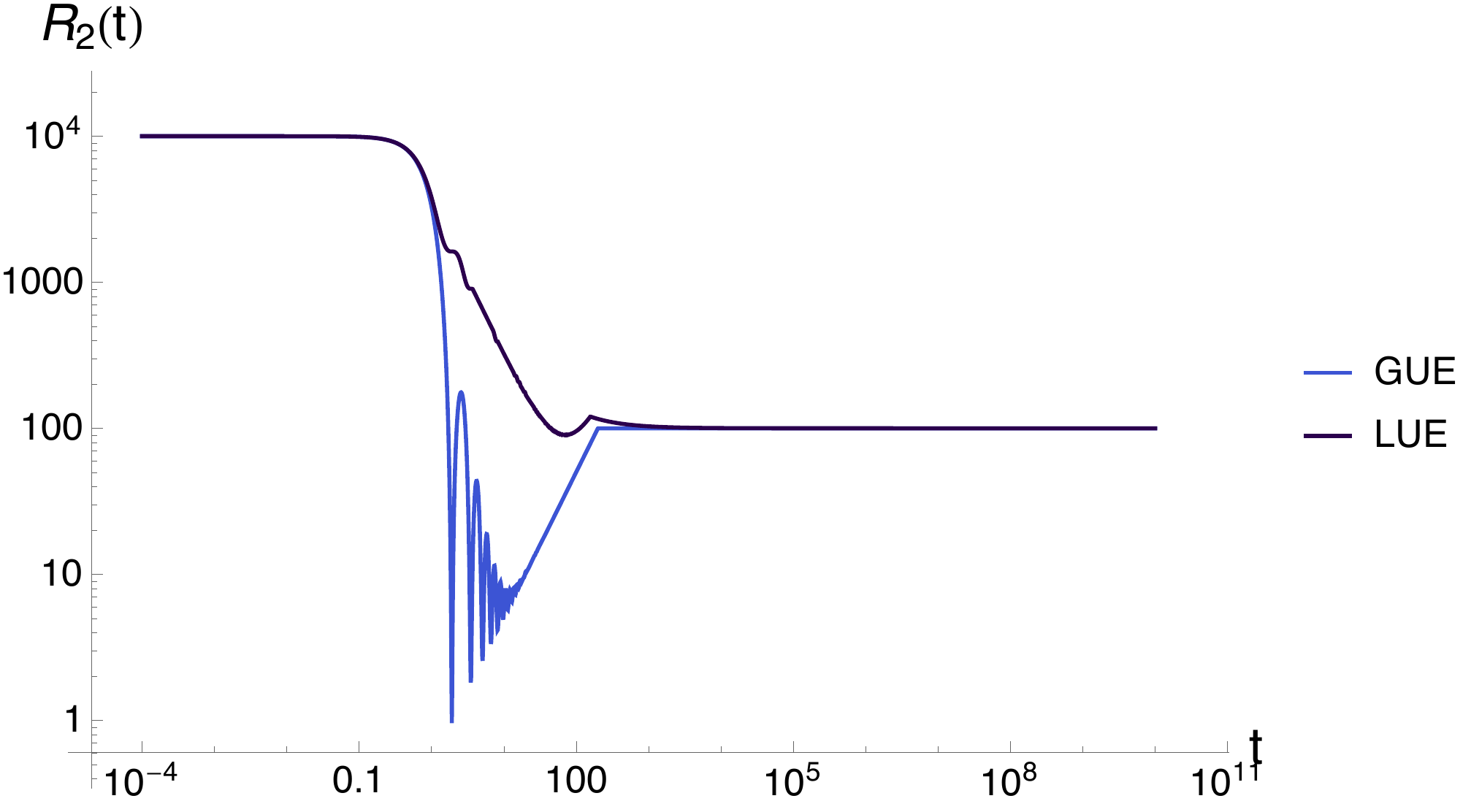}
  \includegraphics[width=0.8\textwidth]{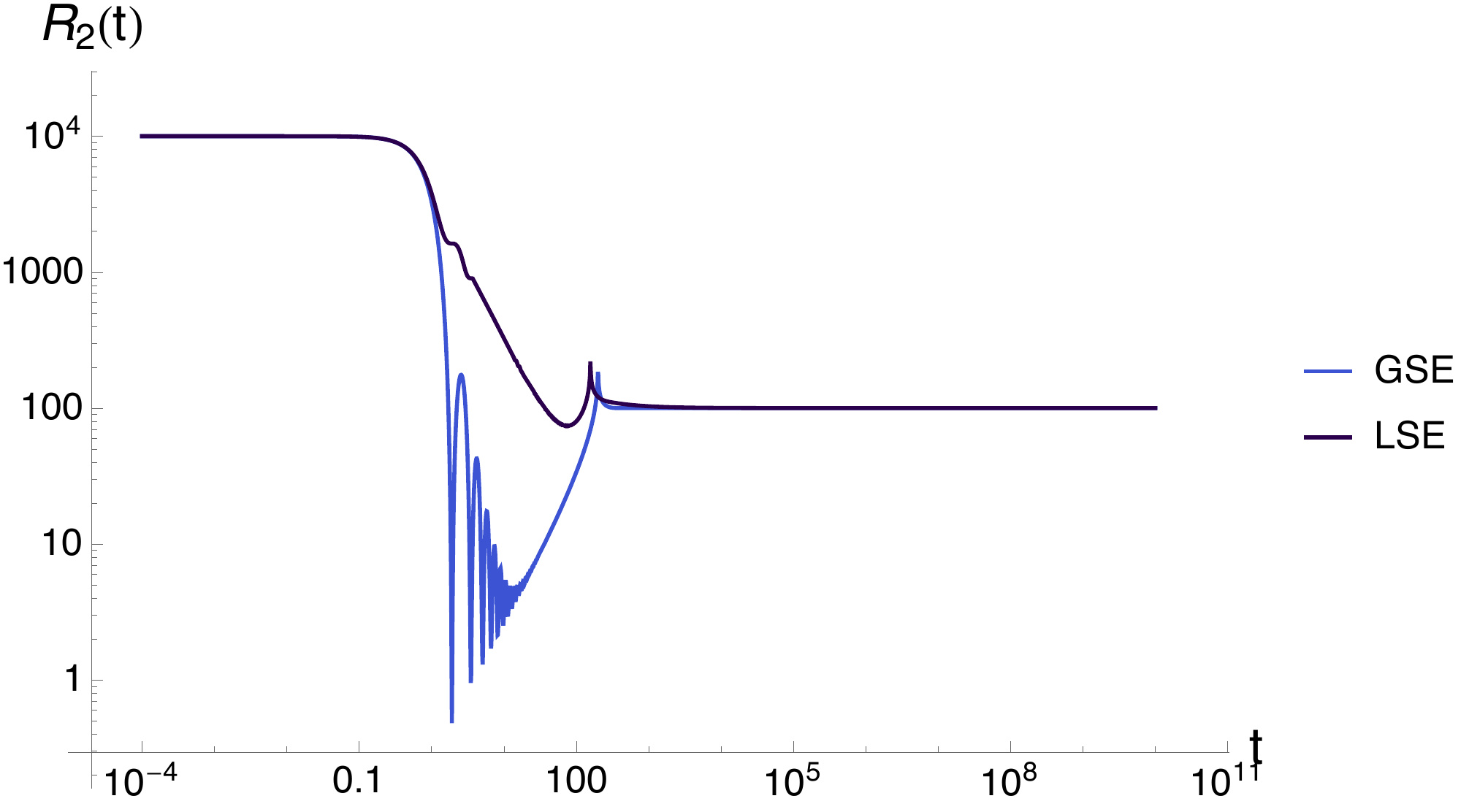}
  \caption{\label{fig5} A direct comparison between Gaussian ensembles and Wishart-Laguerre ensembles in terms of two point form factor $\mathcal{R}_2(t)$ with box cutoff and infinite temperature. We choose $L=100$. Up GOE/LUE; Middle: GUE/LUE; Down: GSE/LSE.}
\end{figure}
\\
\\
There will be an interesting comparison about showing the improvement from the box approximation to the refined form factors. Thus we give the Figure \ref{fig9} for the connected piece of GUE. The box approximation gives a linear result from $(0,0)$ to $(2L,L)$. The plateau value $L$ and the plateau time $2L$ are both correct. However, the correct slope, should be $2/\pi$ instead of $1/2$. Thus, one may consider the Taylor expansion (a naive approximation only chooses the slope, namely the derivative, at relatively early time) to capture the correct slope. Maintaining the correct slope and the plateau value, the plateau time is inaccurate. Thus, a refinement will be consider a nonlinear improvement, which is given by our previous small interval integrals over the short distance kernel. The situation is precisely described in \ref{fig9}.
\begin{figure}[htbp]
\centering
\includegraphics[width=12cm]{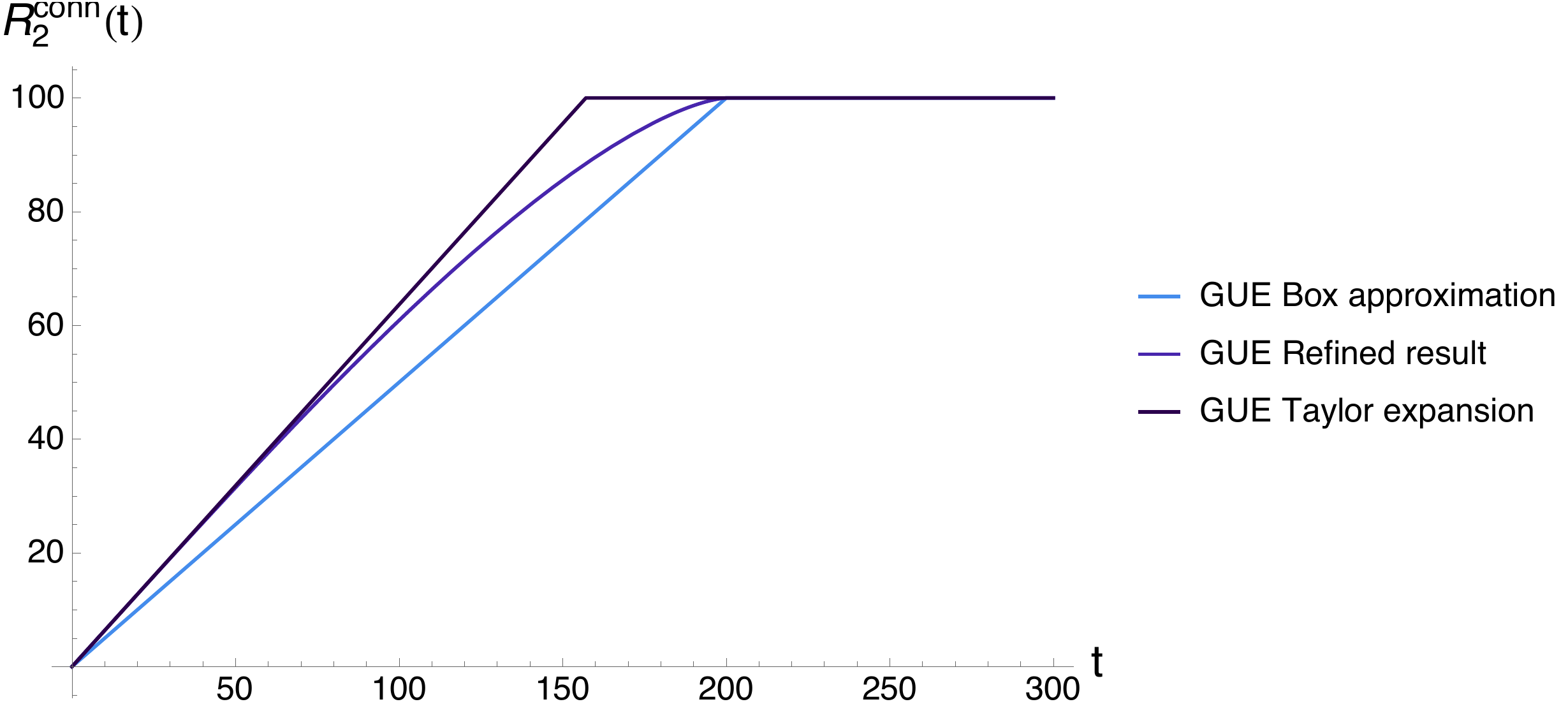}
\caption{GUE connected form factor $\mathcal{R}_2^\text{conn}(t)$ with different approximations in the infinite temperature. We choose $L=100$.}\label{fig9}
\end{figure}
\\
\\
One can generalize this analysis to other Gaussian ensembles and also Wishart-Laguerre ensembles, which are described in Figure \ref{fig8} and \ref{fig11}. One can notice that there is an interesting observation, where the kinky behavior near the plateau time for GSE and LSE ensembles is suppressed, which causes a deviation between the box approximation and the small interval approximation. 
\begin{figure}[htbp]
  \centering
  \includegraphics[width=0.8\textwidth]{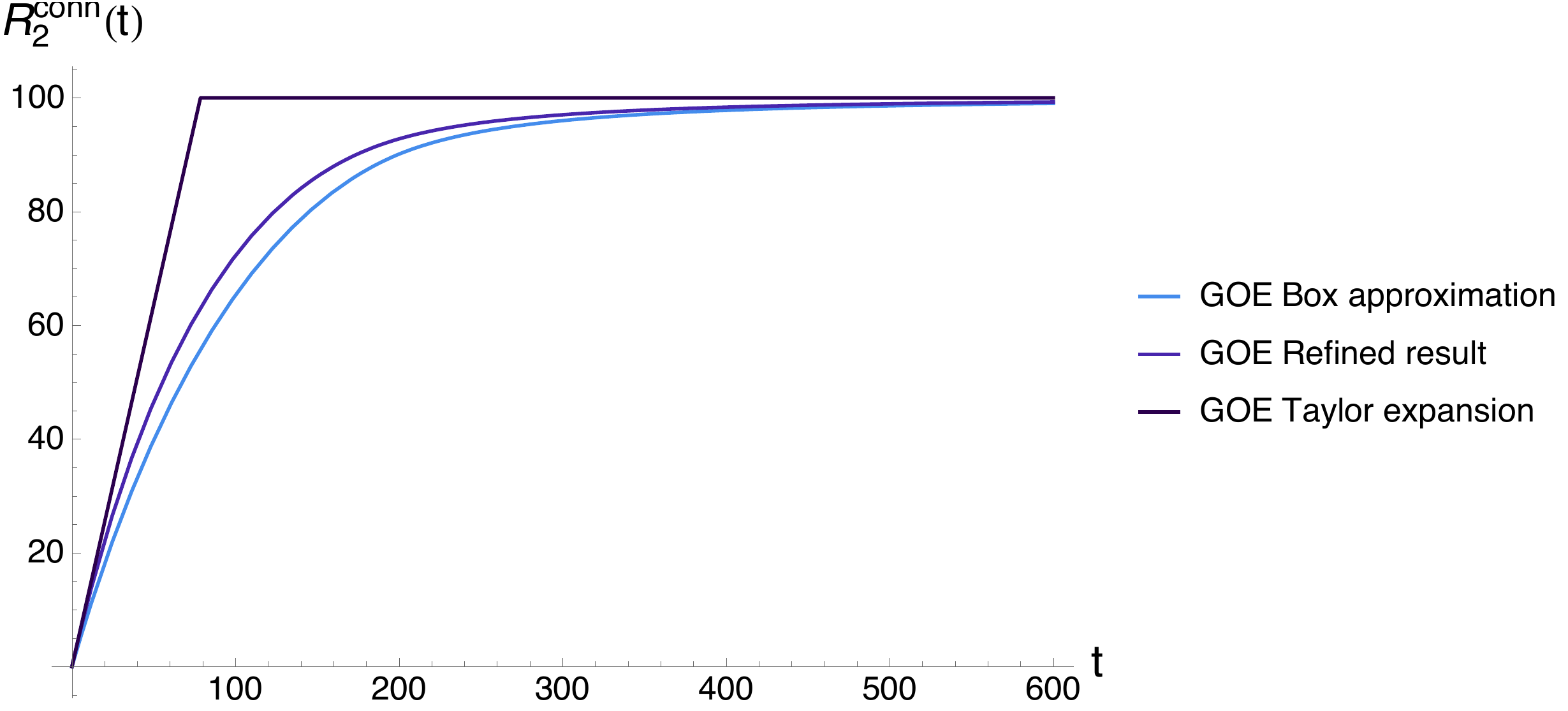}
  \includegraphics[width=0.8\textwidth]{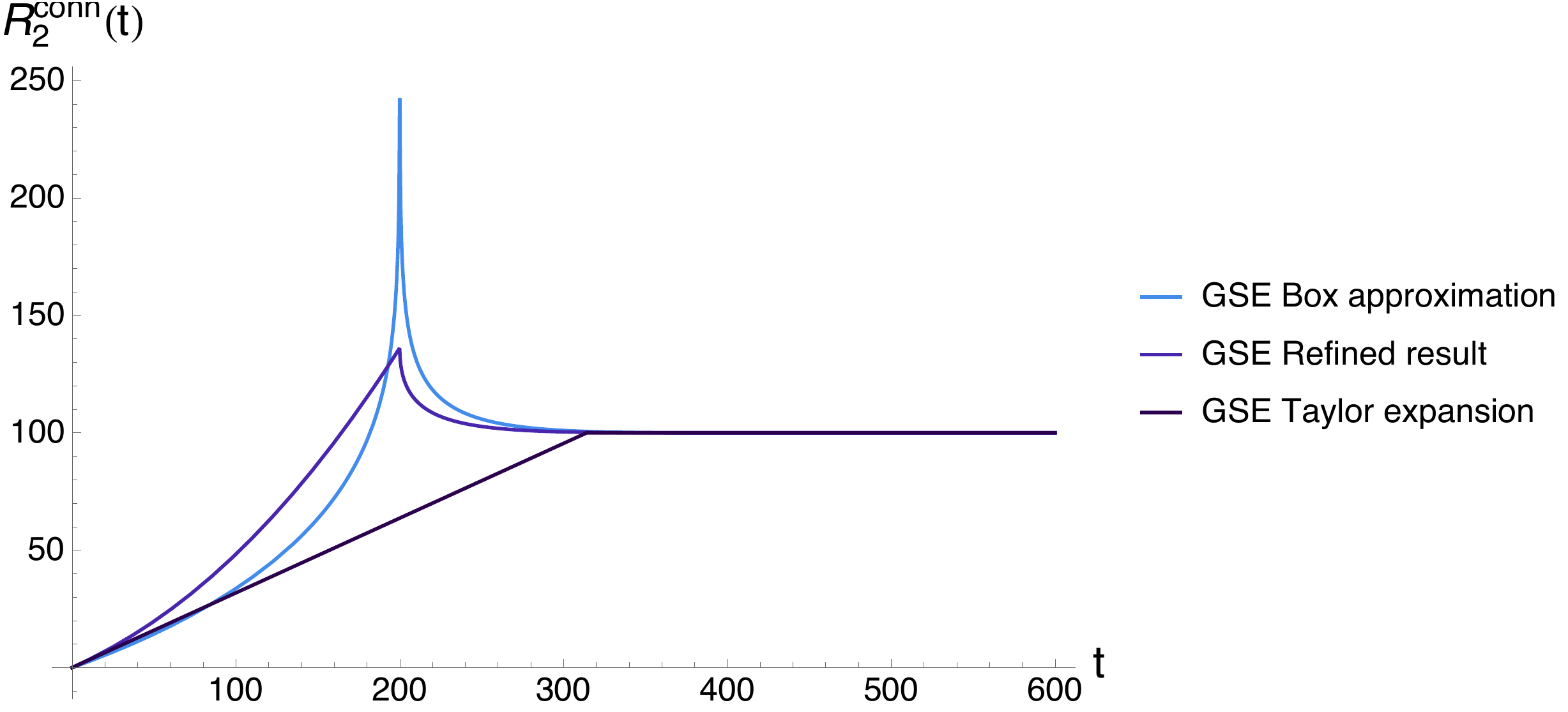}
  \caption{\label{fig8} GOE(up) and GSE(down) connected form factor $\mathcal{R}_2^\text{conn}(t)$ with different approximations in the infinite temperature. We choose $L=100$.}
\end{figure}
\begin{figure}[htbp]
  \centering
  \includegraphics[width=0.8\textwidth]{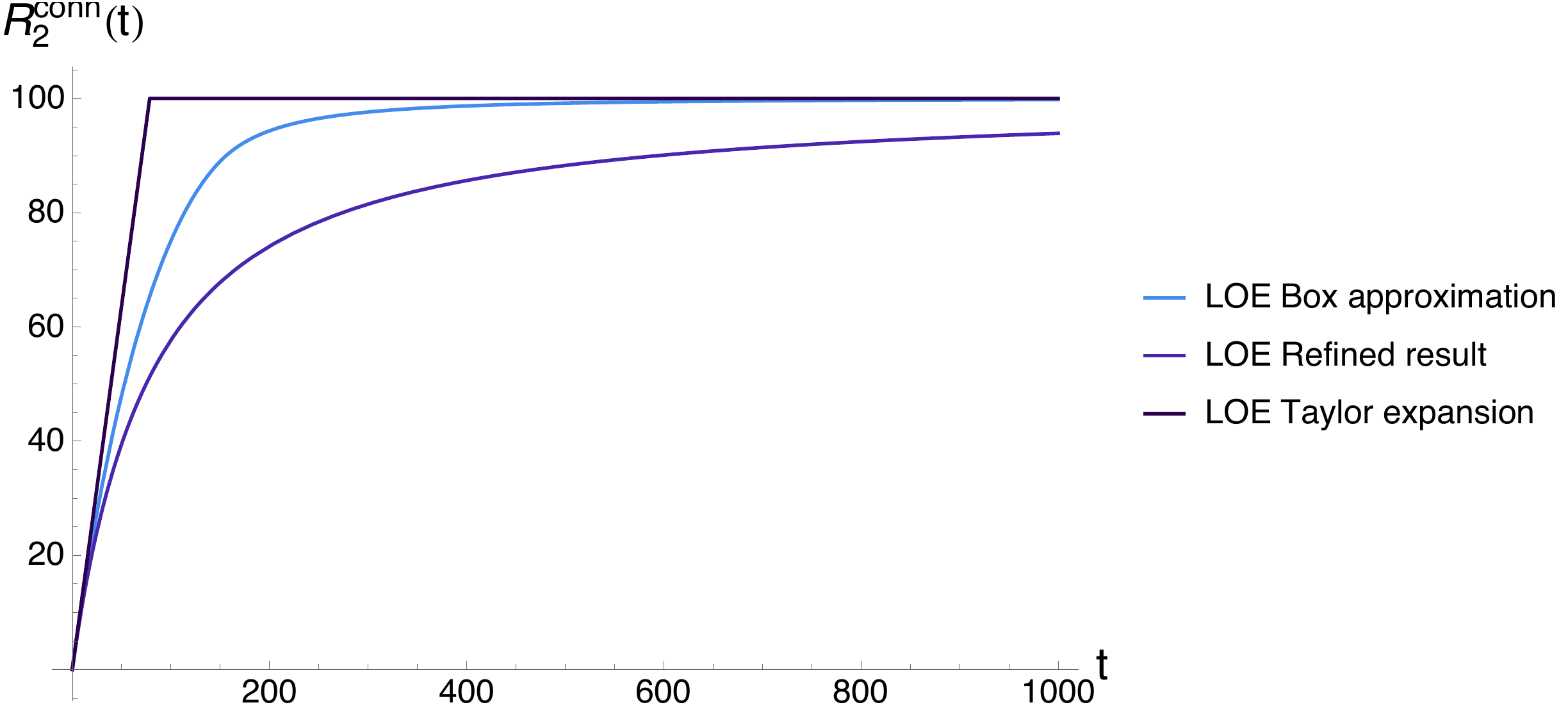}
  \includegraphics[width=0.8\textwidth]{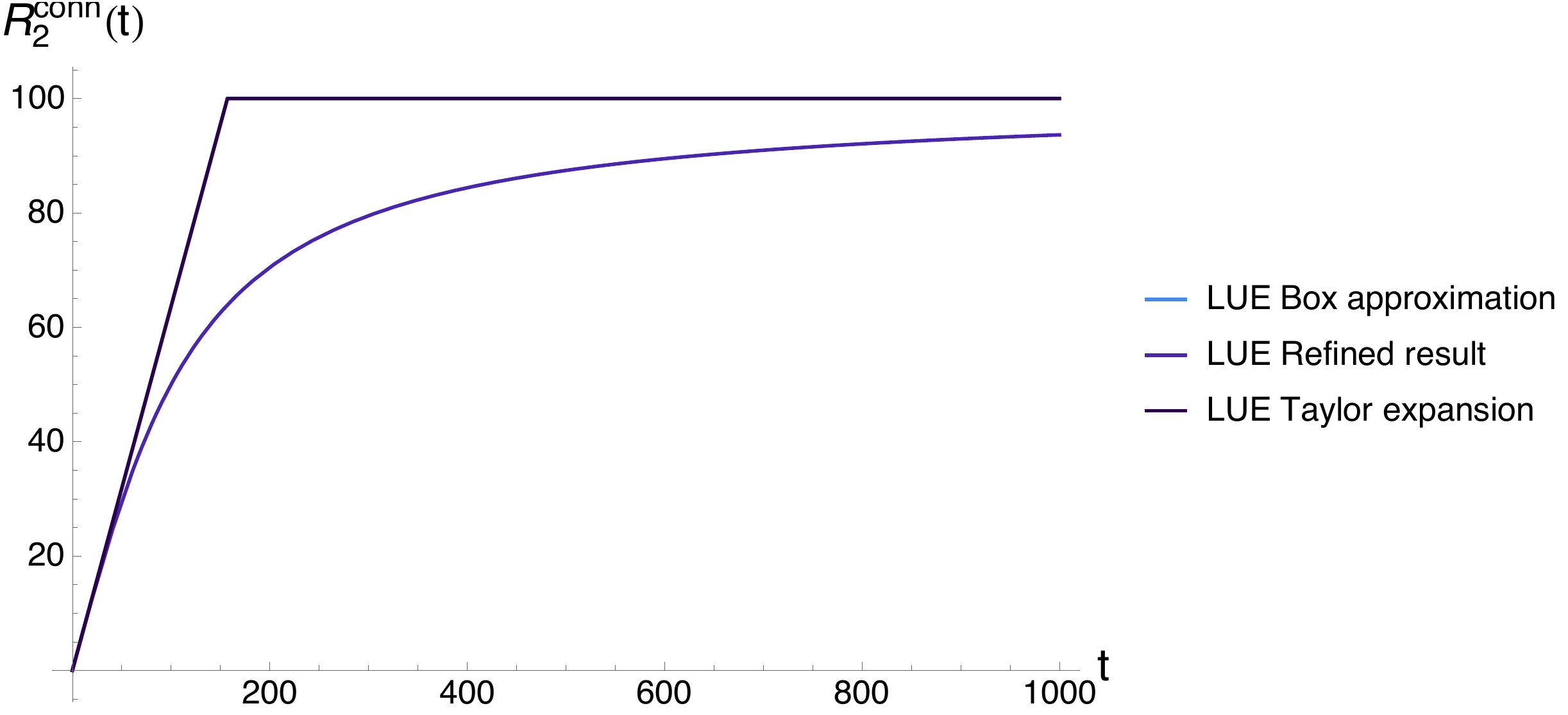}
  \includegraphics[width=0.8\textwidth]{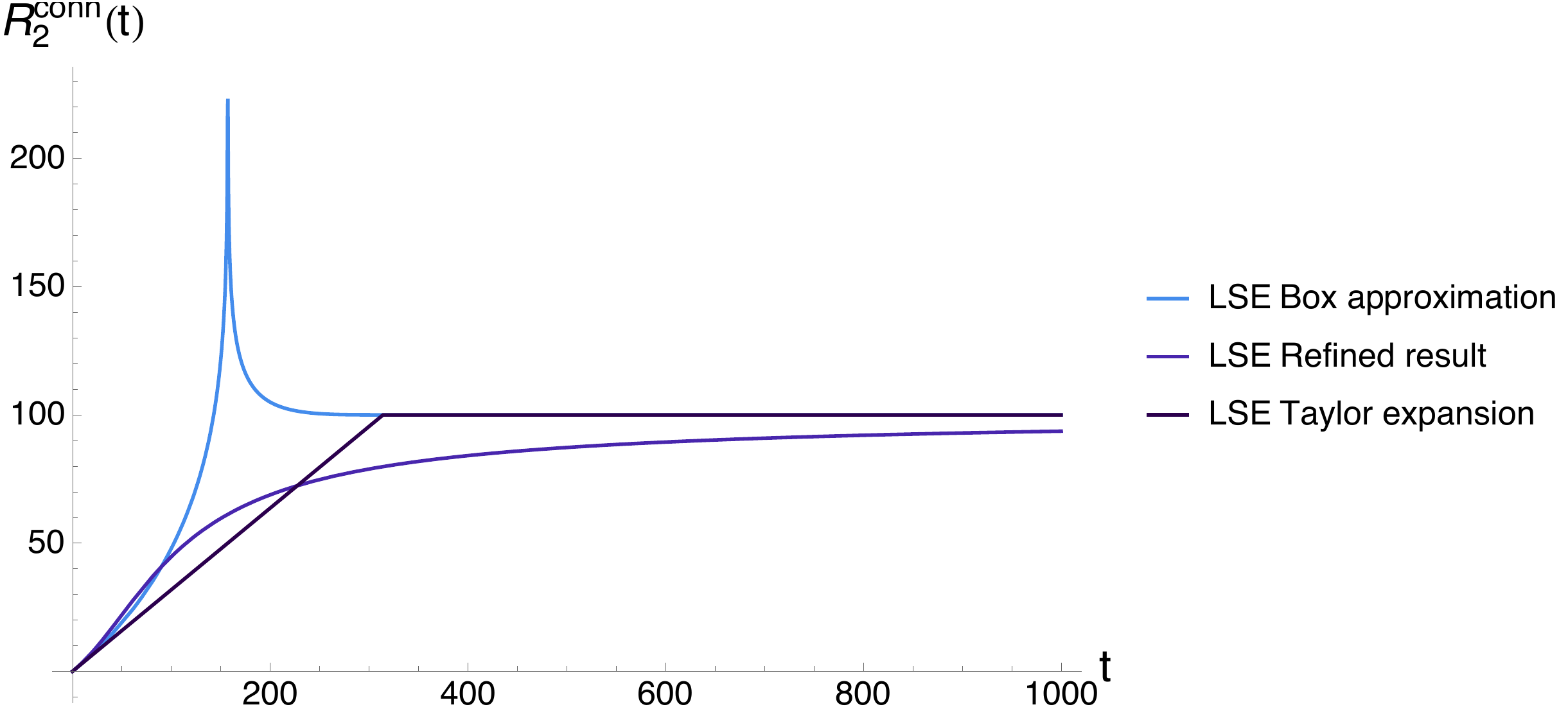}
  \caption{\label{fig11} LOE(up), LUE(middle) and LSE(down) connected form factor $\mathcal{R}_2^\text{conn}(t)$ with different approximations in the infinite temperature. We choose $L=100$. For LUE case by choosing $u$ in the box approximation the Taylor expansion curve and the box approximation curve are the same, so two of three curves are the same for figure in the middle.}
\end{figure}
\\
\\
One can also take a look at the connected finite temperature predictions from the refined kernel. We give them in Figure \ref{fig14} and Figure \ref{fig17} for Gaussian and Wishart-Laguerre ensembles separately. 
\begin{figure}[htbp]
  \centering
  \includegraphics[width=0.8\textwidth]{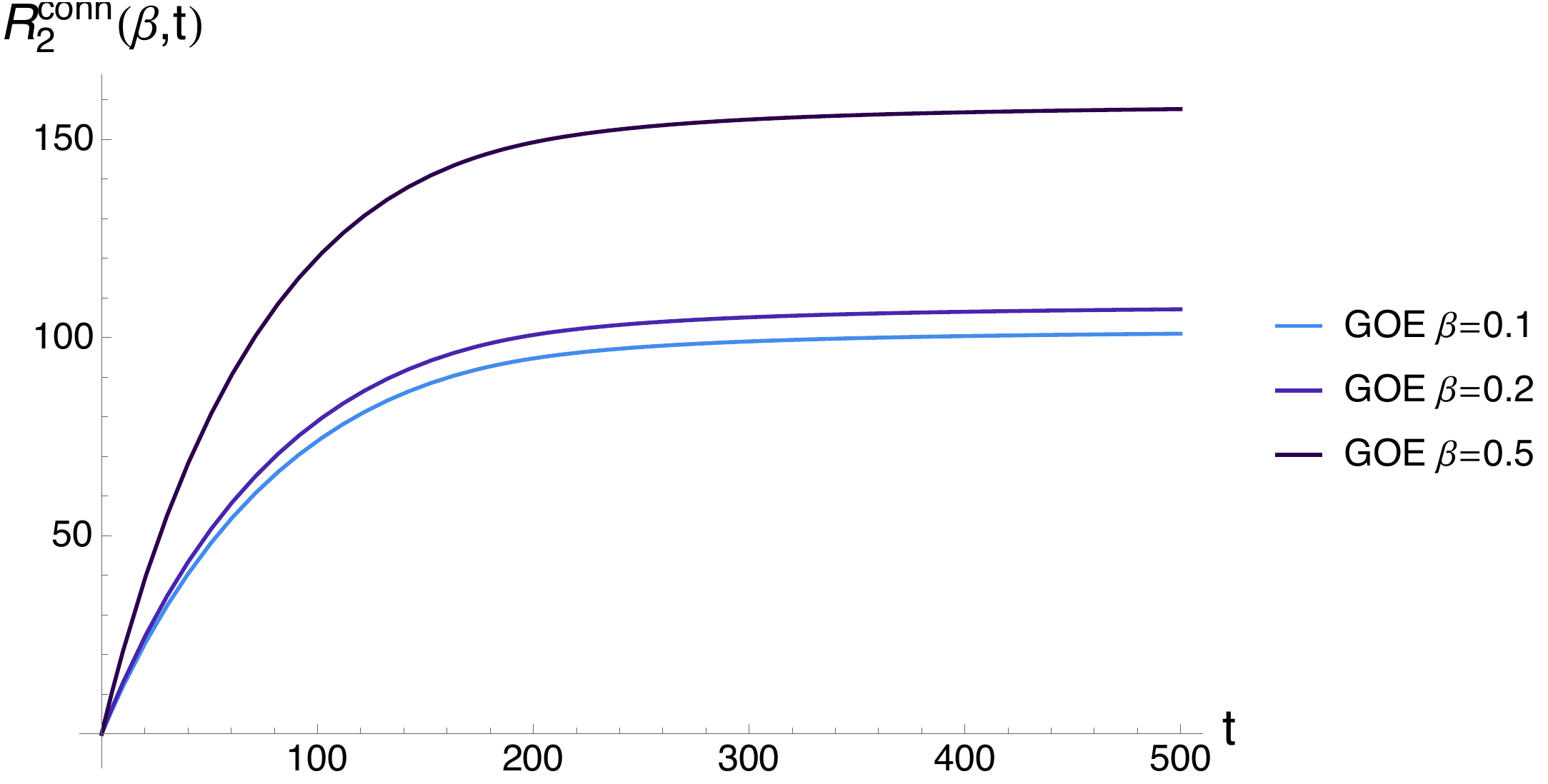}
  \includegraphics[width=0.8\textwidth]{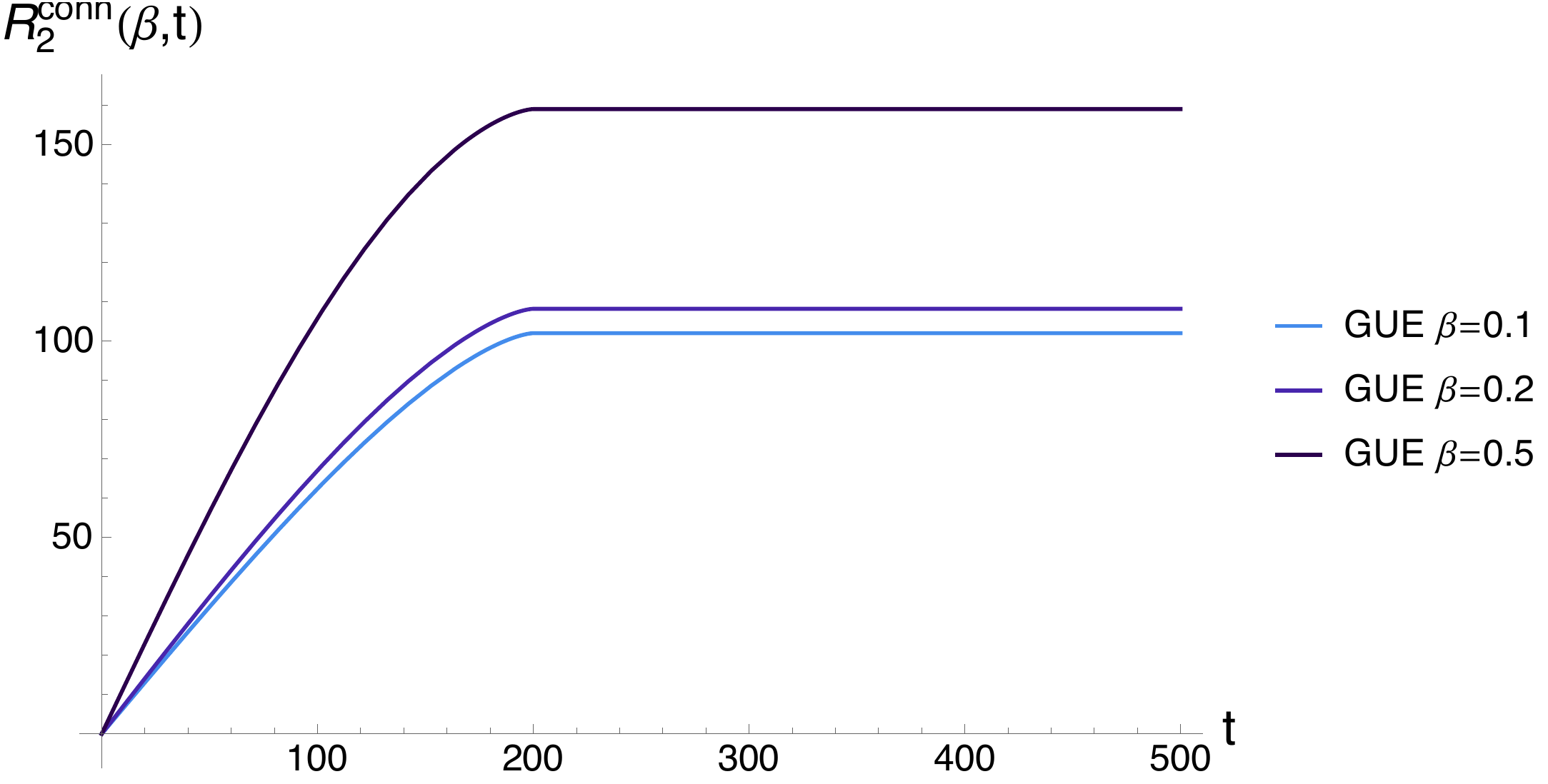}
  \includegraphics[width=0.8\textwidth]{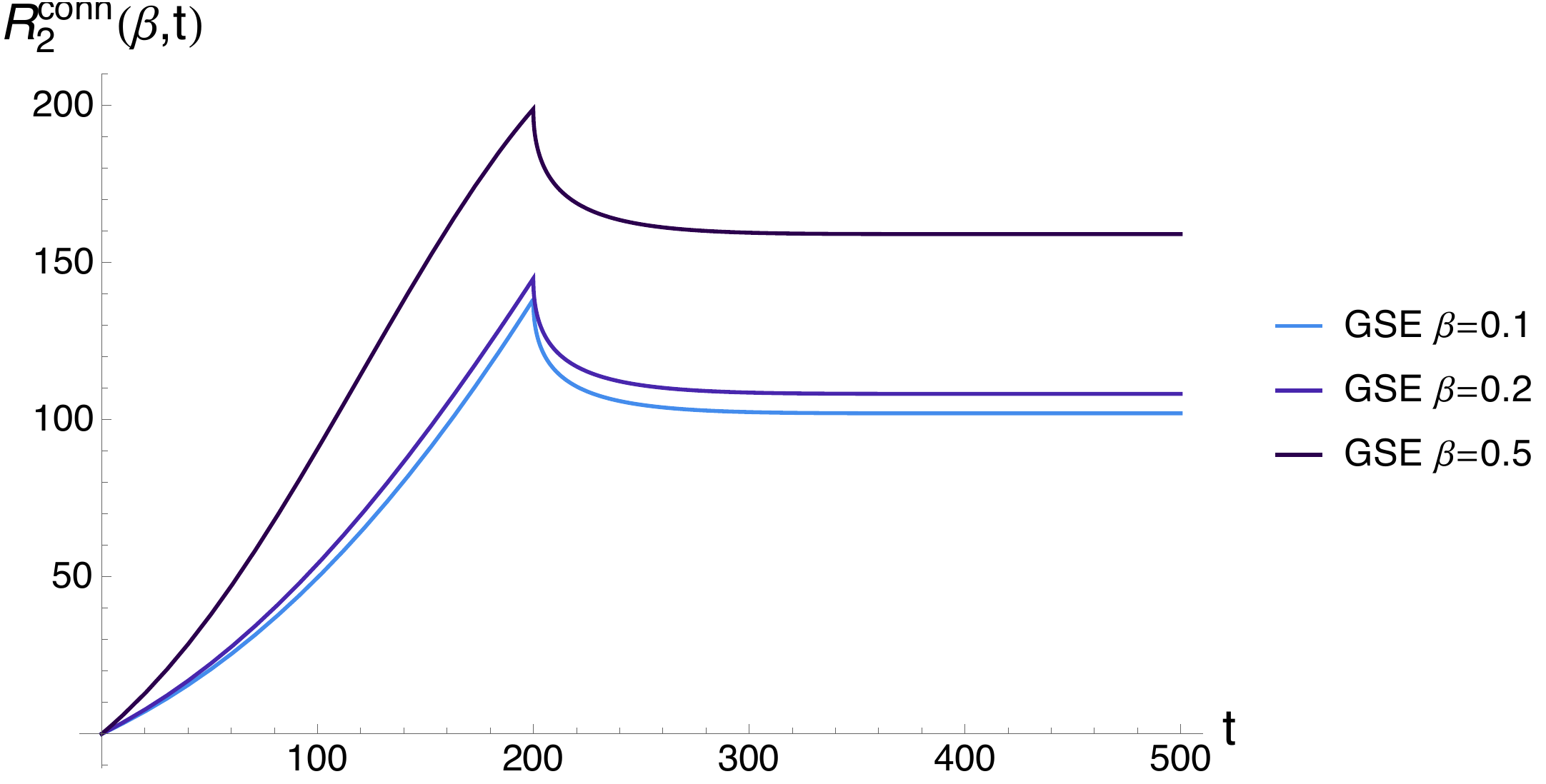}
  \caption{\label{fig14} GOE(up), GUE(middle) and GSE(down) connected form factor $\mathcal{R}_2^\text{conn}(t,\beta)$ for finite temperatures. We choose $L=100$.}
\end{figure}
\begin{figure}[htbp]
  \centering
  \includegraphics[width=0.8\textwidth]{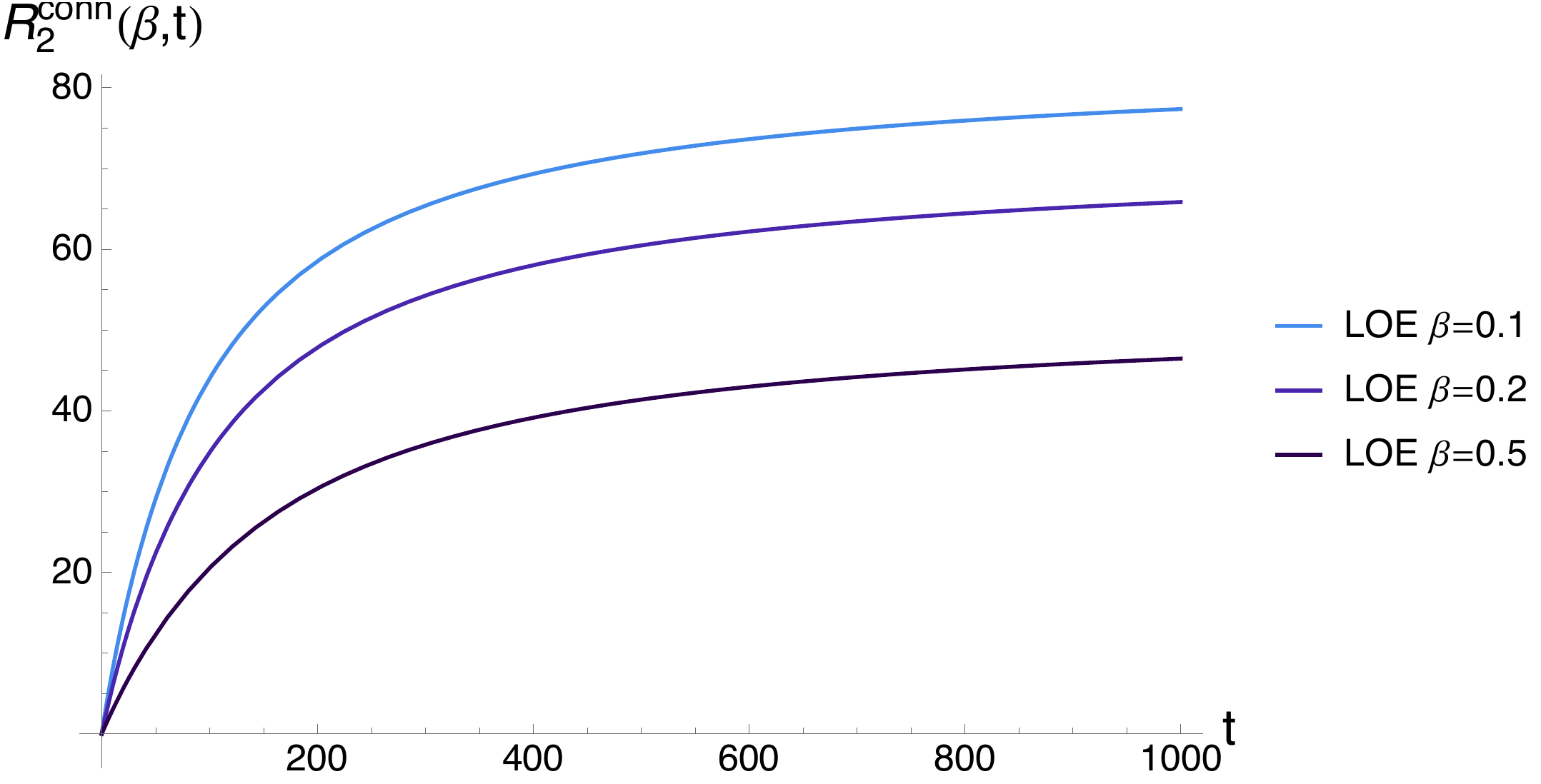}
  \includegraphics[width=0.8\textwidth]{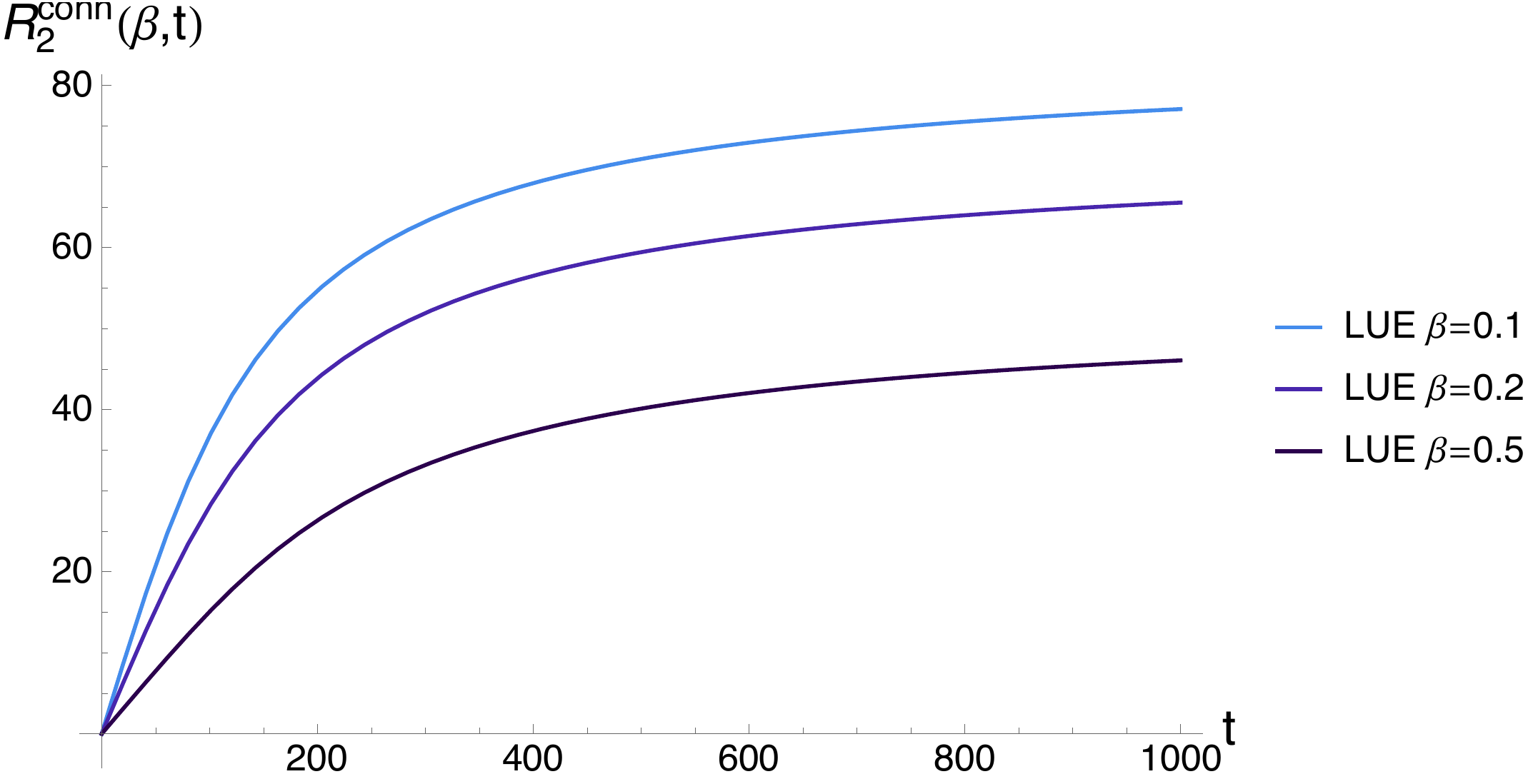}
  \includegraphics[width=0.8\textwidth]{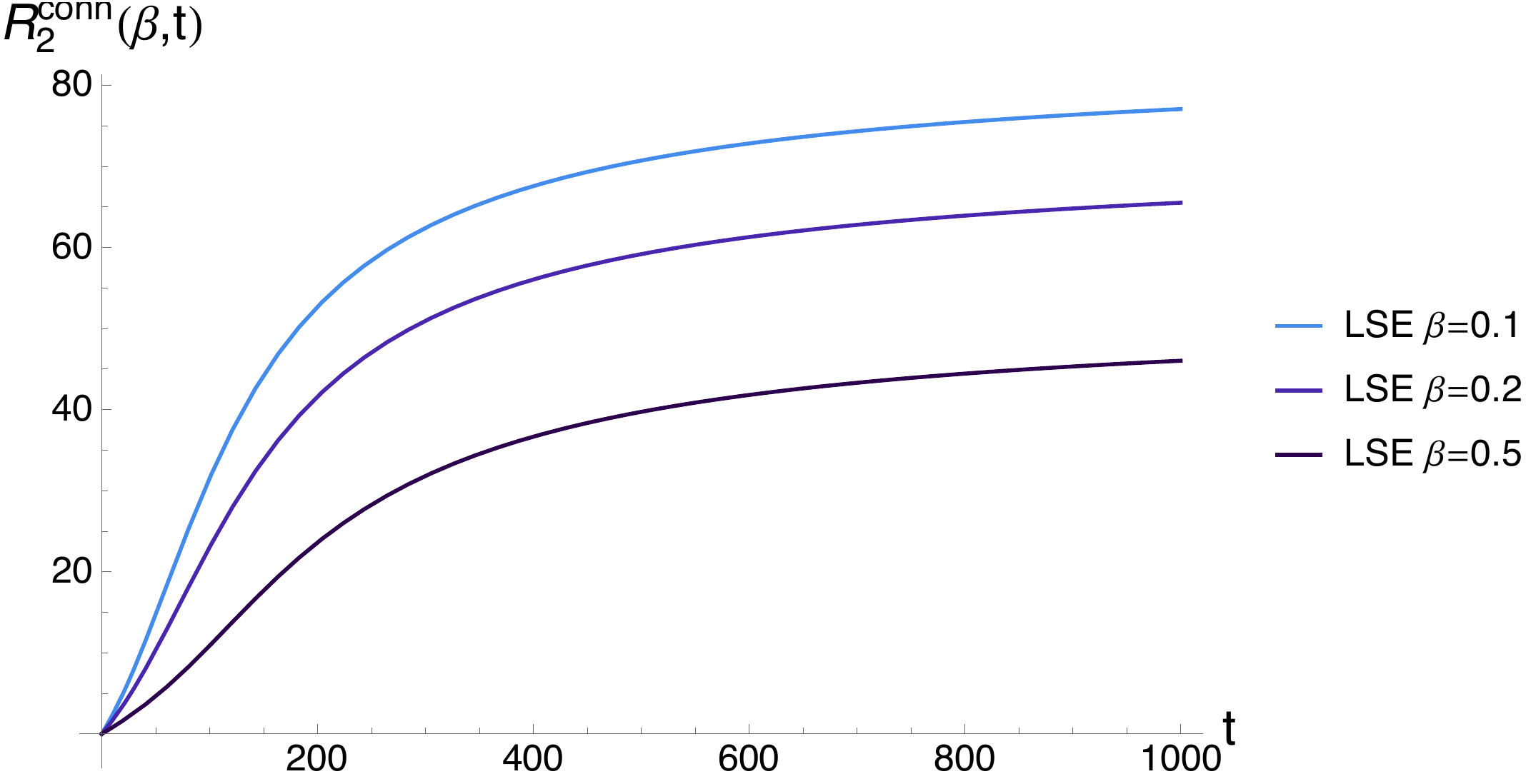}
  \caption{\label{fig17} LOE(up), LUE(middle) and LSE(down) connected form factor $\mathcal{R}_2^\text{conn}(t,\beta)$ for finite temperatures. We choose $L=100$.}
\end{figure}

\section{Applications}\label{Ph}
The spectral form factors of the random matrix theory in the standard ensembles have wide applications in many areas of late time quantum chaos. In this section, we will review and summarize some of the applications with recent interests.
\subsection{SYK-like models and classifications}
One direct application of the random matrix theory form factor results will be matching the qualitative and quantitative behaviors of the spectral form factor of the SYK model. In the majonara SYK model, there exists an eight-fold classification of random matrix theory, with respect to the number of majonara fermions $N$ \cite{You:2016ldz,Garcia-Garcia:2016mno,Cotler:2016fpe}. The classification is $N\text{ mod } 8=0$ for GOE, $N\text{ mod } 8=4$ for GSE, and $N\text{ mod } 8=2,6$ for GUE. The matching is identified for level statistics and the degeneracies.
\\
\\
One can also study the spectral form factor of the theory. One can show that foe the simplest form factor $\mathcal{R}_2$, could also be identified as the combination of the analytic-continued partition function, $\left\langle {Z(\beta  + it)Z(\beta  - it)} \right\rangle  \sim {\mathcal{R}_2}(t,\beta )$. From SYK model, one can read off the dip time, the plateau time and the ramp slope. These quantities could be qualitatively and quantitatively checked by numerical simulations and match with the corresponding random matrix theory prediction \cite{Cotler:2016fpe}. 
\\
\\
One can observe the eight-fold classification of the random matrix theory prediction in the SYK model by observing features of the plots. For instance, one can observe a smooth corner for GOE, a kink for GUE and a sharp peak for GSE around the timescale of the plateau time. These features will show a clear three-fold classification of the SYK model in terms of spectral form factors and could be read off from numerical investigations \cite{Cotler:2016fpe}. 
\\
\\
These ideas could also be generalized to supersymmetric SYK models. In supersymmetrized models, one would expect a disordered supercharge $Q$, and a Hamiltonian $H\sim Q^2$. Thus, if $Q$ is from some Gaussian-like statistics, the result of the squared Gaussian distribution, will be the Wishart-Laguerre-type ensembles. The classification is discussed in \cite{Li:2017hdt,Kanazawa:2017dpd}. For the simplest case ($\mathcal{N}=1$ supersymmetrization, and four-point coupling), the eight-fold classification is modified by $N\text{ mod } 8= 0,6$ for LOE, and $N\text{ mod } 8=2,4$ for LSE. 
\\
\\
An early look of the Wishart-Laguerre ensembles spectral form factor and a connection to supersymmetric SYK model are discussed in \cite{Li:2017hdt,Kanazawa:2017dpd,Hunter-Jones:2017crg}, where the features are clearly different from the Gaussian ensembles. One of the main difference is the early time behavior of the disconnected spectral form factor $\mathcal{R}_2$, which could be obtained from the $r_1(t)$ function we discussed before in these two different ensembles. In Gaussian ensembles we have $r_1(t)\sim 1/t^{3/2}$ while for Wishart--Laguerre ensembles we have $r_1(t)\sim 1/t^{1/2}$. These facts could match with predictions in SYK model, and could be obtained by the one-loop Schwarzian action \cite{Cotler:2016fpe,Stanford:2017thb,Hunter-Jones:2017crg}. Moreover, one can also match the kinky and smoothy behavior around the edge of the plateau from numerics of supersymmetric SYK model \cite{Li:2017hdt}.
\subsection{Out-of-time-ordered correlation functions}
The spectral form factor of the random matrix theory could be related to out-of-time-ordered correlators of the physical models in an interesting way. Here we will discuss the unitary invariance case, where disordered physical models are invariant or nearly invariant under unitary transform. For Gaussian and Wishart-Laguerre disordered models, one may predict them using GUE or LUE.
\\
\\
In this case \cite{Cotler:2017jue,Hunter-Jones:2017crg}, for operator $A$ and $B$, one can compute the two point correlator as 
\begin{align}
\left\langle {A(0)B(t)} \right\rangle  = \left\langle A \right\rangle \left\langle B \right\rangle  + \frac{{{\mathcal{R}_2}(t) - 1}}{{{L^2} - 1}}\left\langle {\left\langle {AB} \right\rangle } \right\rangle 
\end{align}
where 
\begin{align}
\left\langle {\left\langle {AB} \right\rangle } \right\rangle  = \left\langle {AB} \right\rangle  - \left\langle A \right\rangle \left\langle B \right\rangle 
\end{align}
Moreover, if $A$ and $B$ are non-identity Pauli operators, we have 
\begin{align}
\left\langle {A(0)B(t)} \right\rangle  = \left\{ {\begin{array}{*{20}{c}}
{\frac{{{\mathcal{R}_2}(t) - 1}}{{{L^2} - 1}}}&{A = B}\\
0&{A \ne B}
\end{array}} \right.
\end{align}
Thus, if $\mathcal{R}_2(t)\gg 1$, we have 
\begin{align}
\left\langle {A(0)B(t)} \right\rangle  \sim \frac{{{\mathcal{R}_2}(t)}}{{{L^2}}}
\end{align}
Similarly, one can generalize those relations towards four point or higher point functions. For four point function, assuming non-identity Pauli operators $A,B,C,D$ with the relation $ABCD=I$, we have
\begin{align}
\left\langle {A(0)B(t)C(0)D(t)} \right\rangle  \sim \frac{{{\mathcal{R}_4}(t)}}{{{L^4}}}
\end{align}
Another important object in quantum information, which will show the averaged behavior of the out-of-time-ordered correlation function by the \emph{frame potential}. For a given ensemble $\mathcal{E}$, the $k$-th order frame potential is defined by 
\begin{align}
{\cal F}_\mathcal{E}^{(k)} = \int_{U,V \in \mathcal{E}} {dUdV{{\left| {{\rm{Tr(}}U{V^\dag }{\rm{)}}} \right|}^k}} 
\end{align}
One can define $\mathcal{E}$ to be generated by disordered Hamiltonian $H$ with fixed time $t$, $\mathcal{E}_t=\{e^{iHt},H \in \text{disorder ensemble}\}$. So $\mathcal{F}$ is identified as a functional of a disordered system and a fixed time $t$. In \cite{Roberts:2016hpo}, one can show that the frame potential is equal to average of out-of-time-ordered correlators, where the average is over the Pauli group. 
\\
\\
One can compute the relationship between the spectral form factor and the frame potential in the random matrix theory \cite{Cotler:2017jue,Hunter-Jones:2017crg}. For instance, in the two point case we have
\begin{align}
{{\cal F}^{(1)}}(t) = \frac{{\mathcal{R}_2^2(t) + {L^2} - 2{\mathcal{R}_2}(t)}}{{{L^2} - 1}}
\end{align}
One can generate this type of connections to higher point cases and finite temperatures.

\subsection{Page states}
A connection between Wishart-Laguerre ensembles and the Page states is used to be a modified criterion for quantum chaos in terms of wave functions \cite{Chen:2017yzn}. The page state, or alternatively called the random pure state, is defined as the following wavefunction in the Hilbert space $\mathcal{H}=\mathcal{H}_A \otimes \mathcal{H}_B$,
\begin{align}
\left| \psi  \right\rangle  = \sum\limits_{a = 1}^{{N_A}} {\sum\limits_{b = 1}^{{N_B}} {{X_{ab}}\left| {\psi _A^a} \right\rangle \left| {\psi _B^b} \right\rangle } } 
\end{align}
where $X_{ab}$ is the element of the random matrix with the volume $N_A\times N_B$ and one could fix the scaling by normalization condition of the wavefunction. One can assume that this matrix $X$ is Gaussian $N_A\times N_B$ matrix. Thus, the reduced density matrix, when tracing out the system $B$, is given by
\begin{align}
{\rho _A} \sim X{X^\dag }
\end{align}
for subsystem $A$. Now one can consider diagonalization of $\rho_A$ and compute the spectral form factor of it. Because of the squaring structure $XX^\dagger$, the density matrix $\rho_A$ will be a Wishart-Laguerre random matrix (Here $N_A$ and $N_B$ are kept in general, while in our previous computation, we choose the specific case where $N_A=N_B$. When $N_A\ne N_B$, the result will be different but some generic features are similar with the equal case).
\\
\\
This feature will appear in some real chaotic physical systems. In \cite{Chen:2017yzn}, it is claimed that splitting the qubits of the real chaotic system, one would expect that the reduced density matrix, or namely, the entanglement Hamiltonian, will show similar universal spectral correlation, and will match the prediction of Wishart-Laguerre ensemble when considering time evolution. This phenomenon is verified in the context of Floquet system and quantum Ising model.

\section{Conclusion and discussion}\label{conc}
In this paper, we investigate in the very detail, to establish a generic framework on the computational technology of the spectral form factors. We hope that those technologies will give a systematic description of spectral form factors that are used in the field of quantum chaos, and will benefit people studying the connection between random matrix theory and notions of quantum chaos, quantum information and black holes, etc. 
\\
\\
We will highlight some of the points of this paper as the following,
\begin{itemize}
\item Many traditional literatures (for instance \cite{book2}) call the $n$-point spectral form factor as the Fourier transform of $n-1$ eigenvalue variables. To transform the last variable one obtain a delta function in the infinite $L$ limit. For finite but large $L$, one have to invent some regularization technologies. In this paper, we systematically describe the notion of \emph{box approximation}, as a concrete way to realize the cutoff, and apply it to multiple ensembles. And we show in the Gaussian ensemble two point form factor context, approximation beyond the box cutoff must be related to nonlinear physics for the Fourier transform of the sine kernel.
\item We seriously consider how to use the short distance kernel to give a precise prediction for two point form factor with infinite and finite temperatures. Inspired by the treatment from \cite{Cotler:2016fpe,Chen:2017yzn}, we obtain an analytic non-linear connected two point form factor beyond linear approximation in the GUE and LUE, and show the formal and numerical results for the rest cases. 
\item Based on the existing infinite $L$ mathematical algorithms, we illustrate some theorems that could be used to compute higher point form factors for finite but large $L$ for multiple ensembles. We compute four point form factor for them as examples.
\end{itemize}
We hope this research will shed light on the possibilities of the following directions,
\begin{itemize}
\item More general ensembles. Although the situations that are considered are already pretty general, mathematicians and mathematical physicists have a list of more general ensembles. It will be interesting to consider generalized classifications and related ensembles, and compute spectral form factors of them.
\item More physical applications. One may consider applying those form factors to some other chaotic quantum systems and more black hole thought experiments, such as chaotic spin chains or quantum circuits. 
\item Diving deeper into physical meaning of the non-linearity in the connected two point form factor. The breakdown of the naive box cutoff, namely the prediction beyond the linear approximation of the Fourier transform of the sine kernel, might be connected to some physics of thermalization and moreover, holography and gravity \cite{Gharibyan:2018jrp}.
\end{itemize}

\section*{Acknowledgement}
I deeply thank Beni Yoshida for his warm and supportive guidance and Jordan Cotler for nice communications. JL is supported in part by the Institute for Quantum Information and Matter (IQIM), an NSF Physics Frontiers Center (NSF Grant PHY-1125565) with support from the Gordon and Betty Moore Foundation (GBMF-2644), and by the U.S. Department of Energy, Office of Science, Office of High Energy Physics, under Award Number DE-SC0011632.

\end{document}